\documentclass[english,final]{IEEEtran}
 
\makeatletter

\usepackage[T1]{fontenc}
\usepackage[latin9]{inputenc}
\usepackage{amsthm}
\usepackage{amsmath}
\usepackage{amssymb}
\usepackage{graphicx}
\usepackage{dsfont}
\usepackage{cite}
\usepackage{amsmath}
\usepackage{array}
\usepackage{multirow}
\usepackage{caption}
\usepackage{color}

\usepackage{multicol}

\theoremstyle{plain}

\theoremstyle{plain}
\newtheorem{prop}{\protect\propositionname}
\theoremstyle{plain}
\newtheorem{lem}{\protect\lemmaname}
\theoremstyle{plain}
\newtheorem{thm}{\protect\theoremname}
\theoremstyle{plain}
\newtheorem{cor}{\protect\corollaryname}  
\theoremstyle{definition}

\theoremstyle{definition}
\newtheorem{assump}{\protect\assumptionname}
\theoremstyle{definition}
\newtheorem{rem}{\protect\remarkname}

\makeatother
  
\usepackage{babel} 

\providecommand{\claimname}{Claim}
\providecommand{\lemmaname}{Lemma}
\providecommand{\propositionname}{Proposition}
\providecommand{\theoremname}{Theorem}
\providecommand{\corollaryname}{Corollary} 
\providecommand{\definitionname}{Definition}
\providecommand{\assumptionname}{Assumption}
\providecommand{\remarkname}{Remark}

\newcommand{\overbar}[1]{\mkern 1.25mu\overline{\mkern-1.25mu#1\mkern-0.25mu}\mkern 0.25mu}

\newcommand{\openone}{\mathds{1}}

\newcommand{\sbar}{\overbar{s}}

\newcommand{\Sdif}{S_{\mathrm{dif}}}
\newcommand{\Seq}{S_{\mathrm{eq}}}
\newcommand{\sdif}{s_{\mathrm{dif}}}
\newcommand{\seq}{s_{\mathrm{eq}}}
\newcommand{\sdifhat}{\hat{s}_{\mathrm{dif}}}

\newcommand{\Xvdif}{\mathbf{X}_{\overbar{s} \backslash s}}
\newcommand{\Xveq}{\mathbf{X}_{\overbar{s} \cap s}}

\newcommand{\xvdif}{\mathbf{x}_{\overbar{s} \backslash s}}
\newcommand{\xveq}{\mathbf{x}_{\overbar{s} \cap s}}

\newcommand{\dmax}{d_{\mathrm{max}}}

\newcommand{\Zhat}{\hat{Z}}
\newcommand{\Zhatmax}{\hat{Z}_{\mathrm{max}}}

\newcommand{\bmin}{b_{\mathrm{min}}}
\newcommand{\Bmin}{\beta_{\mathrm{min}}}
\newcommand{\bmax}{b_{\mathrm{max}}}
\newcommand{\sigbeta}{\sigma_{\beta}}
\newcommand{\siglv}{\sigma_{\sdif}}
\newcommand{\Tr}{\mathrm{Tr}}

\newcommand{\sign}{\mathrm{sign}}
\newcommand{\SNRdB}{\mathrm{SNR}_{\mathrm{dB}}}

\newcommand{\Weq}{W_{\mathrm{eq}}}
\newcommand{\Wdif}{W_{\mathrm{dif}}}
\newcommand{\weq}{w_{\mathrm{eq}}}
\newcommand{\wdif}{w_{\mathrm{dif}}}
\newcommand{\feq}{f_{\mathrm{eq}}}
\newcommand{\fdif}{f_{\mathrm{dif}}}

\newcommand{\btil}{\widetilde{b}}
\newcommand{\betatil}{\widetilde{\beta}}
\newcommand{\betatileq}{\widetilde{\beta}_{\mathrm{eq}}}
\newcommand{\betatildif}{\widetilde{\beta}_{\mathrm{dif}}}
\newcommand{\btileq}{\widetilde{b}_{\mathrm{eq}}}
\newcommand{\btildif}{\widetilde{b}_{\mathrm{dif}}}

\newcommand{\Sdifhat}{\hat{S}_{\mathrm{dif}}}

\newcommand{\Bcfano}{\mathcal{B}'_{\mathrm{Fano}}}

\newcommand{\Sceq}{\mathcal{S}_{\mathrm{eq}}}
\newcommand{\Scdif}{\mathcal{S}_{\mathrm{dif}}}

\newcommand{\pe}{P_{\mathrm{e}}}

\newcommand{\xv}{\mathbf{x}}

\newcommand{\Xv}{\mathbf{X}}
\newcommand{\yv}{\mathbf{y}}
\newcommand{\Yv}{\mathbf{Y}}

\newcommand{\Zv}{\mathbf{Z}}

\newcommand{\Ac}{\mathcal{A}}
\newcommand{\Bc}{\mathcal{B}}

\newcommand{\Dc}{\mathcal{D}}
\newcommand{\Ec}{\mathcal{E}}

\newcommand{\Sc}{\mathcal{S}}
\newcommand{\Tc}{\mathcal{T}}

\newcommand{\Yc}{\mathcal{Y}}

\newcommand{\EE}{\mathbb{E}}
\newcommand{\PP}{\mathbb{P}}
\newcommand{\RR}{\mathbb{R}}
\newcommand{\ZZ}{\mathbb{Z}}

\newcommand{\var}{\mathrm{Var}}

\newcommand{\Av}{\mathbf{A}}
\newcommand{\Bv}{\mathbf{B}}

\newcommand{\Lc}{\mathcal{L}}

\newcommand{\Iv}{\mathbf{I}}

\newcommand{\Mv}{\mathbf{M}}

\newcommand{\bzero}{\mathbf{0}}

\usepackage{hyperref} 
\usepackage{enumitem}
\usepackage{float}
\providecommand{\tabularnewline}{\\}
\floatstyle{ruled}
\newfloat{algorithm}{tbp}{loa}
\providecommand{\algorithmname}{Algorithm}
\floatname{algorithm}{\protect\algorithmname}

\setenumerate[1]{label=\arabic*.}

\usepackage{geometry} 
\makeatletter
\newcommand{\manuallabel}[2]{\def\@currentlabel{#2}\label{#1}}
\makeatother

\begin{document} 

\title{Limits on Support Recovery with Probabilistic Models: An Information-Theoretic Framework}
\author{Jonathan Scarlett and Volkan Cevher}
\maketitle

\begin{abstract}
The support recovery problem consists of determining a sparse subset of a set of variables that is relevant in generating a set of observations, and arises in a diverse range of settings such as compressive sensing, and subset selection in regression, and group testing.  In this paper, we take a unified approach to support recovery problems, considering general probabilistic models relating a sparse data vector to an observation vector.  We study the information-theoretic limits of both exact and partial support recovery, taking a novel approach motivated by thresholding techniques in channel coding.  We provide general achievability and converse bounds characterizing the trade-off between the error probability and number of measurements, and we specialize these to the linear, 1-bit, and group testing models.  In several cases, our bounds not only provide matching scaling laws in the necessary and sufficient number of measurements, but also sharp thresholds with matching constant factors.  Our approach has several advantages over previous approaches: For the achievability part, we obtain sharp thresholds under broader scalings of the sparsity level and other parameters (e.g., signal-to-noise ratio) compared to several previous works, and for the converse part, we not only provide conditions under which the error probability fails to vanish, but also conditions under which it tends to one.
\end{abstract}

\long\def\symbolfootnote[#1]#2{\begingroup\def\thefootnote{\fnsymbol{footnote}}\footnote[#1]{#2}\endgroup}

\symbolfootnote[0]{ The authors are with the Laboratory for Information and Inference Systems (LIONS), \'Ecole Polytechnique F\'ed\'erale de Lausanne (EPFL) (e-mail: \{jonathan.scarlett,volkan.cevher\}@epfl.ch).

This work was supported in part by the European Commission under Grant ERC Future Proof, SNF 200021-146750 and SNF CRSII2-147633, and EPFL Fellows Horizon2020 grant 665667.

This work was presented in part at the IEEE International Symposium on Information Theory (2015), and at the ACM-SIAM Symposium on Discrete Algorithms (2016).}
\vspace*{-0.5cm}

%
%
\section{Introduction}

The support recovery problem consists of determining a sparse subset of a set of variables that is relevant in producing a set of observations, and arises frequently in disciplines such as group testing \cite{Mal13, Ati12}, compressive sensing (CS) \cite{Fou13}, and subset selection in regression \cite{Mil02}.  The observation models can vary significantly among these disciplines, and it is of considerable interest to consider these in a unified fashion.  This can be done via probabilistic models relating the sparse vector $\beta \in \RR^p$ to a single observation $Y \in \RR$ in the following manner:
\begin{equation}
    (Y|S=s,X=x,\beta=b) \sim P_{Y|X_S\beta_S}(\,\cdot\,|x_s,b_s), \label{eq:intr_model}
\end{equation}
where $S \subseteq \{1,\dotsc,p\}$ represents the set of relevant variables, $X \in \RR^{p}$ is a measurement vector, $X_S$ (respectively, $\beta_S$) is the subvector of $X$ (respectively, $\beta_S$) containing the entries indexed by $S$, and $P_{Y|X_S\beta_S}$ is a given probability distribution.  Given a collection of measurements $\Yv\in\RR^{n}$ and the corresponding measurement matrix $\Xv \in \RR^{n \times p}$ (with each row containing a single measurement vector), the goal is to find the conditions under which the support $S$ can be recovered either \emph{perfectly} or \emph{partially}.  In this paper, we study the information-theoretic limits for this problem, characterizing the number of measurements $n$ required in terms of the sparsity level $k$ and ambient dimension $p$ regardless of the computational complexity.  Such studies are useful for assessing the performance of practical techniques and determining to what extent improvements are possible.

Before proceeding, we state some important examples of models that are captured by \eqref{eq:intr_model}.

\subsubsection*{Linear Model}

The linear model \cite{Wai09,Wai09a} is ubiquitous in signal processing, statistics, and machine learning, and in itself covers an extensive range of applications.  Each observation takes the form
\begin{equation}
    Y  = \langle X, \beta \rangle + Z, \label{eq:intr_linear}
\end{equation}
where $\langle\cdot,\cdot\rangle$ denotes the inner product, and $Z$ is additive noise.  An important quantity in this setting is the signal-to-noise ratio (SNR) $\frac{\EE[\langle X, \beta \rangle^2]}{\EE[Z^2]}$, and in the context of support recovery, the smallest non-zero absolute value $\beta_{\mathrm{min}}$ in $\beta$ has also been shown to play a key role \cite{Wai09,Wan10,Rad11}.

\subsubsection*{Quantized Linear Models}

Quantized variants of the linear model are of significant interest in applications with hardware limitations. An example that we will consider in this paper is the 1-bit model \cite{Bou08}, given by
\begin{equation}
    Y  = \sign\big(\langle X, \beta \rangle + Z\big),
\end{equation}
where the $\sign$ function equals $1$ if its argument is non-negative, and $-1$ if it is negative. 

 \subsubsection*{Group Testing}
 
 Studies of group testing problems began several decades ago \cite{Dor43,Mal78}, and have recently regained significant attention \cite{Ati12,Cha11}, with applications including medical testing, database systems, computational biology, and fault detection.  The goal is to determine a small number of ``defective'' items within a larger subset of items. The items involved in a single test are indicated by $X \in \{0,1\}^p$, and each observation takes the form 
 \begin{equation}
     Y = \openone\bigg\{ \bigcup_{i \in S} \{X_i = 1\} \bigg\} \oplus Z,
 \end{equation}
 with $S$ representing the defective items, $Y$ indicating whether the test contains at least one defective item, and $Z$ representing possible noise (here $\oplus$ denotes modulo-2 addition).  In this setting, one can think of $\beta$ as deterministically having entries equaling one on $S$, and zero on $S^c$. 

The above examples highlight that \eqref{eq:intr_model} captures both discrete and continuous models.  Beyond these examples, several other non-linear models are captured by \eqref{eq:intr_model}, including the logistic, Poisson, and gamma models.  

\subsection{Previous Work and Contributions} \label{sec:PREVIOUS_WORK}

Numerous previous works on the information-theoretic limits of support recovery have focused on the linear model \cite{Wai09,Wan10,Rad11,Fle09,Aks13,Ree12,Ree13,Jin11,Aer10,Tan10,Akc10,Tul13,Sca13e}.  The main aim of these works, and of that the present paper, is to develop necessary and sufficient conditions for which an ``error probability'' vanishes as $p \to \infty$.  However, there are several distinctions that can be made, including:
\begin{itemize}
    \item Random measurement matrices \cite{Wai09,Wan10,Fle09,Rad11} vs.~arbitrary measurement matrices \cite{Aer10,Ree13,Tan10};
    \item Exact support recovery \cite{Wai09,Wan10,Fle09,Rad11} vs.~partial support recovery \cite{Ree12,Ree13,Akc10};
    \item Minimax characterizations for $\beta$ in a given class \cite{Wai09,Wan10,Fle09,Rad11} vs.~average performance bounds for random $\beta$ \cite{Aks13,Tul13,Ree13}.
\end{itemize}
Perhaps the most widely-studied combination of these is that of minimax characterizations for exact support recovery with random measurement matrices.  In this setting, within the class of vectors $\beta$ whose non-zero entries have an absolute value exceeding some threshold $\Bmin$, necessary and sufficient conditions on $n$ are available with matching scaling laws \cite{Wan10,Rad11}.  See also \cite{Wu10,Wu12} for information-theoretic studies of the linear model with a mean square error criterion.

Compared to the linear model, research on the information-theoretic limits of support recovery for non-linear models is relatively scarce. The system model that we have adopted follows those of a line of works seeking mutual information characterizations of sparsity problems \cite{Mal78,Ati12,Aks13,Tan14}, though we make use of significantly different analysis techniques.  Similarly to these works, we focus on random measurement matrices and random non-zero entries of $\beta$.  Other works considering non-linear models have used vastly different approaches such as regularized $M$-estimators \cite{Lee15,Li14} and approximate message passing \cite{Tan14d}.

{\bf High-level Contributions:}  We consider an approach using thresholding techniques akin to those used in information-spectrum methods \cite{Han03}, thus providing a new alternative to previous approaches based on maximum-likelihood decoding and Fano's inequality.  Our key contributions and the advantages of our framework are as follows:
\begin{enumerate}
    \item Considering both exact and partial support recovery, we provide non-asymptotic performance bounds applying to general probabilistic models, along with a procedure for applying them to specific models (\emph{cf.}~Section \ref{sec:APPLYING}).
    \item We explicitly provide the constant factors in our bounds, allowing for more precise characterizations of the performance compared to works focusing on scaling laws (e.g., see \cite{Wai09,Akc10,Rad11}).  In several cases, the resulting necessary and sufficient conditions on the number of measurements coincide up to a multiplicative $1+o(1)$ term, thus providing \emph{exact} asymptotic thresholds (sometimes referred to as \emph{phase transitions} \cite{Ame14,Wu12}) on the number of measurements.
    \item As evidenced in our examples outlined below, our framework often leads to such exact or near-exact thresholds for significantly more general scalings of $k$, SNR, etc.~compared to previous works.
    \item The majority of previous works have developed converse results using Fano's inequality, leading to necessary conditions for $\PP[\mathrm{error}]\to0$.  In contrast, our converse results provide necessary conditions for $\PP[\mathrm{error}]\not\to1$.  The distinction between these two conditions is important from a practical perspective: One may not expect a condition such as $\PP[\mathrm{error}] \ge 10^{-10}$ to be significant, whereas the condition $\PP[\mathrm{error}]\to1$ is inarguably so.
\end{enumerate}

\addtolength{\abovedisplayskip}{-2ex}
\addtolength{\belowdisplayskip}{-2ex}

    \begin{table*}
        \begin{centering}
        \begin{tabular}{|>{\centering}m{1.5cm}|>{\centering}m{0.8cm}|>{\centering}m{2.3cm}|>{\centering}m{1.6cm}|>{\centering}m{4cm}|>{\centering}m{3.8cm}|}
        \hline 
        Model  & Result  & Parameters  & Distributions & Sufficient $n$ for $\PP[\mathrm{error}] \to 0$ & Necessary $n$ for $\PP[\mathrm{error}] \not\to 1$\tabularnewline
        \hline 
           
        \multirow{2}{1.5cm}{\vspace*{-0.7cm}\centering Linear}  & Cor. \ref{cor:linear_genK}  & $k=o(p)$  & Discrete $\beta_{S}$
        Gaussian $\mathbf{X}$ & 
        \[
        \max_{\ell=1,\dotsc k}\frac{(1+\Delta_{\ref*{cor:linear_genK}})\log{p-k \choose \ell}}{f_{\ref*{cor:linear_genK}}(\ell)}
        \]
        
        {\medskip\scriptsize $(\Delta_{\ref*{cor:linear_genK}}\to0$ for various scalings)} & 
        \[
        \max_{\ell=1,\dotsc k}\frac{\log{p-k+\ell \choose \ell}}{f_{\ref*{cor:linear_genK}}(\ell)}
        \]
        \tabularnewline
        \cline{2-6} 
         & Cor. \ref{cor:linear_partial}  & $k\to\infty,k=o(p)$
        
        Partial recovery of proportion $1-\alpha^*$  & Gaussian $\beta_{S}$
        
        Gaussian $\mathbf{X}$ & 
        \[
        \max_{\alpha\in[\alpha^*,1]}\frac{\alpha k\log\frac{p}{k}}{f_{\ref*{cor:linear_partial}}(\alpha)}
        \]
         & 
        \[
        \max_{\alpha\in[\alpha^*,1]}\frac{(\alpha-\alpha^*)k\log\frac{p}{k}}{f_{\ref*{cor:linear_partial}}(\alpha)}
        \]
        \tabularnewline
        \hline 
         & Cor. \ref{cor:1bit_fixedK}  & $k=\Theta(1)$
        
        Low SNR  & Discrete $\beta_{S}$
        
        Gaussian $\mathbf{X}$ & 
        \[
        \max_{\ell=1,\dotsc k}\frac{\ell\log p}{f_{\ref*{cor:1bit_fixedK}}(\ell)}
        \]

        {\scriptsize \medskip (within a factor $\frac{\pi}{2}$ of linear model)} & 
        \[
        \max_{\ell=1,\dotsc k}\frac{\ell\log p}{f_{\ref*{cor:1bit_fixedK}}(\ell)}
        \]
        
        {\scriptsize \medskip (within a factor $\frac{\pi}{2}$ of linear model)}\tabularnewline
        \cline{2-6} 
        1-bit & Cor. \ref{cor:1bit_genK}  & $k=\Theta(p)$
        
        High SNR  & Fixed $\beta_{S}$
        
        Gaussian $\mathbf{X}$ & - & \smallskip $\Omega(p\sqrt{\log p})$
        
        {\scriptsize(compared to $\Theta(p)$ for linear model)}\tabularnewline
        \cline{2-6} 
         & Cor. \ref{cor:1bit_partial}  & $k \to \infty, k=o(p)$
        
        Partial recovery of proportion $1-\alpha^*$ & Gaussian $\beta_{S}$
        
        Gaussian $\mathbf{X}$ & 
        \[
        \max_{\alpha\in[\alpha^*,1]}\frac{\alpha k\log\frac{p}{k}}{f_{\ref*{cor:1bit_partial}}(\alpha)}
        \]
         & 
        \[
        \max_{\alpha\in[\alpha^*,1]}\frac{(\alpha-\alpha^*)k\log\frac{p}{k}}{f_{\ref*{cor:1bit_partial}}(\alpha)}
        \]
        \tabularnewline
        \hline 
        
        \multirow{3}{1.5cm}{ {\vspace*{-1.3cm}\centering Group~testing}}  & Cor. \ref{cor:gt}  & $k=\Theta(p^{\theta})$  & Fixed $\beta_{S}$
        
        Bernoulli $\Xv$ & 
        \[
        \frac{k\log\frac{p}{k}}{f_{\ref*{cor:gt}}(\theta)}
        \]
        
        {\scriptsize ($f_{\ref*{cor:gt}}(\theta) = \log2$ for $\theta \le \frac{1}{3}$)}
        & 
        \[
        \frac{k\log\frac{p}{k}}{\log 2}
        \]
        \tabularnewline
        \cline{2-6}  
        & Cor. \ref{cor:gt_noisy}  & $k=\Theta(p^{\theta})$ 
        
        Noisy (crossover probability $\rho$)  & Fixed $\beta_{S}$
        
        Bernoulli $\mathbf{X}$ & 
        \[
        \frac{k\log\frac{p}{k}}{f_{\ref*{cor:gt_noisy}}(\theta)}
        \]
        
        {\scriptsize ($f_{\ref*{cor:gt_noisy}}(\theta) = \log2 - H_2(\rho)$ for  small $\theta$)}
        & 
        \[
        \frac{k\log\frac{p}{k}}{\log 2 - H_2(\rho)}
        \]
        \tabularnewline
        \cline{2-6} 
        & Cor. \ref{cor:gt_partial}  & $k\to\infty,k=o(p)$  
        
        Partial recovery of proportion $1-\alpha^*$  & Fixed $\beta_{S}$
        
        Bernoulli $\mathbf{X}$ & 
        \[
        \frac{k\log\frac{p}{k}}{\log 2 - H_2(\rho)}
        \]
        & 
        \[
        \frac{(1-\alpha^{*})\big(k\log\frac{p}{k}\big)}{\log 2 - H_2(\rho)}
        \]
        
        \tabularnewline
        \hline     
        
        General discrete observations  & Cor. \ref{cor:gen_finite}  & Arbitrary  & Arbitrary  & - & 
        \[
        \max_{\ell=1,\dotsc k}\frac{\log{p-k+\ell \choose \ell}}{f_{\ref*{cor:gen_finite}}(\ell)+\Delta_{\ref*{cor:gen_finite}}}
        \]
        \tabularnewline
        \hline 
        \end{tabular}
        \par
        \end{centering}
        
        \caption{Overview of main results for exact or partial support recovery under various observation models.  In the necessary and sufficient number of measurements, asymptotically negligible terms have been omitted.  All quantities are defined precisely in Section \ref{sec:EXAMPLES}. } \label{tbl:overview}
    \end{table*}

\addtolength{\abovedisplayskip}{2ex}
\addtolength{\belowdisplayskip}{2ex}

{\bf Contributions for Specific Models:} An overview of our bounds for specific models is given in Table \ref{tbl:overview}, where we state the derived bounds with the asymptotically negligible terms omitted.  All of the models and their parameters are defined precisely in Section \ref{sec:EXAMPLES}; in particular, the functions $f_1,\dotsc,f_{\ref*{cor:gen_finite}}$ and the remainder terms $(\Delta_{\ref*{cor:linear_genK}},\Delta_{\ref*{cor:gen_finite}})$ are given explicitly, and are easy to evaluate.  We proceed by discussing these contributions in more detail, and comparing them to various existing results in the literature:
\begin{enumerate}
    \item \emph{(Linear model)} In the case of exact recovery, we recover the exact thresholds on the required number of measurements given by Jin \emph{et al.} \cite{Jin11}, as well as handling a broader range of scalings of $\beta_{\mathrm{min}} := \min\{|\beta_i| \,:\, \beta_i \ne 0\}$ (see Section \ref{sec:GAUSSIAN_DISC} for details) and strengthening the converse by considering the more stringent condition $\PP[\mathrm{error}] \to 1$.  Our results for partial recovery provide near-matching necessary and sufficient conditions under scalings with $k=o(p)$, thus complementing the extensive study of the scaling $k=\Theta(p)$ by Reeves and Gastpar \cite{Ree12,Ree13}.
    \item \emph{(1-bit model)}  We provide two surprising observations regarding the 1-bit model: Corollary \ref{cor:1bit_fixedK} provides a low-SNR setting where the quantization only increases the asymptotic number of measurements by a factor of $\frac{\pi}{2}$, whereas Corollary \ref{cor:1bit_genK} provides a high-SNR setting where the scaling law is strictly worse than the linear model.  Similar behavior will be observed for partial recovery (Corollaries \ref{cor:linear_partial} and \ref{cor:1bit_partial}) by numerically comparing the bounds for various SNR values.
    \item \emph{(Group testing)} Asymptotic thresholds for group testing with $k=\Theta(1)$ were given previously by Malyutov \cite{Mal78} and Atia and Saligrama \cite{Ati12}.  However, for the case that $k\to\infty$, the sufficient conditions of \cite{Ati12} that introduced additional logarithmic factors.  In contrast, we obtain matching $\Theta\big(k\log\frac{p}{k}\big)$ scaling laws for any sublinear scaling of the form $k=O(p^\theta)$ ($\theta\in(0,1)$).  Moreover, for sufficiently small $\theta$ we obtain exact thresholds.  In particular, for the noiseless setting we show that $n \approx k\log_2\frac{p}{k}$ measurements are both necessary and sufficient for $\theta\le\frac{1}{3}$.  This is in fact the same threshold as that for adaptive group testing \cite{Bal13}, thus proving that non-adaptive Bernoulli measurement matrices are \emph{asymptotically optimal} even when adaptivity is allowed; this was previously known only in the limit as $\theta\to0$ \cite{Ald14}.  For the noisy case, we prove an analogous claim for sufficiently small $\theta$.  A shortened and simplified version of this paper focusing exclusively on group testing can be found in \cite{Sca15b}.
    \item \emph{(General discrete observations)} Our converse for the case of general discrete observations (Corollary \ref{cor:gen_finite}) recovers  that of Tan and Atia \cite{Tan14} for the case that $\beta_S$ is fixed, strengthens it due to a smaller remainder term $\Delta_{\ref*{cor:gen_finite}}$, and provides a generalization to the case that $\beta_S$ is random.
\end{enumerate}

\subsection{Structure of the Paper}

In Section \ref{sec:SETUP}, we introduce our system model. In Section \ref{sec:GENERAL}, we present our main non-asymptotic achievability and converse results for general observation models, and the procedure for applying them to specific problems.  Several applications of our results to specific models are presented in Section \ref{sec:EXAMPLES}.  The proofs of the general bounds are given in Section \ref{sec:PROOFS}, and conclusions are drawn in Section \ref{sec:CONCLUSION}.

\subsection{Notation} \label{sec:NOTATION}

We use upper-case letters for random variables, and lower-case variables for their realizations.  A non-bold character may be a scalar or a vector, whereas a bold character refers to a collection of $n$ scalars (e.g., $\Yv \in \RR^n$) or vectors (e.g., $\Xv \in \RR^{n\times p}$).  We write $\beta_{S}$ to denote the subvector of $\beta$ at the columns indexed by $S$, and $\Xv_{S}$ to denote the submatrix of $\Xv$ containing the columns indexed by $S$.  The complement with respect to $\{1,\dotsc,p\}$ is denoted by $(\cdot)^c$.

The symbol $\sim$ means ``distributed as''.  For a given joint distribution $P_{XY}$, the corresponding marginal distributions are denoted by $P_{X}$ and $P_{Y}$, and similarly for conditional marginals (e.g., $P_{Y|X}$).  We write $\PP[\cdot]$ for probabilities, $\EE[\cdot]$ for expectations, and $\var[\cdot]$ for variances.  We use usual notations for the entropy (e.g., $H(X)$) and mutual information (e.g., $I(X;Y)$), and their conditional counterparts (e.g., $H(X|Z)$, $I(X;Y|Z)$).  Note that $H$ may also denote the differential entropy for continuous random variables; the distinction will be clear from the context.  We define the binary entropy function $H_2(\rho) := -\rho\log\rho - (1-\rho)\log(1-\rho)$, and the Q-function $Q(x) := \PP[W \ge x]$ ($W \sim N(0,1)$).

We make use of the standard asymptotic notations $O(\cdot)$, $o(\cdot)$, $\Theta(\cdot)$, $\Omega(\cdot)$ and $\omega(\cdot)$.  We define the function $[\cdot]^+ = \max\{0,\cdot\}$, and write the floor function as $\lfloor\cdot\rfloor$.  The function $\log$ has base $e$.

\section{Problem Setup} \label{sec:SETUP}

\subsection{Model and Assumptions}

Recall that $p$ denotes the ambient dimension, $k$ denotes the sparsity level, and $n$ denotes the number of measurements.  We let $\Sc$ be the set of subsets of $\{1,\dotsc,p\}$ having cardinality $k$. The key random variables in our setup are the support set $S \in \Sc$, the data vector $\beta\in\RR^p$, the measurement matrix $\Xv\in\RR^{n \times p}$, and the observation vector $\Yv \in \RR^{n}$.\footnote{Extensions to more general alphabets beyond $\RR$ are straightforward.}

The support set $S$ is assumed to be equiprobable on the $p \choose k$ subsets within $\Sc$.  Given $S$, the entries of $\beta_{S^c}$ are deterministically set to zero, and the remaining entries are generated according to some distribution $\beta_S \sim P_{\beta_S}$.  We assume that these non-zero entries follow the same distribution for all of the $p \choose k$ possible realizations of $S$, and that this distribution is permutation-invariant.

The measurement matrix $\Xv$ is assumed to have i.i.d.~values on some distribution $P_{X}$.  We write $P_{X}^{n \times p}$, to denote the corresponding i.i.d.~distributions for matrices, and we write $P_{X}^k$ as a shorthand for $P_{X}^{k \times 1}$.  Given $S$, $\Xv$, and $\beta$, each entry of the observation vector $\Yv$ is generated in a conditionally independent manner, with the $i$-th entry $Y^{(i)}$ distributed according to
\begin{equation}
    (Y^{(i)}|S=s,X^{(i)}=x^{(i)},\beta=b) \sim P_{Y|X_S\beta_S}(\,\cdot\,|x^{(i)}_s,b_s), \label{eq:single_output}
\end{equation}
for some conditional distribution $P_{Y|X_{S}\beta_S}$.  We again assume symmetry with respect to $S$, namely, that $P_{Y|X_{S}\beta_S}$ does not depend on the specific realization, and that the distribution is invariant when the columns of $X_S$ and the entries of $\beta_S$ undergo a common permutation.

Given $\Xv$ and $\Yv$, a decoder forms an estimate $\hat{S}$ of $S$.  Similarly to previous works studying information-theoretic limits on support recovery, we assume that the decoder knows the system model.  We consider two related performance measures.  In the case of exact support recovery, the error probability is given by 
\begin{equation}
     \pe := \PP[\hat{S} \ne S], \label{eq:pe}
\end{equation}
and is taken with respect to the realizations of $S$, $\beta$, $\Xv$, and $\Yv$; the decoder is assumed to be deterministic.  We also consider a less stringent performance criterion requiring that only $k - \dmax$ entries of $S$ are successfully recovered, for some $\dmax \in \{1,\dotsc,k-1\}$.  Following \cite{Ree12,Ree13}, the error probability is given by
\begin{equation}
    \pe(\dmax) := \PP\big[|S \backslash \hat{S}| > \dmax \cup |\hat{S} \backslash S| > \dmax\big]. \label{eq:pe_partial}
\end{equation}
Note that if both $S$ and $\hat{S}$ have cardinality $k$ with probability one, then the two events in the union are identical, and hence either of the two can be removed.

For clarity, we formally state our main assumptions as follows:
{ \setenumerate[1]{label=[A\arabic*]}
\begin{enumerate}
    \item The support set $S$ is uniform on the ${p \choose k}$ subsets of $\{1,\dotsc,p\}$ of size $k$, and the measurement matrix $\Xv$ is i.i.d.~on some distribution $P_X$.
    \item The non-zero entries $\beta_S$ are distributed according to $P_{\beta_S}$, and this distribution is permutation-invariant and the same for all realizations of $S$.
    \item The observation vector $\Yv$ is conditionally i.i.d.~according to $P_{Y|X_S\beta_S}$, and this distribution is the same for all realizations of $S$, and invariant to common permutations of the columns of $X_S$ and entries of $\beta_S$.
    \item The decoder is given $(\Xv,\Yv)$, and also knows the system model including $k$, $P_{Y|X_S\beta_S}$, and $P_{\beta_S}$.
\end{enumerate} }

Our main goal is to derive necessary and sufficient conditions on $n$ and $k$ (as functions of $p$) such that $\pe$ or $\pe(\dmax)$ vanishes as $p\to\infty$.  Moreover, when considering converse results, we will not only be interested in conditions under which $\pe \not\to 0$, but also conditions under which the stronger statement $\pe\to1$ holds.  

In particular, we introduce the terminology that the \emph{strong converse} holds if there exists a sequence of values $n^{*}$, indexed by $p$, such that for all $\eta > 0$, we have $\pe\to0$ when $n \ge n^{*} (1+\eta)$, and $\pe\to1$ when $n \le n^{*} (1-\eta)$.  This is related to the notion of a \emph{phase transition} \cite{Ame14,Wu12}.  More generally, we will refer to conditions under which $\pe\to1$ as \emph{strong impossibility results}, not necessarily requiring matching achievability bounds.  That is, the strong converse conclusively gives a sharp threshold between failure and success, whereas a strong impossibility result may not.

It will prove convenient to work with random variables that are implicitly conditioned on a fixed value of $S$, say $s=\{1,\dotsc,k\}$.  We write $P_{\beta_s}$ and $P_{Y|X_{s}\beta_s}$ in place of $P_{\beta_S}$ and $P_{Y|X_{S}\beta_S}$ to emphasize that $S=s$.  Moreover, we define the corresponding joint distribution
\begin{equation}
    P_{\beta_s X_s Y}(b_s,x_s,y) := P_{\beta_s}(b_s)P_{X}^{k}(x_s)P_{Y|X_s\beta_s}(y|x_s,b_s), \label{eq:distr_sl}
\end{equation}
and its multiple-observation counterpart
\begin{equation}
    P_{\beta_s\Xv_s\Yv}(b_s,\xv_s,\yv) := P_{\beta_s}(b_s)P_X^{n \times k}(\xv_s) P^{n}_{Y|X_s\beta_s}(\yv|\xv_{s},b_s). \label{eq:vector_distr} 
\end{equation}
where $P^{n}_{Y|X_s\beta_s}(\cdot|\cdot,b_s)$ is the $n$-fold product of $P_{Y|X_s\beta_s}(\cdot|\cdot,b_s)$.  

Except where stated otherwise, the random variables $(\beta_s,X_s,Y)$ and $(\beta_s,\Xv_s,\Yv)$ appearing throughout this paper are distributed as
\begin{align}
    (\beta_s,X_s,Y) &\sim  P_{\beta_s X_s Y} \label{eq:joint_dist} \\
    (\beta_s,\Xv_s,\Yv) &\sim P_{\beta_s\Xv_s\Yv}, \label{eq:joint_dist_n}
\end{align}
with the remaining entries of the measurement matrix being distributed as $\Xv_{s^c} \sim P_X^{n \times (p-k)}$, and with $\beta_{s^c}=0$ deterministically.  That is, we condition on a fixed $S=s$ except where stated otherwise.

For notational convenience, the main parts of our analysis are presented with $P_{\beta_s}$, $P_X$ and $P_{Y|X_s\beta_s}$ representing probability mass functions (PMFs), and with the corresponding averages written using summations.  However, except where stated otherwise, our analysis is directly applicable to case that these distributions instead represent probability density functions (PDFs), with the summations replaced by integrals where necessary.  The same applies to mixed discrete-continuous distributions.

\subsection{Information-Theoretic Definitions}

Before introducing the required definitions for support recovery, it is instructive to discuss thresholding techniques in channel coding studies.  These commenced in early works such as \cite{Fei54,Sha57}, and have recently been used extensively in information-spectrum methods \cite{Ver94, Han03}.

\subsubsection{Channel Coding}

We first recall the mutual information, which is ubiquitous in information theory:
\begin{equation}
    I(X;Y) := \sum_{x,y} P_{XY}(x,y)\log\frac{P_{Y|X}(y|x)}{P_{Y}(y)}.
\end{equation}
In deriving asymptotic and non-asymptotic performance bounds, it is common to work directly with the logarithm,
\begin{equation}
    \imath(x;y) := \log\frac{P_{Y|X}(y|x)}{P_{Y}(y)},
\end{equation}
which is commonly known as the \emph{information density}.  The thresholding techniques work by manipulating probabilities of events of the form $\sum_{i=1}^{n} \imath(X_i;Y_i) \le \gamma$ and $\sum_{i=1}^{n} \imath(X_i;Y_i) > \gamma$.  For the former, one can perform a \emph{change of measure} from the conditional distribution $\Yv$ given $\Xv$ to the unconditional distribution of $\Yv$, with a multiplicative constant $e^{-\gamma}$.  For the latter, one can similarly perform a change of measure from $\Yv$ to $(\Yv|\Xv)$.  Hence, in both cases, there is a simple relation between the conditional and unconditional probabilities of the output sequences.

Using these methods, one can get upper and lower bounds on the error probability such that the dominant term is
\begin{equation}
    \PP\bigg[ \frac{1}{n}\sum_{i=1}^{n} \imath(X_i;Y_i) \le I(X;Y) + \zeta_n \bigg]
\end{equation}
for some $\zeta_n = o(1)$.  Assuming that $\{(X_i,Y_i)\}_{i=1}^n$ has some form of i.i.d.~structure, one can analyze this expression using tools from probability theory.  The law of large numbers yields the channel capacity $C=\max_{P_X}I(X;Y)$, and refined characterizations can be obtained using variations of the central limit theorem \cite{Pol10}.


Among the channel coding literature, our analysis is most similar to that of mixed channels \cite[Sec.~3.3]{Han03}, where the relation between the input and output sequences is not i.i.d., but instead conditionally i.i.d.~given another random variable.  In our setting, $\beta_s$ will play the role of this random variable.  See Figure \ref{fig:channel_diagram} for a depiction of this connection.

\begin{figure}
    \begin{centering}
        \includegraphics[width=1\columnwidth]{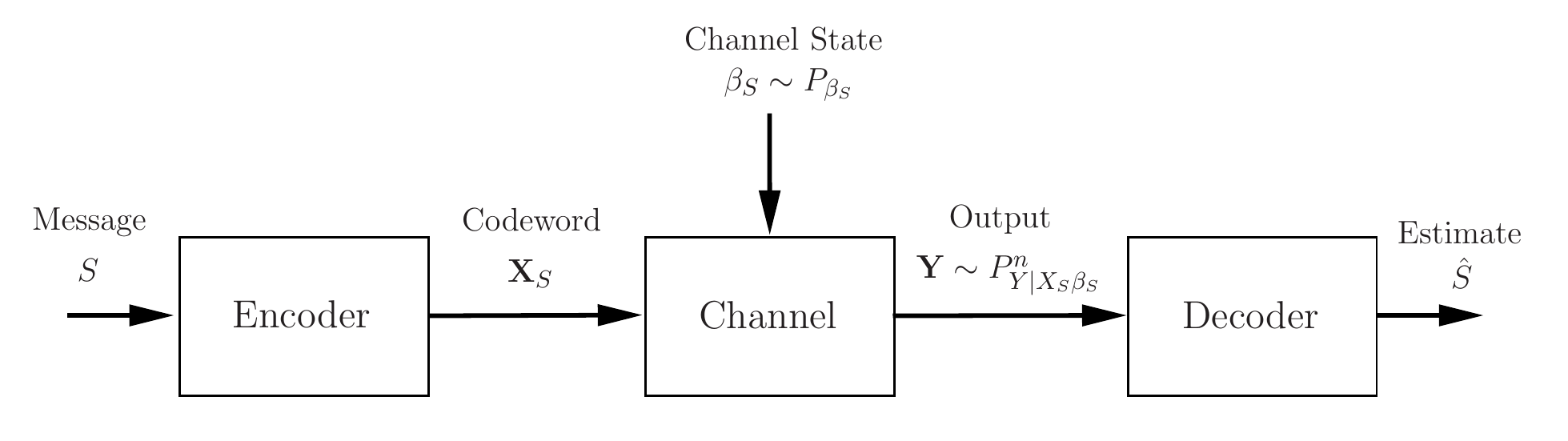}
        \par
    \end{centering}
 
    \caption{ Connection between support recovery and coding over a mixed channel. \label{fig:channel_diagram}}
\end{figure}

\subsubsection{Support Recovery}

As in \cite{Ati12,Aks13}, we will consider partitions of the support set $s \in \Sc$ into two sets $\sdif \ne \emptyset$ and $\seq$.  As will be seen in the proofs, $\seq$ will typically correspond to an overlap between $s$ and some other set $\sbar$ (i.e., $s \cap \sbar$), whereas $\sdif$ will correspond to the indices in one set but not the other (e.g., $s \backslash \sbar$).  There are $2^k - 1$ ways of performing such a partition with $\sdif \ne \emptyset$.

For fixed $s\in\Sc$ and a corresponding pair $(\sdif,\seq)$, we introduce the notation
\begin{align}
    P_{Y|\Xv_{\sdif}\Xv_{\seq}}(\yv|\xv_{\sdif},\xv_{\seq}) &:= P_{\Yv|\Xv_s}(\yv|\xv_s) \label{eq:p_split1} \\
    P_{Y|X_{\sdif}X_{\seq}\beta_s}(y|x_{\sdif},x_{\seq},b_s) &:= P_{Y|X_s\beta_s}(y|x_s,b_s), \label{eq:p_split2}
\end{align}
where $P_{\Yv|\Xv_s}$ is the marginal distribution of \eqref{eq:vector_distr}.  While the left-hand sides of \eqref{eq:p_split1}--\eqref{eq:p_split2} represent the same quantities for any such $(\sdif,\seq)$, it will still prove convenient to work with these in place of the right-hand sides.  In particular, this allows us to introduce the marginal distributions
\begin{align}
    & P_{\Yv|\Xv_{\seq}}(\yv|\xv_{\seq}) \nonumber \\
        &:= \sum_{\xv_{\sdif}} P_X^{n \times \ell}(\xv_{\sdif})P_{Y|\Xv_{\sdif}\Xv_{\seq}}(\yv|\xv_{\sdif},\xv_{\seq}) \\
    & P_{Y|X_{\seq}\beta_s}(y|x_{\seq},b_s) \nonumber \\
        &:= \sum_{x_{\sdif}} P_X^{\ell}(x_{\sdif})P_{Y|X_{\sdif}X_{\seq}\beta_s}(y|x_{\sdif},x_{\seq},b_s),
\end{align}
where $\ell := |\sdif|$.  Using the preceding definitions, we introduce two information densities.  The first contains probabilities averaged over $\beta_s$,
\begin{equation}
    \imath(\xv_{\sdif}; \yv | \xv_{\seq} ) := \log\frac{ P_{\Yv|\Xv_{\sdif}\Xv_{\seq}}(\yv|\xv_{\sdif},\xv_{\seq}) }{ P_{\Yv|\Xv_{\seq}}(\yv|\xv_{\seq}) }, \label{eq:idens_multi}
\end{equation}
whereas the second conditions on $\beta_s=b_s$:
\begin{equation}
\imath^n(\xv_{\sdif}; \yv | \xv_{\seq}, b_s ) := \sum_{i=1}^{n} \imath(x^{(i)}_{\sdif}; y^{(i)} | x^{(i)}_{\seq}, b_s), \label{eq:idens_bn}
\end{equation}
where the single-letter information density is  
\begin{equation}
	\imath(x_{\sdif}; y | x_{\seq}, b_s) := \log\frac{ P_{Y|X_{\sdif}X_{\seq}\beta_s}(y|x_{\sdif},x_{\seq},b_s) }{ P_{Y|X_{\seq}\beta_s}(y|x_{\seq},b_s) }. \label{eq:idens_b}
\end{equation}
As mentioned above, we will generally work with discrete random variables for clarity of exposition, in which case the ratio is between two PMFs.  In the case of continuous observations the ratio is instead between two PDFs, and more generally this can be replaced by the Radon-Nikodym derivative as in the channel coding setting \cite{Pol10}.

Averaging \eqref{eq:idens_b} with respect to the random variables in \eqref{eq:joint_dist} conditioned on $\beta_s = b_s$ yields a conditional mutual information, which we denote by
\begin{equation}
	 I_{\sdif,\seq}(b_s) := I(X_{\sdif};Y|X_{\seq},\beta_s=b_s). \label{eq:condMI}
\end{equation}
This quantity will play a key role in our bounds, which will typically have the form
\begin{equation}
    \pe \approx \PP\bigg[ n \le \max_{(\sdif,\seq)}\frac{{p \choose |\sdif|}}{ I_{\sdif,\seq}(\beta_s) } \bigg],
\end{equation}
as will be made more precise in the subsequent sections.

%
%
\section{General Achievability and Converse Bounds} \label{sec:GENERAL}

In this section, we provide general results holding for arbitrary models satisfying the assumptions given in Section \ref{sec:SETUP}.  Each of the results for exact recovery has a direct counterpart for partial recovery.  For clarity, we focus on the former throughout Sections \ref{sec:ACHIEVABILITY} and \ref{sec:APPLYING}, and then proceed with the latter in Section \ref{sec:PARTIAL}.

\subsection{Initial Non-Asymptotic Bounds} \label{sec:ACHIEVABILITY}

Here we provide our main non-asymptotic upper and lower bounds on the error probability.  These bounds bear a strong resemblance to analogous bounds from the channel coding literature \cite{Han03}; in each case, the dominant term involves tail probabilities of the information density given in \eqref{eq:idens_bn}.  The mean of the information density is the mutual information in \eqref{eq:condMI}, which thus arises naturally in the subsequent necessary and sufficient conditions on $n$ upon showing that the deviation from the mean is small with high probability.  The procedure for doing this given a specific model will be given in Section \ref{sec:APPLYING}.

We start with our achievability result. Here and throughout this section, we make use of the random variables defined in \eqref{eq:joint_dist_n}. 

\begin{thm} \label{thm:ach1}
    For any constants $\delta_1 > 0$ and $\gamma$, there exists a decoder such that
    \begin{multline}
    \pe \le \PP\bigg[ \bigcup_{(\sdif,\seq)\,:\,\sdif\ne\emptyset} \bigg\{ \imath^n(\Xv_{\sdif}; \Yv | \Xv_{\seq}, \beta_s ) \\ \le \log {{p-k} \choose |\sdif|} + \log\bigg(\frac{k^2}{\delta_1^2}{k \choose |\sdif|}^2\bigg) + \gamma \bigg\} \bigg] + P_0(\gamma) + 2\delta_1, \label{eq:thresh_sl}
    \end{multline}
    where
    \begin{equation}
        P_0(\gamma) := \PP\bigg[ \log \frac{P_{\Yv|\Xv_s,\beta_s}(\Yv|\Xv_s,\beta_s)}{P_{\Yv|\Xv_{s}}(\Yv|\Xv_{s})} > \gamma \bigg]. \label{eq:P_0}
   \end{equation}
\end{thm}
\begin{proof}
    See Section \ref{sec:PROOF_GEN_ACH}.
\end{proof}

\begin{rem} \label{rem:P_0}
    The probability in the definition of $P_0(\gamma)$ is not an i.i.d.~sum, and the techniques for ensuring that $P_0(\gamma) \to 0$ vary between different settings.  The following approaches will suffice for all of the applications in this paper:
    \begin{enumerate}
        \item In the case that $P_{\beta_S}$ is discrete, $P_{\Yv|\Xv_{s}}(\yv|\xv_{s}) = \sum_{b_s} P_{\beta_s}(b_s) P_{\Yv|\Xv_s,\beta_s}(\yv|\xv_s,b_s)$, and it follows that 
        \begin{equation}
            \gamma = \log\frac{1}{\min_{b_s}P_{\beta_s}(b_s)} \implies P_0(\gamma) = 0. \label{eq:discrete_gamma}
       \end{equation}
        Moreover, this can be strengthened by noting from the proof of Theorem \ref{thm:ach1} that $\gamma$ may depend on $\beta_s$, and choosing $\gamma(b_s) = \log\frac{1}{P_{\beta_s}(b_s)}$ accordingly.
        \item Defining     
            \begin{align}
                I_0 &:= I(\beta_s;\Yv|\Xv_{s}) \label{eq:I_0} \\
                V_0 &:= \var\bigg[\log \frac{P_{\Yv|\Xv_s,\beta_s}(\Yv|\Xv_{s},\beta_s)}{P_{\Yv|\Xv_{s}}(\Yv|\Xv_{s})}\bigg], \label{eq:V_0} 
            \end{align}
            we have for any $\delta_0 > 0$ that
            \begin{equation}
              \gamma = I_0 + \sqrt{\frac{V_0}{\delta_0}} \implies P_0(\gamma) \le \delta_0. \label{eq:gamma_cheby}
          \end{equation}
             This follows directly from Chebyshev's inequality.
        \item Defining 
            \begin{equation}
            I_{0,+} := \EE\bigg[\bigg[\log \frac{P_{\Yv|\Xv_s\beta_s}(\Yv|\Xv_{s},\beta_s)}{P_{\Yv|\Xv_{s}}(\Yv|\Xv_{s})}\bigg]^+\bigg], \label{eq:I_0+} 
          \end{equation} 
          we have for any $\delta_0 > 0$ that
            \begin{equation}
              \gamma = \frac{I_{0,+}}{\delta_0} \implies P_0(\gamma) \le \delta_0. \label{eq:gamma_markov}
          \end{equation}
             This follows directly from Markov's inequality.
    \end{enumerate}
\end{rem}

The proof of Theorem \ref{thm:ach1} is based on a decoder the searches for a unique support set $s$ such that 
\begin{equation}
\imath(\xv_{\sdif}; \yv | \xv_{\seq} ) > \gamma_{|\sdif|} \label{eq:decoder_cond}
\end{equation}
for some $\{\gamma_{\ell}\}_{\ell=1}^k$ and all $2^{k} - 1$ partitions $(\sdif,\seq)$ of $s$ with $\sdif\ne\emptyset$.  Since the numerator in \eqref{eq:idens_multi} is the likelihood of $\yv$ given $(\xv_{\sdif},\xv_{\seq})$, this decoder can be thought of a weakened version of the maximum-likelihood (ML) decoder.  Like the ML decoder, computational considerations make its implementation intractable.

The following theorem provides a general non-asymptotic converse bound.  

\begin{thm} \label{thm:conv}
    Fix $\delta_1 > 0$, and let $(\sdif(b_s),\seq(b_s))$ be an arbitrary partition of $s=\{1,\dotsc,k\}$ (with $\sdif\ne\emptyset$) depending on $b_s\in\RR^k$.  For any decoder, we have
    \begin{multline}
    \pe \ge \PP\bigg[ \imath^n(\Xv_{\sdif(\beta_s)};\Yv|\Xv_{\seq(\beta_s)},\beta_s) \\ \le \log {{p-k+|\sdif(\beta_s)|} \choose |\sdif(\beta_s)|} + \log\delta_1 \bigg] - \delta_1. \label{eq:conv_sl}
    \end{multline} 
\end{thm}
\begin{proof}
    See Section \ref{sec:PROOF_GEN_CONV}.
\end{proof} 

The proof of Theorem \ref{thm:conv} is based on Verd\'u-Han type bounding techniques \cite{Ver94}.

\subsection{Techniques for Applying Theorems \ref{thm:ach1} and \ref{thm:conv}} \label{sec:APPLYING}

The bounds presented in the preceding theorems do not directly reveal the number of measurements required to achieving a vanishing error probability.  In this subsection, we present the steps that can be used to obtain such conditions.  We provide examples in Section \ref{sec:EXAMPLES}.

The idea is to use a concentration inequality to bound the first term in \eqref{eq:thresh_sl} (or \eqref{eq:conv_sl}), which is possible due to the fact that each summation $\imath^n$ is conditionally i.i.d.~given $\beta_s$.  We proceed by providing the details of these steps separately for the achievability and converse.  We start with the former.
\begin{enumerate}
    \item Observe that, conditioned on $\beta_s=b_s$, the mean of $\imath^n(\Xv_{\sdif}; \Yv | \Xv_{\seq}, \beta_s )$ is $nI_{\sdif,\seq}(b_s)$, where $I_{\sdif,\seq}(b_s)$ is defined in \eqref{eq:condMI}.
    \item Fix $\delta_2 \in (0,1)$, and suppose that for a fixed value $b_s$ of $\beta_s$, we have for all $(\sdif,\seq)$ that
    \begin{multline}
        \log{{p-k} \choose |\sdif|} + \log\bigg(\frac{k^2}{\delta_1^2}{k \choose |\sdif|}^2\bigg) + \gamma \\ \le n(1 - \delta_2) I_{\sdif,\seq}(b_s), \label{eq:B_cond_ach0}
    \end{multline}
    and 
    \begin{multline}
         \hspace*{-6ex}\PP\Big[ \imath^n(\Xv_{\sdif}; \Yv | \Xv_{\seq}, b_s) \le n(1 - \delta_2)I_{\sdif,\seq}(b_s) \,\big|\, \beta_s = b_s \Big] \\ \le \psi_{|\sdif|}(n,\delta_2) \label{eq:psi_ach}
    \end{multline}
    for some  functions $\{\psi_\ell\}_{\ell=1}^{k}$ (e.g., these may arise from Chebyshev's inequality or Bernstein's inequality \cite[Ch.~2]{Bou13}).   Combining these conditions with the union bound, we obtain
    \begin{multline}
        \hspace*{-3ex} \PP\bigg[ \bigcup_{(\sdif,\seq)\,:\,\sdif\ne\emptyset} \bigg\{ \imath^n(\Xv_{\sdif}; \Yv | \Xv_{\seq}, \beta_s ) \\ \le \log {{p-k} \choose |\sdif|} + \log\bigg(\frac{k^2}{\delta_1^2}{k \choose |\sdif|}^2\bigg) + \gamma \bigg\} \,\Big|\,\beta_s=b_s\bigg] \\ \le \sum_{\ell=1}^k{k \choose \ell} \psi_\ell(n,\delta_2). 
    \end{multline}
    \item Observe that the condition in \eqref{eq:B_cond_ach0} can be written as
    \begin{equation}
        n \ge \frac{ \log{{p-k} \choose |\sdif|} + \log\Big(\frac{k^2}{\delta_1^2}{k \choose |\sdif|}^2\Big) + \gamma }{ I_{\sdif,\seq}(b_s) (1-\delta_2) }. \label{eq:final_ach}
    \end{equation}
\end{enumerate}

\noindent We summarize the preceding findings in the following.

\begin{thm} \label{thm:simplified_ach}
    For any constants $\delta_1>0$, $\delta_2\in(0,1)$ and $\gamma$, and functions $\{\psi_{\ell}\}_{\ell=1}^{k}$ ($\psi_\ell : \ZZ\times\RR\to\RR$), define the set
    \begin{multline}
        \Bc(\delta_1,\delta_2,\gamma) := \big\{ b_s \,:\, \text{\eqref{eq:psi_ach} and \eqref{eq:final_ach} hold for all } \\ (\sdif,\seq)\text{ with }\sdif\ne\emptyset\big\}. \label{eq:setB}
    \end{multline}
    Then we have
    \begin{equation}
    \pe \le \PP\big[ \beta_s \notin \Bc(\delta_1,\delta_2,\gamma)\big] + \sum_{\ell=1}^k{k \choose \ell}\psi_\ell(n,\delta_2) + P_0(\gamma) + 2\delta_1. \label{eq:nonasymp_ach}
    \end{equation}
\end{thm}

\begin{rem} \label{rem:delta2}
    The preceding arguments remain unchanged when $\delta_2$ also depends on $\ell = |\sdif|$.  We leave this possible dependence implicit throughout this section, since a fixed value will suffice for all but one of the models considered in Section \ref{sec:EXAMPLES}.
\end{rem}

In the case that \eqref{eq:psi_ach} holds for all $b_s$ (or more generally, within a set whose probability under $P_{\beta_s}$ tends to one) and the final three terms in \eqref{eq:nonasymp_ach} vanish, the overall upper bound approaches the probability, with respect to $P_{\beta_s}$, that \eqref{eq:final_ach} fails to hold.  In many cases, the second logarithm in the numerator therein is dominated by the first.  It should be noted that the condition that the second term in \eqref{eq:nonasymp_ach} vanishes can also impose conditions on $n$.  For most of the examples presented in Section \ref{sec:EXAMPLES}, the condition in \eqref{eq:final_ach} will be the dominant one; however, this need not always be the case, and it depends  on the concentration inequality used in \eqref{eq:psi_ach}.

The application of Theorem \ref{thm:conv} is done using similar steps, so we provide less detail.  Fix $\delta_2 > 0$, and suppose that, for a fixed value $b_s$ of $\beta_s$, the pair $(\sdif,\seq)=(\sdif(b_s),\seq(b_s))$ is such that
\begin{equation}
    \log{{p-k+|\sdif|} \choose |\sdif|} - \log \delta_1 \ge n(1 + \delta_2)I_{\sdif,\seq}(b_s), \label{eq:B_cond_conv}
\end{equation}
and
\begin{multline}
\hspace*{-2ex}\PP\Big[ \imath^n(\Xv_{\sdif}; \Yv | \Xv_{\seq}, b_s) \le n(1 + \delta_2)I_{\sdif,\seq}(b_s) \,\big|\, \beta_s = b_s \Big] \\ \ge 1 - \psi'_{|\sdif|}(n,\delta_2) \label{eq:psi_conv}
\end{multline}
for some function $\psi'_{|\sdif|}$.  Combining these conditions, we see that the first probability in \eqref{eq:conv_sl}, with an added conditioning on $\beta_s = b_s$, is lower bounded by $1 - \psi'_{|\sdif|}(n,\delta_2)$.  In the case that $\psi'_{\ell}$ is defined for multiple $\ell$ values corresponding to different values of $b_s$, we can further lower bound this by $1 - \max_{\ell}\psi'_\ell(n,\delta_2)$.
    
Next, we observe that \eqref{eq:B_cond_conv} holds if and only if
\begin{equation}
n \le \frac{\log {{p-k+|\sdif|} \choose |\sdif|} - \log\delta_1 }{ I_{\sdif,\seq}(b_s)(1 + \delta_2) }. \label{eq:final_conv}
\end{equation} 
Recalling that the partition $(\sdif,\seq)$ is an arbitrary function of $\beta_s$, we can ensure that this coincides with
\begin{equation}
n \le \max_{(\sdif,\seq) \,:\,\sdif\ne\emptyset} \frac{\log {{p-k+|\sdif|} \choose |\sdif|}  - \log\delta_1 }{ I_{\sdif,\seq}(b_s) (1 + \delta_2) }  \label{eq:final_conv_max}
\end{equation}
by choosing each pair $(\sdif,\seq)$ as a function of $b_s$ to achieve this maximum.

Finally, we note that the maximum over $\ell$ in the above-derived term  $1 - \max_{\ell}\psi'_\ell(n,\delta_2)$ may be restricted to any set $\Lc \subseteq \{1,\dotsc,k\}$ provided that $|\sdif|$ is constrained similarly in \eqref{eq:final_conv_max}; one simply chooses the partition $(\sdif(b_s),\seq(b_s))$ so that $\ell = |\sdif|$ always lies in this set.  Putting everything together, we have the following.

\begin{thm} \label{thm:simplified_conv}
	For any set $\Lc \subseteq \{1,\dotsc,k\}$, constants $\delta_1>0$ and $\delta_2 > 0$, and functions $\{\psi'_{\ell}\}_{\ell\in\Lc}$ ($\psi'_\ell : \ZZ\times\RR\to\RR$), define the set
    \begin{multline}
        \Bc'(\delta_1,\delta_2) := \big\{ b_s \,:\, \text{\eqref{eq:psi_conv} and \eqref{eq:final_conv} hold for all } \\ (\sdif,\seq)\text{ with }|\sdif| \in \Lc \big\}. \label{eq:setB'}
    \end{multline}
    Then we have
    \begin{equation}
    \pe \ge \PP\big[ \beta_s \in \Bc'(\delta_1,\delta_2) \big]\Big(1 - \max_{\ell \in \Lc}\psi'_\ell(n,\delta_2)\Big) - \delta_1. \label{eq:nonasymp_conv}
    \end{equation}
\end{thm}

If the pair $(\sdif,\seq)$ had been fixed in Theorem \ref{thm:conv}, as opposed to being a function of $\beta_s$, then we would have only obtained a weaker result with the statement ``for all $(\sdif,\seq)$'' in \eqref{eq:setB'} replaced by a fixed pair.  Assuming that the remainder terms in \eqref{eq:nonasymp_conv} are insignificant, this weaker result is of the form $\pe \gtrsim \max_{(\sdif,\seq)} \PP\big[n \le f(\sdif,\seq,\beta_s)\big]$ rather than $\pe \gtrsim \PP\big[n \le \max_{(\sdif,\seq)} f(\sdif,\seq,\beta_s)\big]$.  This can lead to significantly different bounds on the sample complexity, and the distinction is crucial in our applications in Section \ref{sec:EXAMPLES}.  As described in the proof in Section \ref{sec:PROOFS}, the key to obtaining this difference is in applying a refined version of an argument based on a genie.

The general steps in applying Theorems \ref{thm:simplified_ach} and \ref{thm:simplified_conv} to specific problems are outlined in Procedure \ref{pr:procedure}. 

\begin{algorithm}
\caption*{ \manuallabel{pr:procedure}{1} \textbf{Procedure 1:} Steps for Obtaining Necessary and Sufficient Conditions on $n$ from Theorems \ref{thm:simplified_ach} and \ref{thm:simplified_conv} }

\begin{enumerate}
\item \emph{(Identify a Typical Set)} Construct a sequence of ``typical'' sets $\Tc_{\beta} \subseteq \RR^k$ of non-zero entries, indexed by $p$, such that $\PP\big[\beta_s \in \Tc_{\beta}\big] \to 1$, thus restricting the vectors $b_s$ for which  $\imath(X_{\sdif};Y|X_{\seq},b_s)$ needs to be characterized.
\item \emph{(Bound the Information Density Tail Probabilities)} Using a concentration inequality for i.i.d.~summations (e.g., Chebyshev, Bernstein), bound the  tail probabilities in \eqref{eq:psi_ach} and \eqref{eq:psi_conv} for each $(\sdif,\seq)$ and $b_s \in \Tc_{\beta}$, with a fixed constant $\delta_2$.  Upon making these dependent on $(\sdif,\seq,b_s)$ only through $\ell := |\sdif|$, the bounds are denoted by $\psi_{\ell}(n,\delta_2)$ and $\psi'_{\ell}(n,\delta_2)$.
\item \emph{(Control the Remainder Terms)} By suitable rearrangements, find conditions on $n$ under which the terms $\sum_{\ell}{k \choose \ell}\psi_\ell(n,\delta_2)$ and $\max_{\ell \in \Lc}\psi'_\ell(n,\delta_2)$ in \eqref{eq:nonasymp_ach} and \eqref{eq:nonasymp_conv} vanish, thus ensuring that their contribution is negligible.  Similarly,  choose $\delta_1$ to vanish with $p$ so that its contribution is negligible, and for the achievability part, choose $\gamma$ such that the remainder term $P_0(\gamma)$ vanishes (\emph{cf.}~Remark \ref{rem:P_0}).
\item \emph{(Combine and Simplify)} Combine the previous steps as follows:
\begin{enumerate}
    \item Construct the set of non-zero entries $\Bc(\delta_1,\delta_2,\gamma) \subseteq \RR^{k}$ (respectively, $\Bc'(\delta_1,\delta_2)$) in \eqref{eq:setB} (respectively, \eqref{eq:setB'});
    \item Deduce from \eqref{eq:nonasymp_ach} (respectively, \eqref{eq:nonasymp_conv}) and Step 3 that $\pe \le \PP[\beta_s\notin\Bc(\delta_1,\delta_2,\gamma)] + o(1)$ (respectively,  $\pe \ge \PP[\beta_s\in\Bc'(\delta_1,\delta_2)] + o(1)$);
    \item From the properties of the typical set $\Tc_{\beta}$ in Steps 1--2, deduce that $\pe\to0$ (respectively, $\pe\to1$) when $n$ satisfies \eqref{eq:final_ach} (respectively, \eqref{eq:final_conv}) for all $b_s \in \Tc_{\beta}$;
    \item Augment this condition on $n$ with Step 3.
\end{enumerate} 
\end{enumerate}
\end{algorithm}

In our experience, the choice of $\Tc_{\beta}$ in the first step of Procedure \ref{pr:procedure} usually comes naturally given the specific model.  On the other hand, it is often less straightforward to find a sufficiently powerful concentration inequality in Step 2.  A simple choice is Chebyshev's inequality, which expresses $\psi_{\ell}$ and $\psi'_{\ell}$ in terms of $I_{\sdif,\seq}(b_s)$ (see \eqref{eq:condMI}) and the corresponding variances of the information densities.  This choice is usually effective for the converse, wheres the achievability part typically requires sharper concentration inequalities such as Bernstein's inequality, due to the combinatorial terms in \eqref{eq:nonasymp_ach}.

\subsection{Extensions to Partial Recovery} \label{sec:PARTIAL}

We now turn to the partial support recovery criterion in \eqref{eq:pe_partial}.  The changes in the analysis required to generalize Theorems \ref{thm:ach1} and \ref{thm:conv} to this setting are given in Section \ref{sec:PROOF_PARTIAL}; rather than repeating each of these, we focus our attention on the resulting analogues of Theorems \ref{thm:simplified_ach} and \ref{thm:simplified_conv}.

\begin{thm} \label{thm:simplified_ach_p}
    For any constants $\delta_1>0$, $\delta_2\in(0,1)$ and $\gamma>0$, and functions $\{\psi_{\ell}\}_{\ell=\dmax+1}^{k}$ ($\psi_\ell : \ZZ\times\RR\to\RR$), define the set
    \begin{multline}
        \Bc(\delta_1,\delta_2,\gamma) := \big\{ b_s \,:\, \text{\eqref{eq:psi_ach} and \eqref{eq:final_ach} hold for all }\\ (\sdif,\seq)\text{ with }|\sdif|\in\{\dmax+1,\dotsc,k\} \big\}. \label{eq:setB_p}
    \end{multline}
    Then we have
    \begin{multline}
    \pe(\dmax) \le \PP\big[ \beta_s \notin \Bc(\delta_1,\delta_2,\gamma)\big] \\ + \sum_{\ell=\dmax+1}^{k} {k \choose \ell}\psi_\ell(n,\delta_2) + P_0(\gamma) + 2\delta_1, \label{eq:nonasymp_ach_p}
    \end{multline}
    where $P_0$ is defined in \eqref{eq:P_0}.
\end{thm}

For the converse part, \eqref{eq:final_conv} is replaced by 
\begin{equation}
    n \ge \frac{\log {{p-k+|\sdif|} \choose |\sdif|} - \log\sum_{d=0}^{\dmax}{{p-k} \choose d}{|\sdif| \choose d} - \log\delta_1 }{ I_{\sdif,\seq}(b_s) }, \label{eq:final_conv_p}
\end{equation}
and we have the following analog of Theorem \ref{thm:simplified_conv}.

\begin{thm} \label{thm:simplified_conv_p}
	For any set $\Lc \subseteq \{\dmax+1,\dotsc,k\}$, constants $\delta_1>0$ and $\delta_2\in(0,1)$, and functions $\{\psi'_{\ell}\}_{\ell\in\Lc}$ ($\psi'_\ell : \ZZ\times\RR\to\RR$), define the set
    \begin{multline}
        \Bc'(\delta_1,\delta_2) := \big\{ b_s \,:\, \text{\eqref{eq:psi_conv} and \eqref{eq:final_conv_p} hold for all } \\ (\sdif,\seq)\text{ with }|\sdif|\in\Lc \big\}. \label{eq:setB'_p}
    \end{multline}
    Then we have
    \begin{equation}
    \pe(\dmax) \ge \PP\big[ \beta_s \in \Bc'(\delta_1,\delta_2) \big]\Big(1 - \max_{\ell\in\Lc}\psi'_\ell(n,\delta_2)\Big) - \delta_1. \label{eq:nonasymp_conv_p}
    \end{equation}
\end{thm}

The applications of Theorems \ref{thm:simplified_ach_p} and \ref{thm:simplified_conv_p} follow identical steps to Procedure \ref{pr:procedure}.  However, it will be seen that the restriction $|\sdif| > \dmax$ can in fact considerably simplify these steps, since it removes the need to obtain concentration inequalities for smaller values of $|\sdif|$.

\subsection{Comparison to Fano's Inequality} \label{sec:FANO}

Most previous works on the information-theoretic limits of sparsity recovery have made use of Fano's inequality \cite[Sec.~2.11]{Cov01}.  For this reason, we provide here a discussion on the relative merits of this approach and our approach.  To this end, we consider the following bound, which can be obtained by combining the analysis of \cite{Ati12,Aks13} with our refined genie argument:
\begin{equation}
    \pe \ge \sum_{b_s}P_{\beta_S}(b_s)\max\bigg\{0, 1 - \frac{nI_{\sdif(b_s),\seq(b_s)}(b_s) + 1}{\log{ p-k+|\sdif(b_s)| \choose |\sdif(b_s)| }}\bigg\}
\end{equation}
in the notation of Theorem \ref{thm:conv}.  By analyzing this bound similarly to Section \ref{sec:APPLYING}, we obtain for any $\delta_2 > 0$ that
\begin{equation}
    \pe \ge \delta_2\, \PP\big[\beta_s \in \Bcfano(\delta_2)\big] - \frac{1}{\log(p-k+1)}, \label{eq:pe_Fano}
\end{equation}
where
\begin{multline}
    \Bcfano(\delta_2) := \bigg\{ b_s \,:\, n \le \frac{\log {{p - k + |\sdif|} \choose |\sdif|}}{I_{\sdif,\seq}(b_s)}(1-\delta_2) \\ \text{ for all }(\sdif,\seq)\text{ with }\sdif\ne\emptyset\bigg\}. \label{eq:pe_Fano2}
\end{multline}
A similar result for partial recovery can also be derived by incorporating the arguments from \cite{Ree13} and the present paper.  

As discussed in the introduction, the key advantage of Theorem \ref{thm:simplified_conv} is that it provides a more precise characterization of how far the error probability is from zero, and in particular, the conditions under which $\pe \to 1$ (strong impossibility results). On the other hand, the bound on $\pe$ in \eqref{eq:pe_Fano} is always bounded away from one for fixed $\delta_2$, and becomes increasingly weak for small $\delta_2$. 

The advantage of Fano's inequality is that it only requires the mutual information to be computed, whereas our approach also requires the application of a concentration inequality.  This, in turn, typically requires the variance of the information density to be characterized, which is not always straightforward.  However, as discussed following Procedure \ref{pr:procedure}, the main difficulty associated with these concentration inequalities is typically in finding one which is sufficiently powerful for the \emph{achievability} part.  Thus, the added difficulty in the converse may not add to the overall difficulty in deriving matching achievability and converse bounds.  

\section{Applications to Specific Models} \label{sec:EXAMPLES}

In this section, we present applications of Theorems \ref{thm:simplified_ach}--\ref{thm:simplified_conv_p} to the linear \cite{Wai09}, 1-bit \cite{Bou08}, and group testing \cite{Ati12} models, and to more general models with discrete observations \cite{Tan14}.  Throughout the section, we make use of general concentration inequalities given in Appendix \ref{sec:CONCENTRATION}.  We also make use of the following variance quantity:
\begin{equation}
V_{\sdif,\seq}(b_s) := \var\big[ \imath(X_{\sdif}; Y | X_{\seq}, b_s) \,\big|\, \beta_s = b_s \big]. \label{eq:Vb}
\end{equation}

\subsection{Linear Model with Discrete $\beta_s$} \label{sec:GAUSSIAN_DISC}

Here we consider the linear model, where each observation takes the form
\begin{equation}
    Y = \langle X, \beta \rangle + Z, \label{eq:linear_model}
\end{equation}
where $Z \sim N(0,\sigma^2)$ for some $\sigma > 0$. 

Without loss of generality, we consider the fixed support set $s=\{1,\dotsc,k\}$.  Following the setup of \cite{Jin11}, we let $\beta_s$ be a uniformly random permutation of a fixed vector $(b_1,\dotsc,b_k)$, and we choose $P_X \sim N(0,1)$.  Since both the measurement matrices and the noise are Gaussian, the mutual information in \eqref{eq:condMI} is given by \cite[Ch.~10]{Cov01}
\begin{equation}
    I_{\sdif,\seq}(b_s) = \frac{1}{2}\log\Big(1 + \frac{1}{\sigma^2} \sum_{i \in \sdif}b_i^2\Big). \label{eq:I_linear}
\end{equation}

Throughout this subsection, we denote $\bmin := \min_{i} |b_i|$ and $\bmax := \max_{i} |b_i|$.  We assume that $\sigma^2 = \Theta(1)$, and that $\bmin = \Theta(\bmax)$ and $0<\bmin = O(1)$; note that $\bmin = o(1)$ is allowed.  The steps of Procedure \ref{pr:procedure} are as follows.

\subsubsection*{Step 1}

We trivially choose the typical set $\Tc_{\beta}$ to contain all vectors on the support of $P_{\beta_S}$.

\subsubsection*{Step 2}

We make use of the following concentration inequality based on Bernstein's inequality.

\begin{prop} \label{prop:conc_linear2}
    Under the preceding setup for the linear model, we have for all $(\sdif,\seq)$ and $b_s$ that
    \begin{multline}
        \PP\bigg[ \big|\imath^n(\Xv_{\sdif}; \Yv | \Xv_{\seq}, b_s) - nI_{\sdif,\seq}(b_s) \big| \ge n\delta \,\Big|\, \beta_s = b_s \bigg] \\ \le 2\exp\bigg(-\frac{\delta^2 n}{2(4\alpha_{\sdif}^2+\delta\alpha_{\sdif})} \bigg),
    \end{multline}
    where
    \begin{equation}
        \alpha_{\sdif} := \frac{2\siglv(\sigma + \siglv)}{\sigma^2 + \siglv^2} \label{eq:alpha_linear}
   \end{equation}
   with $\siglv^2 := \sum_{i \in \sdif} b_i^2$.
\end{prop}
\begin{proof}
    See Appendix \ref{sec:LINEAR_PROOFS}.
\end{proof} 

Setting $\delta = \delta_2 I_{\sdif,\seq}(b_s)$, it follows that in \eqref{eq:psi_ach} and \eqref{eq:psi_conv} we can set 
\begin{multline}
\psi_{\ell}(n,\delta_2) = \psi'_{\ell}(n,\delta_2)  = 2\max_{(\sdif,\seq,b_s)\,:\,|\sdif|=\ell} \\ \exp\bigg(- \frac{ ( \delta_2 I_{\sdif,\seq}(b_s) )^2 n}{2(4\alpha_{\sdif}+\delta_2 I_{\sdif,\seq}(b_s))\alpha_{\sdif} } \bigg). \label{eq:linear_psi}
\end{multline}

\subsubsection*{Step 3}

In accordance with the first item of Remark \ref{rem:P_0}, we set $\gamma$ as in \eqref{eq:discrete_gamma} so that $P_0(\gamma) = 0$. 

We focus on the conditions on $n$ under which the term $\sum_{\ell=1}^k{k \choose \ell}\psi_\ell(n,\delta_2)$ in \eqref{eq:nonasymp_ach} vanishes; the term containing $\psi'_{\ell}$ in \eqref{eq:nonasymp_conv} can be handled in a similar yet simpler fashion.  By the assumptions $\sigma^2 = \Theta(1)$ and $\bmax = \Theta(\bmin)$, we readily obtain $I_{\sdif,\seq}(b_s) = \Theta(\log(1+\ell\bmin^2))$ and $\alpha_{\sdif}^2 = \Theta(\min\{1,\ell\bmin^2\})$ using \eqref{eq:I_linear} and \eqref{eq:alpha_linear}, where $\ell = |\sdif|$. Using these growth rates and upper bounding the summation in \eqref{eq:nonasymp_ach} by $k$ times the corresponding maximum, we see that $\sum_{\ell=1}^k{k \choose \ell}\psi_\ell(n,\delta_2) \to 0$ provided that the following holds for some sufficiently small constant $\zeta$ (depending on $\delta_2$):
\begin{multline}
    \frac{ n \log^2(1+\ell\bmin^2)}{\min\{1,\ell\bmin^2\} + \log(1+\ell\bmin^2) \sqrt{\min\{1,\ell\bmin^2\}} }\zeta \\ - \ell \log\frac{k}{\ell} - \log k \to \infty \label{eq:lin_n_cond1}
\end{multline}
for all $\ell$.  We now treat two cases separately:
\begin{itemize}
    \item If $\ell\bmin^2 = o(1)$, the first term in \eqref{eq:lin_n_cond1} behaves as $\Theta(n \ell\bmin^2)$; by rearranging, we conclude that it suffices that $n\to\infty$ and $n = \Omega\big(\frac{\log k}{\bmin^2}\big)$ with a sufficiently large implied constant.
    \item If $\ell\bmin^2 = \Omega(1)$, the first term in \eqref{eq:lin_n_cond1} behaves as $\Omega(n)$, and it thus suffices that $n\to\infty$ and $n=\Omega(k)$ with a sufficiently large implied constant.
\end{itemize}
Thus, the overall condition that we require is $n\to\infty$ and
\begin{equation}
    n = \Omega\Big( \frac{\log k}{\bmin^2} \Big) \quad \text{ and } \quad n = \Omega(k), \label{eq:linear_n_cond}
\end{equation}
with sufficiently large implied constants.  For the converse, the analogous condition to \eqref{eq:lin_n_cond1} contains only the first term on the left-hand side (the difference being due to the fact that the combinatorial term in \eqref{eq:nonasymp_ach} is not present in \eqref{eq:nonasymp_conv}), and a similar argument reveals that it suffices that $n = \omega\big( \frac{1}{\bmin^2} \big)$.

\subsubsection*{Step 4}

Combining the preceding steps and applying asymptotic simplifications, we obtain the following.

\begin{cor}  \label{cor:linear_genK}
    Under the preceding setup for the linear model with $\sigma^2 = \Theta(1)$, $\bmin=\Theta(\bmax)$, $\bmin^2 = O(1)$, $k=o(p)$, and $m_{\beta}$ distinct elements in $(b_1,\dotsc,b_k)$, we have $\pe \to 0$ as $p\to\infty$ provided that
    \begin{equation}
    n \ge \max_{\sdif\ne\emptyset} \frac{ \log {{p-k} \choose |\sdif|} }{ \frac{1}{2}\log\Big(1 + \frac{1}{\sigma^2} \sum_{i \in \sdif}b_i^2 \Big)} (1 + \eta), \label{eq:linear2_ach}
    \end{equation}
    under any one of the following additional conditions: (i) $k = \Theta(1)$; (ii) $k=o(\log p)$ and $m_{\beta} = \Theta(1)$; (iii) $k = O((\log p)^{\theta})$ for some $\theta>0$, and $m_{\beta} = 1$; (iv) $k=\Theta(p^{\theta})$ for some $\theta\in(0,1)$, $\bmin^2 = \Theta\big( \frac{\log k}{k} \big)$, and $m_{\beta} = 1$.
    
    Conversely, without any additional conditions, we have $\pe \to 1$ as $p\to\infty$ whenever
    \begin{equation}
    n \le \max_{\sdif\ne\emptyset} \frac{ \log {p-k+|\sdif| \choose |\sdif|} }{ \frac{1}{2}\log\Big(1 + \frac{1}{\sigma^2} \sum_{i \in \sdif}b_i^2 \Big)} (1 - \eta) \label{eq:linear2_conv}
    \end{equation}
    for some $\eta > 0 $. 
\end{cor} 
\begin{proof}
    The converse part follows from \eqref{eq:final_conv} with $\delta_1\to0$ sufficiently slowly.  To check the condition $n = \omega\big( \frac{1}{\bmin^2} \big)$ stated following \eqref{eq:linear_n_cond}, we may assume without loss of generality that \eqref{eq:linear2_conv} holds with equality, since the decoder can always choose to ignore additional measurements.  When equality holds, we observe that for the worst-case $\sdif$ with $\ell=1$, the denominator therein behaves as $O(\bmin^2)$ (since $\bmin^2 = O(1)$) and the numerator behaves as $\Theta(\log p)$, and hence, the condition $n = \omega\big( \frac{1}{\bmin^2} \big)$ is satisfied.
    
     For the achievability part, we first use \eqref{eq:final_ach} to obtain 
    \begin{equation}
    n \ge \max_{\sdif\ne\emptyset} \frac{ \log {{p-k} \choose |\sdif|} + 2\log\big(k{k\choose|\sdif|}\big) + k\log m_{\beta} }{ \frac{1}{2}\log\Big(1 + \frac{1}{\sigma^2} \sum_{i \in \sdif}b_i^2 \Big)} (1 + \eta), \label{eq:linear2_ach_0}
    \end{equation}
    where the final term in the numerator arises from \eqref{eq:discrete_gamma} since $P_{\beta_s}(b_s)$ is the same for all permutations of $(b_1,\dotsc,b_k)$, and is lower bounded by $m_{\beta}^{-k}$.   Observe that the first term in the numerator behaves as $\Theta(|\sdif| \log p)$ for each of the cases in the corollary statement, and the second term behaves as $\Theta\big(\log k + |\sdif| \log \frac{k}{|\sdif|}\big)$.
    
    In cases (i)--(iii), we have $\log k = o(\log p)$, and it immediately follows that the numerator in \eqref{eq:linear2_ach_0} is dominated by the first term, and hence, the others can be factored into $\eta$ in \eqref{eq:linear2_ach}.  Moreover, in case (i), both conditions in \eqref{eq:linear_n_cond} are dominated by the objective in \eqref{eq:linear2_ach_0} with $\ell:=|\sdif|=1$, which behaves as $\Theta\big(\frac{\log p}{\bmin^2}\big)$.  In cases (ii)--(iii), the first condition in \eqref{eq:linear_n_cond} is again dominated by the term in \eqref{eq:linear2_ach_0} with $\ell=1$.  The second condition is dominated by the term with $\ell = k$, which behaves as $\Theta\big( \frac{k\log p}{ \log(1+k\bmin^2) } \big) = \Omega\big(k \frac{\log p}{\log k}\big)$.
    
    In case (iv), the first term in the numerator of \eqref{eq:linear2_ach_0} may not be dominant for small $\ell := |\sdif|$, since $\log k = \Theta(\log p)$.  However, by observing that the objective scales as $\Theta\big( \frac{\ell \log p}{ \log(1+\ell\bmin^2) } \big)$ and using the assumed scaling of $\bmin^2$, it is readily verified that the maximum can only be achieved with $\ell = \Theta(k)$.  For any such maximizer, we have $\log{ p-k \choose \ell } = \Theta(k \log p)$, and hence, the second term in the numerator of \eqref{eq:linear2_ach_0} can be factored into $\eta$, as it behaves as $O(k)$.  The two conditions in \eqref{eq:linear_n_cond} are identical under the given scaling of $\bmin^2$, and are dominated by the objective in \eqref{eq:linear2_ach} with $\ell=k$, which behaves as $\Theta\big( \frac{k \log k}{\log\log k} \big)$.
\end{proof}

In the case that $\bmin = \Theta(1)$, the thresholds given in Corollary \ref{cor:linear_genK} coincide with those given in the main results of \cite{Jin11}.  Our framework has the advantage of handling the case that $\bmin = o(1)$, as well as providing the strong converse ($\pe \to 1$) instead of the weak converse ($\pe \not\to 0$).  However, it should be noted that the achievability parts of \cite{Jin11} have the notable advantage of using a decoder that does not depend on the distribution of $\beta_s$.

On first glance, the bounds in \eqref{eq:linear2_ach}--\eqref{eq:linear2_conv} may appear to be difficult to evaluate, since the maximizations are over $2^{k}-1$ non-empty subsets $\sdif$.  However, it is in fact only $k$ of them that need to be computed, since for any given $\ell = |\sdif|$ the maximizing $\sdif$ is the one with the smallest corresponding value of $\sum_{i \in \sdif}b_i^2$.  

\subsubsection*{Comparison to the LASSO}

Conditions for the support recovery of the computationally tractable LASSO algorithm were given by Wainwright \cite{Wai09a}.  Several comparisons to the information-theoretic limits were given in \cite{Wai09,Wai09a} in terms of scaling laws; here we complement these comparisons by briefly discussing the corresponding constant factors.  For simplicity, we focus on the case that the non-zero entries are all equal to a common value $b_0 = \frac{c_{\beta}}{k}$ (for some constant $c_{\beta}$ representing the per-sample SNR) and $k$ is poly-logarithmic in $p$, corresponding to case (iii) of Corollary \ref{cor:linear_genK}.

The results of \cite{Wai09a} state that LASSO requires at least $(2k\log p)(1+o(1))$ measurements regardless of $c_{\beta}$, and that this bound is also achievable in the limit as $c_{\beta} \to \infty$.  On the other hand, Corollary \ref{cor:linear_genK} reveals that for the optimal decoder, the coefficient to $k \log p$ can be arbitrarily small provided that $c_{\beta}$ is large enough.  More precisely, applying some simple manipulations to \eqref{eq:linear2_ach}, we find that the coefficient to $k \log p$ is $\sup_{\alpha\in(0,1]} \frac{\alpha}{\frac{1}{2}\log(1+c_{\beta}\alpha)}$, where $\alpha$ represents the ratio $\frac{|\sdif|}{k}$.  It is easy to verify that the maximum is achieved at $\alpha = 1$, yielding the constant $\frac{2}{\log(1+c_{\beta})}$.  We conclude that the LASSO provably yields a suboptimal constant when $c_{\beta} > 1$, and fails to achieve the optimal logarithmic decay.  However, it should be noted that our decoder requires knowledge of $k$ and $c_{\beta}$, whereas the LASSO does not (except possibly via their role in determining the regularization parameter).

\subsection{Linear Model with Gaussian $\beta_S$ and Partial Recovery} \label{sec:GAUSSIAN_CONT}

In this subsection, we consider the setup of Section \ref{sec:GAUSSIAN_DISC} with two changes: We let the distribution of $\beta_s$ be continuous rather than discrete, and we consider partial recovery instead of exact recovery.  More specifically, we let $\beta_s$ be i.i.d.~on $N(0,\sigbeta^2)$ for some variance $\sigbeta^2$, and we consider the recovery condition in \eqref{eq:pe_partial} with
\begin{equation}
    \dmax = \lfloor \alpha^* k \rfloor
\end{equation}
for some $\alpha^* \in (0,1)$ not varying with $p$.  We again choose $P_X \sim N(0,1)$.  We assume $\sigbeta^2 = \frac{c_{\beta}}{k}$ for some $c_{\beta} > 0$ not depending on $p$, corresponding to a fixed per-sample SNR.

We begin with the following auxiliary result.

\begin{prop} \label{prop:boundMI}
    Under the preceding setup for the linear model, the quantities $I_0$ and $V_0$ defined in \eqref{eq:I_0}--\eqref{eq:V_0} satisfy
    \begin{align}
        I_0  &\le \frac{k}{2} \log\Big( 1 + \frac{n\sigbeta^2}{\sigma^2} \Big) \label{eq:bmi_I0} \\
        V_0 &\le 2n. \label{eq:bmi_V0}
    \end{align}
\end{prop}
\begin{proof}
    See Appendix \ref{sec:LINEAR_PROOFS}.
\end{proof}

We now proceed with the steps of Procedure \ref{pr:procedure} (with the suitable changes from exact recovery to partial recovery, \emph{cf.}~Section \ref{sec:PARTIAL}).

\subsubsection*{Step 1}

Our choice of the typical set $\Tc_{\beta}$ is based on the following proposition characterizing the behavior of the $\lfloor \alpha k \rfloor$ entries of $\beta_s$ having the smallest magnitude for fixed $\alpha$.  We define the random variable $\beta'_s$ to be the permutation of $\beta_s$ whose entries are listed in increasing order of magnitude.

\begin{prop} \label{prop:empiricalB}
    For any $\alpha\in(0,1]$, we have
    \begin{equation}
        \lim_{k\to\infty} \frac{1}{k\sigbeta^2}\sum_{i=1}^{\lfloor \alpha k \rfloor} (\beta'_s)_i^2 = g(\alpha) \label{eq:empiricalB}
    \end{equation}
    with probability one, where
    \begin{equation}
        g(\alpha) := \int_{0}^{\infty} \big[ \alpha - F_{\chi^2}(u)\big]^+ \,du, \label{eq:g_alpha}
    \end{equation}
    and $F_{\chi^2}$ is the cumulative distribution function of a $\chi^2$ random variable with one degree of freedom.
\end{prop}
\begin{proof}
    Letting $\hat{F}_k$ be the empirical distribution of the values $\big\{\frac{1}{\sigbeta^2}\beta_i^2\big\}_{i=1}^k$, we have from the Glivenko-Cantelli theorem \cite[Thm.~19.1]{Van00} that $\sup_{u} |\hat{F}_k(u) - F_{\chi^2}(u)| \to 0$ almost surely. This immediately implies that the sum of the $\lfloor \alpha k \rfloor$ smallest values in $\big\{\frac{1}{\sigbeta^2}\beta_i^2\big\}$, normalized by the number of values $k$, converges almost surely to the integral of $F_{\chi^2}^{-1}(u)$ from $0$ to $\alpha$.  It is easily verified graphically that this integral can equivalently be written as \eqref{eq:g_alpha}.
\end{proof}

Based on this result and its proof, we set $\Tc_{\beta}$ to be the set of vectors $b_s$ such that $\sup_{u} |\hat{F}_k(u) - F_{\chi^2}(u)| \le \epsilon$, where $\epsilon$ is chosen to decay sufficiently slowly so that $\PP[\beta_s \in \Tc_{\beta}] \to 1$.  Thus, within the typical set, the empirical distribution of the non-zero entries closely follows a $\chi^2$ random variable.

An important consequence of this choice of typical set regards the behavior of the mutual information in \eqref{eq:I_linear}.  For a fixed set size $|\sdif|$, the partition $(\sdif,\seq)$ minimizing this mutual information is the one with the smallest value of $\sum_{i\in\sdif}b_i^2$.  Within the typical set, we immediately obtain from Proposition \ref{prop:empiricalB} that the corresponding mutual information behaves as follows when $|\sdif| = \lfloor \alpha k \rfloor$:
\begin{equation}
    I_{\sdif,\seq}(b_s) \to \frac{1}{2} \log\Big(1 + \frac{c_{\beta}}{\sigma^2} g(\alpha) \Big), \label{eq:linear_I_asymp}
\end{equation}
where we recall that $c_{\beta} = k\sigbeta^2$ is a constant.

\subsubsection*{Step 2}

We again make use of Proposition \ref{prop:conc_linear2} and its subsequent expression for $\psi_{\ell}$ and $\psi'_{\ell}$ in \eqref{eq:linear_psi}.

\subsubsection*{Step 3}

We choose $\gamma = I_0 + \sqrt{\frac{V_0}{\delta_0}}$ as in \eqref{eq:gamma_cheby} for some $\delta_0 > 0$, thus ensuring that $P_0(\gamma) \le \delta_0$.

For the terms in Theorems \ref{thm:simplified_ach_p}--\ref{thm:simplified_conv_p} containing $\psi_{\ell}$ and $\psi'_{\ell}$, we first note that since we are considering partial recovery, we may focus on values of $\ell=|\sdif|$ greater than $\alpha^{*} k$.  By our choice of $\Tc_{\beta}$, we may also focus on realizations $b_s$ of $\beta_s$ satisfying \eqref{eq:empiricalB}.  For such realizations, we have for all $\sdif$ with $|\sdif|=\ell=\Theta(k)$ that $\sum_{i\in\sdif} b_i^2 = \Omega(1)$, which implies that $\alpha_{\sdif}^2 = \Theta(1)$ in \eqref{eq:alpha_linear} and $I_{\sdif,\seq}(b_s) = \Omega(1)$ in \eqref{eq:I_linear}.  The analogous condition to \eqref{eq:lin_n_cond1} thus simplifies to $n I' \gg k$ for some $I' = \Omega(1)$, giving the following condition under which the second term in \eqref{eq:nonasymp_ach} vanishes:
\begin{equation}
    n = \Omega(k), \label{eq:lin_cond_mod}
\end{equation}
with a sufficiently large implied constant.  For the converse part, it suffices to have the weaker condition $n=\omega(1)$.

\subsubsection*{Step 4}

Combining the above steps, we get the following.

\begin{cor} \label{cor:linear_partial}
    Under the preceding setup for the linear model with $k\to\infty$, $k=o(p)$, $\sigbeta^2 = \frac{c_{\beta}}{k}$ for some $c_{\beta} > 0$, and $\dmax = \lfloor \alpha^* k \rfloor$ for some $\alpha^*\in(0,1)$, we have $\pe(\dmax) \to 0$ as $p\to\infty$ provided that
    \begin{equation}
    n \ge \max_{\alpha\in[\alpha^*,1]} \frac{ \alpha k \log \frac{p}{k} }{ \frac{1}{2}\log\Big(1 + \frac{c_{\beta}}{\sigma^2} g(\alpha) \Big)} (1 + \eta) \label{eq:linear3_ach}
    \end{equation}
    for some $\eta > 0$, where $g(\cdot)$ is defined in \eqref{eq:g_alpha}. Conversely, $\pe(\dmax) \to 1$ as $p\to\infty$ whenever
    \begin{equation}
    n \le \max_{\alpha\in[\alpha^*,1]} \frac{ (\alpha - \alpha^*) k \log \frac{p}{k} }{ \frac{1}{2}\log\Big(1 + \frac{c_{\beta}}{\sigma^2} g(\alpha) \Big)} (1 - \eta) \label{eq:linear3_conv}
    \end{equation}
    for some $\eta > 0 $. 
\end{cor}
\begin{proof}
    The condition in \eqref{eq:linear3_ach} is obtained using \eqref{eq:final_ach} and \eqref{eq:linear_I_asymp}.  By the assumption $k=o(p)$, the numerator in \eqref{eq:linear3_ach} coincides with $\log {{p-k} \choose \lfloor \alpha k \rfloor}$ up to remainder terms in Stirling's approximation that can be factored into $\eta$.  The factor $\log\big(\frac{k^2}{\delta_1^2}{k \choose |\sdif|}^2\big)$ in \eqref{eq:final_ach} has been factored into $\eta$; this is valid when $\delta_1 \to 0$ sufficiently slowly due to the fact that $\log\big(k {k \choose |\sdif|}\big) = O(k)$, whereas (again using the assumption $k=o(p)$) the numerator in \eqref{eq:linear3_ach} behaves as $\omega(k)$.  We claim that the factor $\gamma = I_0 + \sqrt{\frac{V_0}{\delta_0}}$ resulting from \eqref{eq:gamma_cheby} can also be factored into $\eta$ for some vanishing sequence of parameters $\delta_0$ indexed by $p$.  To see this, we consider without loss of generality the ``worst-case'' setting in which \eqref{eq:linear3_ach} holds with equality.  We readily obtain $n = \Theta( k \log \frac{p}{k} )$, which in turn implies from Proposition \ref{prop:boundMI} that $I_0 = O\big(k \log(1 + k\sigbeta^2 \log\frac{p}{k} )\big) = O(k\log\log\frac{p}{k}) $ and $\sqrt{V_0} = O(\sqrt{k \log \frac{p}{k}})$.  Thus, $I_0 + \sqrt{\frac{V_0}{\delta_0}}$ is dominated by the numerator of \eqref{eq:linear3_ach} if $\delta_0$ is chosen to decay as (for example) $\Theta\big( \frac{1}{\log k} \big)$.  The fact that $n = \Theta( k \log \frac{p}{k} )$ also implies \eqref{eq:lin_cond_mod}.
    
    The converse bound in \eqref{eq:linear3_ach} is obtained similarly using \eqref{eq:final_ach}, except for the term $\alpha^*$ in the numerator.  To see how this arises, we consider an arbitrary value of $\alpha\in(\alpha^*,1]$ and set $\ell = \lfloor \alpha k\rfloor$; the case $\alpha=\alpha^*$ follows by continuity.   The term $\log {{p-k+\ell} \choose \ell}$ is handled in the same way as the term $\log {{p-k} \choose \ell}$ above, so we focus on the term $\log\sum_{d=0}^{\dmax}{{p-k} \choose d}{\ell \choose d}$.  This is upper bounded by $\max_{d=0,\dotsc,\dmax}\log\big((1+\dmax){{p-k} \choose d}{k \choose d}\big)$.  Similarly to the achievability part, we can factor $\log\big((1+\dmax)\log {k \choose d}\big)$ into $\eta$, so we are left with $\log {{p-k} \choose \dmax}$.  Approximating this using Stirling's approximation as before, and recalling that $\dmax = \lfloor \alpha^* k \rfloor$, we obtain the desired term $\alpha^{*}k\log\frac{p}{k}$.
\end{proof}

While the achievability and converse bounds in Corollary \ref{cor:linear_partial} do not have the same constants, the two are similar, and always have the same scaling laws.  In the limit as $c_{\beta} \to \infty$, we have $\frac{1}{2}\log\big(1 + \frac{c_{\beta}}{\sigma^2} g(\alpha) \big) = \frac{1}{2}(\log c_{\beta})(1+o(1))$; in this case, the maxima in \eqref{eq:linear3_ach}--\eqref{eq:linear3_conv} are both achieved with $\alpha\to1$, and hence, the two bounds coincide to within a multiplicative factor of $\frac{1}{1-\alpha^*}$.

Corollary \ref{cor:linear_partial} is related to the setting studied by Reeves and Gastpar \cite{Ree12,Ree13}, but considers $k = o(p)$ instead of $k = \Theta(p)$.  Despite this difference, it is instructive to compare the bounds upon letting the implied constant in the $\Theta(p)$ scaling tend to zero.  A careful comparison reveals that the converse bounds coincide in this limit, whereas our achievability bound is slightly better, in that the analogous bound in \cite{Ree12} multiplies $\frac{c_{\beta}}{\sigma^2} g(\alpha)$ by $(\sqrt{2}-1)^2 \approx 0.17$; see \cite[Eq.~(21)]{Ree12} and \cite[Eq.~(25)]{Ree13}.

In Section \ref{sec:NUMERICAL_PARTIAL}, we present some numerical results for this setting.

\subsection{1-bit Model with Discrete $\beta_S$}  \label{sec:1BIT_DISC}

We now turn to the quantized counterpart of \eqref{eq:linear_model}:
\begin{equation}
    Y = \sign\big(\langle X_S, \beta_S \rangle + Z\big).
\end{equation}
As in Section \ref{sec:GAUSSIAN_DISC}, we fix $s=\{1,\dotsc,k\}$ and let $\beta_s$ be a uniformly random permutation of a fixed vector $(b_1,\dotsc,b_k)$, and we set $P_{X} \sim N(0,1)$.   We again write the minimum and maximum absolute values of $\{b_i\}_{i=1}^k$ as $\bmin$ and $\bmax$.  

The following proposition gives the required characterizations on the mutual information terms and the corresponding variance terms.  Recall the binary entropy function $H_2(\cdot)$ and the Q-function $Q(\cdot)$ defined in Section \ref{sec:NOTATION}.

\begin{prop} \label{prop:1bit_moments}
    Under the preceding setup for the 1-bit model, we have the following:
    
    (i) The mutual information $I_{\sdif,\seq}(b_s)$ is given by
    \begin{multline}
       I_{\sdif,\seq}(b_s) =  \EE\bigg[ H_2\bigg( Q\bigg( W \sqrt{ \frac{\sum_{i\in\seq}b_i^2}{\sigma^2 + \sum_{i\in\sdif}b_i^2} } \bigg) \bigg) \\  - H_2\bigg( Q\bigg( W \sqrt{\frac{1}{\sigma^2}\sum_{i \in s} b_i^2} \bigg) \bigg) \bigg], \label{eq:1bit_Ilv}
   \end{multline}
    where $W \sim N(0,1)$.
    
    (ii) If $k=\Theta(1)$, $\sigma^2=\Theta(1)$, $\bmin = \Theta(\bmax)$, and $\bmin^2 = o(1)$, then 
    \begin{equation}
        I_{\sdif,\seq}(b_s) = \bigg(\frac{1}{\pi\sigma^2} \sum_{i\in\sdif}b_i^2\bigg)\big(1+o(1)\big). \label{eq:1bit_Ilv_asymp}
    \end{equation}
    
    (iii) If $k=\Theta(p)$, $\sigma^2 = \Theta(1)$, and the entries of $b_s$ all equal a common value $b_0$ such that $b_0^2 = \Theta\big(\frac{\log p}{p}\big)$, then the mutual information quantities $I_{\sdif,\seq}(b_s)$ with $|\sdif|=1$ all equal a common value $I_1$ satisfying
    \begin{align}
        I_1 &= \frac{1}{2} \frac{\frac{b_0^2}{\sigma^2}}{\sqrt{2\pi k\frac{b_0^2}{\sigma^2}}} \EE\Big[ W \log\frac{1-Q(W)}{Q(W)} \Big] (1+o(1)) \label{eq:1bit_I1_exact} \\
              &= \Theta\bigg( \frac{\sqrt{\log p}}{p} \bigg), \label{eq:1bit_I1}
    \end{align}
    where $W \sim N(0,1)$.
    
    (iv) The variance $V_{\sdif,\seq}(b_s)$ defined in \eqref{eq:Vb} satisfies
    \begin{multline}
        V_{\sdif,\seq}(b_s) \le c_0 \Bigg( \frac{1}{\sigma^2}\sum_{i\in\sdif}b_i^2 +
        \bigg(\frac{1}{\sigma^2}\sum_{i\in\sdif}b_i^2\bigg)^2 \\ +
        \min\bigg\{1,\bigg(\frac{1}{\sigma^2}\sum_{i\in\sdif}b_i^2\bigg)^2\bigg\}\frac{1}{\sigma^2}\sum_{i\in\seq}b_i^2
        \Bigg) \label{eq:1bit_Vlv}
    \end{multline}
    for some universal constant $c_0$.
\end{prop}
\begin{proof}
    See Appendix \ref{sec:1BIT_PROOFS}.
\end{proof}

Below we present two corollaries corresponding to different scalings of $k$ and the SNR, namely, those given in parts (ii) and (iii) of Proposition \ref{prop:1bit_moments}.  We proceed by simultaneously presenting the steps of Procedure \ref{pr:procedure} for both settings.

\subsubsection*{Step 1}

As in Section \ref{sec:GAUSSIAN_DISC}, we choose the trivial typical set $\Tc_{\beta}$ containing all vectors on the support of $P_{\beta_S}$.

\subsubsection*{Step 2}

We make use of Chebyshev's inequality in Proposition \ref{prop:chebyshev} in Appendix \ref{sec:CONCENTRATION}. Choosing $\delta = \delta_2 I_{\sdif,\seq}(b_s)$ in \eqref{eq:conc_linear1}, it follows that we may set
\begin{multline}
    \psi_{\ell}(n,\delta_2) = \psi'_{\ell}(n,\delta_2) \\ = \max_{(\sdif,\seq,b_s)\,:\,|\sdif|=\ell} \frac{V_{\sdif,\seq}(b_s)}{n\delta_2^2 I_{\sdif,\seq}(b_s)^2}. \label{eq:1bit_cheby}
\end{multline}

\subsubsection*{Step 3}

We again choose $\gamma$ as in \eqref{eq:discrete_gamma} so that $P_0(\gamma) = 0$. Consider the setting described in part (ii) of Proposition \ref{prop:1bit_moments}.  Under the scalings therein, \eqref{eq:1bit_Ilv_asymp} and \eqref{eq:1bit_Vlv} both behave as $\Theta( \bmin^2 )$.  Hence, and using \eqref{eq:1bit_cheby} and the fact that $k=\Theta(1)$, the second term in \eqref{eq:nonasymp_ach} vanishes provided that 
\begin{equation}
    n=\omega\Big( \frac{1}{\bmin^2} \Big). \label{eq:1bit_n_cond_a}
\end{equation}
The setting described in part (iii) of Proposition \ref{prop:1bit_moments} is handled similarly.  We set $\Lc=\{1\}$ in Theorem \ref{thm:simplified_conv}, thus focusing only on $\ell := |\sdif| = 1$.  Denoting the corresponding variance $V_{\sdif,\seq}(b_s)$ by $V_1$, it follows by substituting the scalings of $k$, $\sigma^2$ and $b_0^2$ into \eqref{eq:1bit_Vlv} that $V_1 = O\big( \frac{\log p}{p} \big)$. It thus follows from \eqref{eq:1bit_I1} and \eqref{eq:1bit_cheby} that $\psi'_{1}(n,\delta_2)$ vanishes provided that 
\begin{equation}
    n=\omega(p). \label{eq:1bit_n_cond_b}
\end{equation}

\subsubsection*{Step 4}

Combining the above steps and applying asymptotic simplifications, we obtain the following corollaries.
 
\begin{cor} \label{cor:1bit_fixedK}
    Under the preceding setup for the 1-bit model with $k=\Theta(1)$, $\sigma^2 = \Theta(1)$, $\bmin = \Theta(\bmax)$, and $\bmin = o(1)$, we have $\pe \to 0$ as $p\to\infty$ provided that
    \begin{equation}
    n \ge \max_{\sdif\ne\emptyset} \frac{ |\sdif| \log p }{ \frac{1}{\pi\sigma^2}\sum_{i \in \sdif}b_i^2} (1 + \eta) \label{eq:1bit_n_cond1}
    \end{equation}
    for some $\eta > 0$. Conversely, $\pe \to 1$ as $p\to\infty$ whenever
    \begin{equation}
    n \le \max_{\sdif\ne\emptyset} \frac{ |\sdif| \log p }{ \frac{1}{\pi\sigma^2}\sum_{i \in \sdif}b_i^2} (1 - \eta)  \label{eq:1bit_n_cond2}
    \end{equation}
    for some $\eta > 0$. 
\end{cor}
\begin{proof}
     We obtain \eqref{eq:1bit_n_cond1} and \eqref{eq:1bit_n_cond2} from \eqref{eq:final_ach} and \eqref{eq:final_conv} respectively.  The denominators are obtained directly from part  (ii) of Proposition \ref{prop:1bit_moments}, and the numerators follow from the identity $\log{p \choose |\sdif|} = (|\sdif| \log p)(1+o(1))$, which holds whenever $k=\Theta(1)$ and hence $|\sdif|=\Theta(1)$.  By the assumption $k=\Theta(1)$, the remaining terms in \eqref{eq:final_ach} (including the choice of $\gamma$ in \eqref{eq:discrete_gamma}) can be factored into $\eta$.  The condition in \eqref{eq:1bit_n_cond_a} is implied by \eqref{eq:1bit_n_cond1} (or by \eqref{eq:1bit_n_cond2} when equality holds) by the same argument as Corollary \ref{cor:linear_genK}.
\end{proof}

\begin{cor} \label{cor:1bit_genK}
    Under the preceding setup for the 1-bit model with $k=\Theta(p)$, $\sigma^2 = \Theta(1)$, and the entries of $\beta_s$ deterministically equaling a common value $b_0$ such that $b_0^2 = \Theta\big(\frac{\log p}{p}\big)$, we have $\pe \to 1$ provided that
    \begin{align}
        n &\le \frac{\log p}{ \frac{1}{2} \frac{\frac{b_0^2}{\sigma^2}}{\sqrt{2\pi k\frac{b_0^2}{\sigma^2}}} \EE\Big[ W \log\frac{1-Q(W)}{Q(W)} \Big] } (1-\eta) \label{eq:1bit_n_cond3} \\
            &= \Theta\big( p \sqrt{\log p} \big)
    \end{align}
    for some $\eta\in(0,1)$, where $W \sim N(0,1)$.
\end{cor}
\begin{proof}    
    The condition in \eqref{eq:1bit_n_cond3} follows using \eqref{eq:final_conv} with $|\sdif| = 1$; the numerator behaves as $(\log p) (1+o(1))$, and the denominator behaves according to \eqref{eq:1bit_I1_exact}.  The additional condition in  \eqref{eq:1bit_n_cond_b} is satisfied when \eqref{eq:1bit_n_cond3} holds with equality.
\end{proof}

In the same way as \eqref{eq:linear2_ach}--\eqref{eq:linear2_conv}, one can compute \eqref{eq:1bit_n_cond1}--\eqref{eq:1bit_n_cond2} without evaluating all $2^k - 1$ objective values; for a given value of $|\sdif|$, the maximum is achieved by the set $\sdif$ with the smallest value of $\sum_{i \in \sdif}b_i^2$.

The asymptotic identities used in the proof of Corollary \ref{cor:1bit_fixedK} can directly be applied to \eqref{eq:linear2_ach}--\eqref{eq:linear2_conv} with $k=\Theta(1)$ and $\bmin = o(1)$, and the resulting expressions are precisely those in \eqref{eq:1bit_n_cond1}--\eqref{eq:1bit_n_cond2} with $\frac{1}{\pi}$ replaced by $\frac{1}{2}$.  Thus, this is a case where there is only a minor loss in the performance due to the quantization; the corresponding asymptotic number of measurements only increases by a factor of $\frac{\pi}{2} \approx 1.57$.

In contrast, Corollary \ref{thm:simplified_conv} describes a setting where the linear model and its 1-bit counterpart lead to significantly different requirements on the number of measurements. Under the scaling described therein, the necessary and sufficient number of measurements for the linear model behaves as $\Theta(p)$ \cite[Table I]{Rad11}.  Thus, the 1-bit quantization increases the required number of measurements from linear to super-linear in the ambient dimension.

\subsection{1-bit Model with Gaussian $\beta_S$ and Partial Recovery} \label{sec:1BIT_CONT}
 
 We now consider the 1-bit counterpart of the setting studied in Section \ref{sec:GAUSSIAN_CONT}, where $\beta_s$ is i.i.d.~on $N(0,\sigbeta^2)$ for some $\sigbeta^2 = \frac{c_{\beta}}{k}$, and we seek partial recovery as in \eqref{eq:pe_partial}  with $\dmax = \lfloor \alpha^* k \rfloor$.  We make use of the following.
 
 \begin{prop} \label{prop:1bit_boundMI}
     Under the preceding setup for the 1-bit model, the quantity $I_{0,+}$ in \eqref{eq:I_0+} satisfies
     \begin{equation}
         I_{0,+} \le \frac{k}{2} \log\Big( 1 + \frac{n\sigbeta^2}{\sigma^2} \Big) + \sqrt{ k\log\Big( 1 + \frac{n\sigbeta^2}{\sigma^2} \Big)  }. \label{eq:1bit_bmi_I0+}
     \end{equation}
     for some universal constant $c'_0$.
\end{prop}
\begin{proof}
    By the data processing inequality, $I_0$ must satisfy \eqref{eq:bmi_I0} even in the 1-bit setting.  We immediately obtain \eqref{eq:1bit_bmi_I0+} from the identity $I_{0,+} \le I_0 + \sqrt{2I_0}$ given in \cite{Bar00}.
\end{proof}

We now turn to the steps for providing a counterpart to Corollary \ref{cor:linear_partial}.  We define the function
\begin{multline}
    \Psi(\alpha,c_{\beta},\sigma) := \EE\bigg[ H_2\bigg( Q\bigg( W \sqrt{\frac{c_{\beta} (1-g(\alpha)) }{ \sigma^2 + c_{\beta} g(\alpha)}} \bigg) \bigg)  \\ - H_2\bigg( Q\bigg( W \sqrt{\frac{c_{\beta}}{\sigma^2}} \bigg) \bigg) \bigg], \label{eq:1bit_Psi}
\end{multline}
where $W \sim N(0,1)$, and $g(\alpha)$ is defined in \eqref{eq:g_alpha}.

\subsubsection*{Step 1}

We choose the same typical set $\Tc_{\beta}$ as that in Section \ref{sec:GAUSSIAN_CONT}, thus ensuring that \eqref{eq:empiricalB} holds for all sequences of typical vectors.  It follows that $\sum_{i\in\sdif}b_i^2 \to c_{\beta}g(\alpha)$ and $\sum_{i\in\seq}b_i^2 \to c_{\beta}(1-g(\alpha))$ for the pair $(\sdif,\seq)$ with corresponding sizes $(\ell,k-\ell)$ ($\ell = \lfloor \alpha k \rfloor$) such that $\sum_{i\in\sdif}b_i^2$ is minimized.  We observe from \eqref{eq:1bit_Ilv} that minimizing $\sum_{i\in\sdif}b_i^2$ also amounts to minimizing $I_{\sdif,\seq}(b_s)$ for a fixed value of $|\sdif|$, as was the case for the linear model.  If $\frac{|\sdif|}{k}$ converges to a given constant $\alpha$, then the corresponding mutual information converges as follows, in accordance with \eqref{eq:1bit_Ilv} and \eqref{eq:1bit_Psi}:
\begin{equation}
    I_{\sdif,\seq}(b_s) \to \Psi(\alpha,c_{\beta},\sigma). \label{eq:1bit_I_asymp}
\end{equation}

\subsubsection*{Step 2}

We make use of the general concentration inequality given in Proposition \ref{prop:gen_discrete} in Appendix \ref{sec:CONCENTRATION}; setting $\delta = \delta_2 I_{\sdif,\seq}(b_s)$ in \eqref{eq:conc_gen_disc} in Appendix \ref{sec:CONCENTRATION} gives
\begin{align}
    & \psi_{\ell}(n,\delta_2) - \psi'_{\ell}(n,\delta_2) \nonumber \\ 
    & \hspace*{-0.5ex}= \max_{(\sdif,\seq,b_s)\,:\,|\sdif|=\ell} 2\exp\bigg(- \frac{(\delta_2 I_{\sdif,\seq}(b_s))^2 n}{2(8|\Yc|+2\delta_2 I_{\sdif,\seq}(b_s))} \bigg). \label{eq:1bit_gen_disc}
\end{align}

\subsubsection*{Step 3}

We choose $\gamma = \frac{I_{0,+}}{\delta_0}$ as in \eqref{eq:gamma_markov}, ensuring that $P_0(\gamma) \le \delta_0$. The other remainder terms are controlled in the same way as Section \ref{sec:GAUSSIAN_CONT}.  We again use the fact that the typical realizations $b_s$ of $\beta_s$ satisfy \eqref{eq:empiricalB}, and yield $\sum_{i\in\sdif} b_i^2 = \Omega(1)$, and hence $\alpha_{\sdif}^2 = \Theta(1)$ in \eqref{eq:alpha_linear}. We also have $I_{\sdif,\seq}(b_s) = \Theta(1)$ in \eqref{eq:1bit_Ilv}; this is seen by noting that the smallest mutual information for a fixed $|\sdif| = \lfloor \alpha k \rfloor$ satisfies \eqref{eq:1bit_I_asymp}, and the mutual information upper bounded by $\log2$ since the observations are binary.  It follows that the exponent in \eqref{eq:1bit_gen_disc} behaves as $\Theta(n)$; hence, following the arguments in Section \ref{sec:GAUSSIAN_CONT}, we conclude that the second term in \eqref{eq:nonasymp_ach} vanishes provided that 
\begin{equation}
    n = \Omega(k) \label{eq:1bit_n_cond_c}
\end{equation}
with a sufficiently large implied constant.  Once again, for the converse part, one can analogously show that the weaker condition $n=\omega(1)$ suffices.

\subsubsection*{Step 4}

Combining the above steps, we get the following.

\begin{cor} \label{cor:1bit_partial}
    Under the preceding setup for the 1-bit model with $k\to\infty$ and $k= o(p)$, $\sigma^2 = \Theta(1)$, $\sigbeta^2 = \frac{c_{\beta}}{k}$ for some $c_{\beta} > 0$,  and $\dmax = \lfloor \alpha^* k \rfloor$ for some $\alpha^*\in(0,1)$, we have $\pe(\dmax) \to 0$ as $p\to\infty$ provided that
    \begin{equation}
    n \ge \max_{\alpha\in[\alpha^*,1]} \frac{ \alpha k \log \frac{p}{k} }{ \Psi(\alpha,c_{\beta},\sigma) } (1 + \eta) \label{eq:1bit_ach2}
    \end{equation}
    for some $\eta > 0$, where $\Psi$ is defined in \eqref{eq:1bit_Psi}. Conversely, $\pe(\dmax) \to 1$ as $p\to\infty$ whenever
    \begin{equation}
    n \le \max_{\alpha\in[\alpha^*,1]} \frac{ (\alpha - \alpha^*) k \log \frac{p}{k} }{ \Psi(\alpha,c_{\beta},\sigma) } (1 - \eta) \label{eq:1bit_conv2}
    \end{equation}
    for some $\eta > 0 $.  
\end{cor}
\begin{proof}
    As usual, we begin with the conditions in \eqref{eq:final_ach} and \eqref{eq:final_conv}.  The denominators in \eqref{eq:1bit_ach2}--\eqref{eq:1bit_conv2} follow directly by applying \eqref{eq:1bit_I_asymp}. Moreover, the terms $\alpha k\log\frac{p}{k}$ and $(\alpha - \alpha^*) k\log\frac{p}{k}$ in the numerators are obtained in an identical fashion to Corollary \ref{cor:linear_partial} once we show that there exists a vanishing sequence of constants $\delta_0$, indexed by $p$, such that the remainder term $\gamma = \frac{I_{0,+}}{\delta_0}$ resulting from \eqref{eq:gamma_markov} can be factored into $\eta$.  To see this, we note that the right-hand side of \eqref{eq:1bit_ach2} behaves as $\Theta(k \log p)$, whereas from Proposition \ref{prop:1bit_boundMI} (with the scalings $n=\Theta(k\log p)$ and $\sigbeta^2 = \Theta\big(\frac{1}{k}\big)$),  $I_{0,+}$ behaves as $O\big(k\log\log p\big)$.  We may thus set $\delta_0$ to be (for example) $\frac{\log\log p}{\sqrt{\log p}}$.   Finally, we observe that \eqref{eq:1bit_n_cond_c} holds whenever \eqref{eq:1bit_ach2} holds, and similarly for the converse part.
\end{proof}

The main difference in \eqref{eq:1bit_ach2}--\eqref{eq:1bit_conv2} compared to the linear counterparts in \eqref{eq:linear3_ach}--\eqref{eq:linear3_conv} is the behavior in the limit as $c_{\beta} := k\sigbeta^2 \to \infty$.  As stated following Corollary \ref{cor:linear_partial}, the denominator in the linear setting behaves as $(\log c_{\beta}) (1+o(1))$, thus tending towards infinity.  In contrast, for the 1-bit setting, we have $\Psi \le \log 2$ due to the fact that $H_2(\cdot) \in [0,\log 2]$, and thus the denominator cannot grow unbounded.  These observations are consistent with Corollary \ref{cor:1bit_genK}, which shows that 1-bit CS can require significantly more measurements compared to the linear setting when the signal-to-noise ratio (SNR) is sufficiently high.  

\subsection{Numerical Evaluations for Partial Recovery} \label{sec:NUMERICAL_PARTIAL}
 
 In this subsection, we present numerical calculations for the settings considered in Sections \ref{sec:GAUSSIAN_CONT} and \ref{sec:1BIT_CONT}.  We set $\alpha^* = 0.1$, $\sigma^2 = 1$,  and $k = o(p)$.  We consider values of $\sigbeta^2$ of the form $\sigbeta^2 = \frac{c_{\beta}}{k}$ for fixed $c_{\beta}$.  Similarly to \cite{Ree12,Ree13}, we present our results in terms of
 \begin{equation}
     \SNRdB := 10\log\frac{k\sigbeta^2}{\sigma^2} = 10\log c_{\beta}, \label{eq:SNRdB}
 \end{equation}
 which represents the per-sample SNR in dB.
 
  \begin{figure}
      \begin{centering}
          \includegraphics[width=0.95\columnwidth]{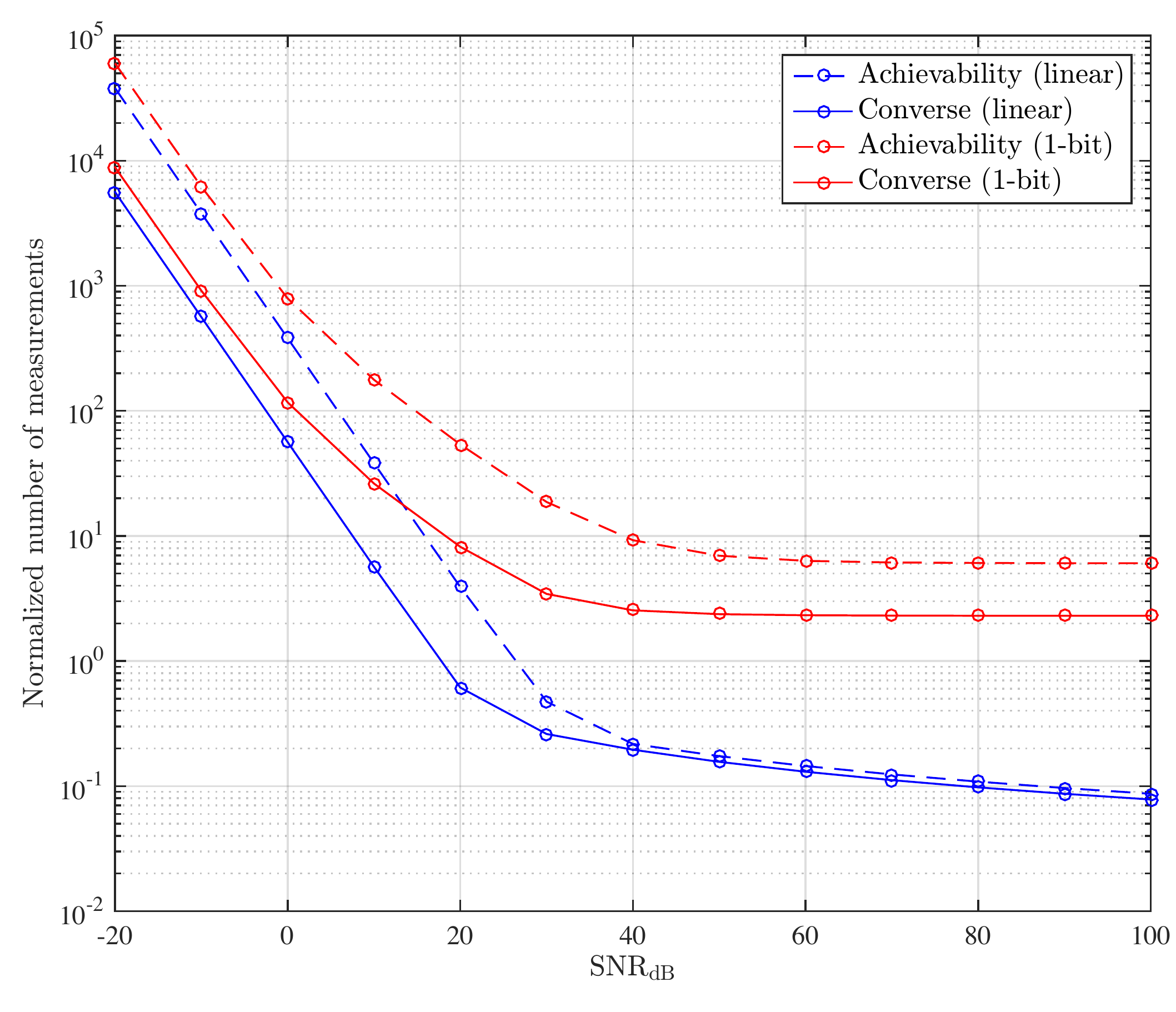}
          \par
      \end{centering}
   
      \caption{ Asymptotic thresholds on the number of measurements required for partial support recovery for the linear and 1-bit models, with $\alpha^* = 0.1$.  The number of measurements is normalized by  $k\log\frac{p}{k}$, and $\SNRdB$ is defined in \eqref{eq:SNRdB}. \label{fig:partial_recovery}}
  \end{figure}
 
 Figure \ref{fig:partial_recovery} plots the asymptotic thresholds on the number of measurements from Corollaries \ref{cor:linear_partial} and \ref{cor:1bit_partial}.  For both the linear and 1-bit settings, there is a close correspondence between the necessary and sufficient number of measurements.  The bounds for the two models nearly coincide at low SNR, which is consistent with Corollary \ref{cor:1bit_fixedK}.
 
 The behavior of the bounds at high SNRs is also consistent with our previous discussions.  In the linear setting, the ratio between the bounds narrows to approximately $1.11$ as the SNR grows large, which coincides with the value $\frac{1}{1-\alpha^*}$ given in the discussion following Corollary \ref{cor:linear_partial}.  Moreover, as discussed following Corollary \ref{cor:1bit_partial}, the number of measurements steadily decreases for increasing SNRs for the linear model, while saturating at an asymptotic limit for the 1-bit model.

\subsection{Group Testing} \label{sec:GROUP_TESTING}

\subsubsection{Noiseless Case with Exact Recovery}

Here we consider the noiseless group testing problem, where each observation is deterministically generated according to
\begin{equation}
Y = \openone\Big\{ \bigcup_{i \in S} \{X_i = 1\} \Big\}. \label{eq:gt_noiseless}
\end{equation}
We consider Bernoulli measurement matrices with $P_X(1) = 1 - P_X(0) = \frac{\nu}{k}$, where $\nu$ is a constant not depending on $p$.  Here there is no latent variable $\beta_s$, which can equivalently be thought of as corresponding to $\beta_s$ equaling the vector of ones deterministically.  This implies that $I_{\sdif,\seq}(b_s)$ depends only on $\ell = |\sdif|$, and we emphasize this by writing it as $I_{\ell}$.  Our setting readily handles both fixed and growing $k$; since the former is already well-understood \cite{Mal78,Ati12,Laa14}, we focus our attention on the case that $k\to\infty$, and in particular on the case that $k = \Theta(p^\theta)$ for some $\theta\in(0,1)$.

\begin{prop} \label{prop:gt_boundI}
    Under the noiseless group testing setup, consider arbitrary sequences of sparsity levels $k\to\infty$ and $\ell \in \{1,\dotsc,k\}$, indexed by $p$.  If $\frac{\ell}{k} = o(1)$, then 
    \begin{equation}
         I_{\ell} = \bigg(e^{-\nu}\nu \frac{\ell}{k} \log\frac{k}{\ell}\bigg)(1+o(1)). \label{eq:gt_I_ord}
    \end{equation}
    Moreover, if $\frac{\ell}{k}\to\alpha \in(0,1]$, then
    \begin{equation}
        I_{\ell} = e^{-(1-\alpha)\nu} H_2\big(e^{-\alpha \nu}\big) (1+o(1)). \label{eq:gt_I_const}
   \end{equation}
\end{prop}
\begin{proof}
    See Appendix \ref{sec:PROOFS_GT}.
\end{proof}

We proceed with the steps of Procedure \ref{pr:procedure}.

\subsubsection*{Step 1}

The first step is trivial; $\beta_s$ is deterministic, and thus the typical set $\Tc_{\beta}$ is a singleton.

\subsubsection*{Step 2}

In contrast to the previous examples, we use different concentration inequalities to handle different values of $\ell$.  Moreover, in accordance with Remark \ref{rem:delta2}, we let $\delta_2$ depend on $\ell$, writing it as $\delta_{2,\ell}$.  For ``large'' values of $\ell$ (to be made precise below), we will apply the general bound in Proposition \ref{prop:gen_discrete} in Appendix \ref{sec:CONCENTRATION}.  For ``small'' values of $\ell$, we use the following to obtain an improved bound.

\begin{prop} \label{prop:conc_gt}
    For the noiseless group testing problem, consider sequences $k\to\infty$ and $\ell$, indexed by $p$, such that $\frac{\ell}{k} \to 0$.  For any $\epsilon > 0$ and $\delta_{2,\ell} \in (0,1)$ bounded away from zero and one, the following holds for sufficiently large $p$:
    \begin{multline}
    \PP\Big[ \imath^n(\Xv_{\sdif}; \Yv | \Xv_{\seq}, b_s) \le nI_{\ell}(1-\delta_{2,\ell}) \Big] \\ \le \exp\bigg(-n\frac{\ell}{k} e^{-\nu}\nu\Big( (1-\delta_{2,\ell})\log(1-\delta_{2,\ell}) + \delta_{2,\ell} \Big)(1-\epsilon)\bigg) \label{eq:conc_gt}
    \end{multline}
    for all $(\sdif,\seq)$ with $|\sdif|=\ell$.
\end{prop}
\begin{proof}
    See Appendix \ref{sec:PROOFS_GT}.
\end{proof}

From the bounds in \eqref{eq:conc_gt} and \eqref{eq:conc_gen_disc} in Appendix \ref{sec:CONCENTRATION}, we may fix $\epsilon > 0$ and choose the following when $p$ is sufficiently large:
\begin{itemize}
    \item For $\ell \le \ell \le \lfloor\frac{k}{\log k}\rfloor$:
        \begin{multline}
            \hspace*{-6ex}\psi_\ell(n,\delta_{2,\ell}) = \\ \hspace*{-3ex} \exp\bigg(-n\frac{\ell}{k} e^{-\nu}\nu\Big( (1-\delta_{2,\ell})\log(1-\delta_{2,\ell}) + \delta_{2,\ell} \Big)(1-\epsilon)\bigg). \label{eq:gt_psi1}
        \end{multline}
    \item For $\ell > \lfloor\frac{k}{\log k}\rfloor$:
\end{itemize}
\begin{equation}
    \psi_\ell(n,\delta_{2,\ell}) = 2\exp\bigg(- \frac{(\delta_{2,\ell} I_{\ell})^2 n}{2(16+2\delta_{2,\ell} I_{\ell})} \bigg).  \label{eq:gt_psi2}
\end{equation}
For the converse, we only use the latter of these two cases, setting $\psi'_{\ell}(n,\delta_{2,\ell}) = 2\exp\big(- \frac{(\delta_{2,\ell} I_{\ell})^2 n}{2(16+2\delta_{2,\ell} I_{\ell})} \big)$.

\subsubsection*{Step 3}

Since $\beta_s$ is deterministic, we may trivially set $\gamma=0$ to obtain $P_0(\gamma) = 0$ in \eqref{eq:P_0}.

For the converse, we set $\Lc = \{k\}$ in Theorem \ref{thm:simplified_conv}.  From the above choice of $\psi'_{\ell}$ and the growth of $I_k$ in \eqref{eq:gt_I_const}, we immediately obtain that $\psi'_{k} \to 0$ whenever $n=\omega(1)$.  The achievability part requires more effort; we summarize the findings in the following proposition.

\begin{prop} \label{prop:gt_psi}
    Let $k = \Theta( p^{\theta} )$ for some $\theta \in (0,1)$. 
    
    (i) For any $\eta > 0$, there exists $\delta_2^{(1)} \in (0,1)$ and a choice of $\epsilon > 0$ in \eqref{eq:gt_psi1} such that $\sum_{\ell=1}^{\lfloor \frac{k}{\log k} \rfloor} {k \choose \ell} \psi_{\ell}(n,\delta_{2}^{(1)}) \to 0$ provided that 
    \begin{equation}
        n \ge \frac{ \frac{\theta}{1-\theta} k\log\frac{p}{k} }{ e^{-\nu}\nu } (1+\eta). \label{eq:gt_n_cond_psi}
    \end{equation}
    (ii) For any $\delta_2^{(2)} \in (0,1)$, we have $\sum_{\lfloor \frac{k}{\log k} \rfloor +1}^k {k \choose \ell} \psi_{\ell}(n,\delta_{2}^{(2)}) \to 0$ provided $n = \Omega\big(k\log\frac{p}{k}\big)$.
\end{prop}
\begin{proof}
    See Appendix \ref{sec:PROOFS_GT}.
\end{proof}

The idea here is that for the smaller values of $\ell$, it is the concentration inequality that dominates the final bound, so we let $\delta_{2,\ell} = \delta_2^{(1)}$ be closer to one to provide better concentration behavior.  For large values of $\ell$, the opposite is true, so we let $\delta_{2,\ell} = \delta_2^{(2)}$ be close to zero.

\subsubsection*{Step 4}

We obtain the following corollary by combining the previous steps and applying asymptotic simplifications.

\begin{cor} \label{cor:gt}
    For the noiseless group testing problem with $k=\Theta(p^{\theta})$ ($\theta\in(0,1)$) and an optimized parameter $\nu$, we have $\pe \to 0$ as $p\to\infty$ provided that
    \begin{equation}
        n \ge \inf_{\nu>0}\max\bigg\{ \frac{ \theta }{ e^{-\nu}\nu(1-\theta) }, \frac{ 1 }{ H_2(e^{-\nu}) }\bigg\} \Big(k \log{\frac{p}{k}}\Big) (1 + \eta)   \label{eq:gt_ach}
    \end{equation}
    for some $\eta > 0$.  Conversely, we have $\pe \to 1$ as $p\to\infty$ whenever
    \begin{equation}
    n \le \frac{ k \log{\frac{p}{k}} }{ \log 2 }  (1 - \eta) \label{eq:gt_conv}
    \end{equation}
    for some $\eta > 0 $. 
\end{cor}
\begin{proof}
    We first consider the achievability part.  We immediately obtain the first term in the maximum in \eqref{eq:gt_ach} from \eqref{eq:gt_n_cond_psi}, so it remains to derive the second term.  We start with \eqref{eq:final_ach}; by substituting $\gamma = 0$ and taking $\delta_1 \to 0$ sufficiently slowly, we obtain
    \begin{equation}
        n \ge \max_{\ell=1,\dotsc,k} \frac{ \log{ {{p-k} \choose \ell} } + 2\log\big(k {k \choose \ell}\big) }{I_{\ell}(1 - \delta_{2,\ell})} \big(1+o(1)\big). \label{eq:gt_n_cond3}
    \end{equation}
    Using \eqref{eq:gt_I_ord}--\eqref{eq:gt_I_const} and the asymptotic identity $\log {{p-k} \choose \ell} = \Theta\big( \ell \log \frac{p}{\ell} \big)$ we see that the objective in \eqref{eq:gt_n_cond3} behaves as
    \begin{equation}
        \Theta\bigg( \frac{k \log \frac{p}{\ell} }{ 1 + \log\frac{k}{\ell}} \bigg) \label{eq:gt_growth}
   \end{equation}
   whenever the constants $\{\delta_{2,\ell}\}$ are bounded away from one. This behaves as $\Theta\big(k \log \frac{p}{k}\big)$ when $\frac{\ell}{k} = \Theta(1)$, and as $\Theta\big( \frac{ k \log \frac{p}{k} }{ \log\frac{k}{\ell} } + k \big)$ when $\frac{\ell}{k} = o(1)$ (the latter of these is seen by writing $\log\frac{p}{\ell} = \log\frac{p}{k} + \log\frac{k}{\ell}$).  Thus, the maximum in \eqref{eq:gt_n_cond3} can only be achieved by a sequence such that $\frac{\ell}{k} = \Theta(1)$. Moreover, with $\frac{\ell}{k} = \Theta(1)$, we see from the assumption $k = o(p)$ that the term $2\log\big(k {k \choose \ell}\big) = O(k)$ is dominated by $\log { {p-k} \choose \ell } = \Theta\big( k \log \frac{p}{k} \big)$, and can thus be factored into the $o(1)$ remainder term in \eqref{eq:gt_n_cond3}.  This yields the condition
    \begin{equation}
        n \ge \max_{\ell=1,\dotsc,k} \frac{\ell \log \frac{p}{\ell}}{ I_{\ell}(1 - \delta_{2,\ell}) }\big(1+o(1)\big). \label{gt:growth2}
    \end{equation}
     Since the maximum can only be achieved asymptotically if $\frac{\ell}{k} = \Theta(1)$, we proceed by considering $\frac{\ell}{k} \to \alpha$ for some arbitrary $\alpha\in(0,1]$.  Under this scaling, $\ell\log\frac{p}{\ell}$ behaves as $\big(\alpha k\log\frac{p}{k}\big) (1+o(1))$.  Moreover, according to Proposition \ref{prop:gt_psi}, we can choose $\delta_{2,\ell}$ to be arbitrarily small for all $\ell$ values except those below $\lfloor \frac{k}{\log k} \rfloor$.  Such values behave as $o(k)$, and thus do achieve the maximum in \eqref{gt:growth2}.  Combining these observations with \eqref{eq:gt_I_const}, the right-hand side of \eqref{gt:growth2} yields the condition
    \begin{equation}
       n \ge \max_{\alpha\in(0,1]}\frac{\alpha k\log\frac{p}{k}}{ e^{-(1-\alpha)\nu} H_2\big(e^{-\alpha \nu}\big) } \big(1+\eta\big), \label{gt:growth3}
   \end{equation}
   where $\eta$ may be arbitrarily small.  By a change of variable $\lambda = e^{-\alpha\nu}$, the coefficient to $k\log\frac{p}{k}$ can be written as $\frac{1}{\nu}e^{\nu} \frac{\lambda\log\frac{1}{\lambda}}{H_2(\lambda)}$.  This is easily verified to be decreasing in $\lambda\in[0,1]$, which implies that the maximizing value of $\alpha$ is one, and yields the second term in \eqref{eq:gt_ach}.
    
    The converse part is similar but considerably simpler; by setting $\Lc=\{k\}$ in Theorem \ref{thm:simplified_ach}, we obtain $\alpha=1$ immediately.  The denominator $\log2$ in \eqref{eq:gt_conv} is obtained by maximizing $H_2(e^{-\nu})$ over $\nu$, and the condition $n\to\infty$ stated before Proposition \ref{prop:gt_psi} is clearly satisfied when \eqref{eq:gt_conv} holds with equality.
\end{proof}

By setting $\nu=\log2$ in \eqref{eq:gt_ach}, it is readily verified that the necessary and sufficient conditions coincide for $\theta \le \frac{1}{3}$, and in fact yield the same threshold as \emph{adaptive} group testing \cite{Bal13}.  To our knowledge, this was only known previously in the limit as $\theta\to0$ \cite{Ald14}.  Further comparisons to previous works are provided at the end of this subsection.

\subsubsection{Noisy Case with Exact Recovery} \label{sec:GT_NOISY}

We now turn to the noisy counterpart of \eqref{eq:gt_noiseless}:
\begin{equation}
    Y = \openone\bigg\{ \bigcup_{i \in S} \{X_i = 1\} \bigg\} \oplus Z, \label{eq:gtn_model}
\end{equation}
where $Z\in\{0,1\}$ is additive noise, and $\oplus$ denotes modulo-2 addition.  For concreteness, we focus on the case that $Z \sim \mathrm{Bernoulli}(\rho)$ for some $\rho\in(0,\frac{1}{2})$ not varying with $p$, though other noise models also fall into our framework (e.g., see \cite{Ati12}).  As discussed below, we do not attempt to provide results with constants that are optimized to the same extent as the noiseless case, and we thus set $\nu = \log 2$, i.e., $P_X \sim \mathrm{Bernoulli}\big( \frac{\log 2}{k} \big)$.

We follow Procedure \ref{pr:procedure} in a similar fashion to the noiseless case, altering the statements of Proposition \ref{prop:gt_boundI}--\ref{prop:gt_psi} accordingly.   To avoid repetition, we give the modified propositions and their proofs in Appendix \ref{sec:PROOFS_GT_NOISY}, and state the resulting corollary here.  The main difference is that in the analog of Proposition \ref{prop:gt_psi}, we let $\delta_2^{(1)}$ remain arbitrary, thus leading to the optimization parameter $\delta_2$ in the following.
 
 \begin{cor} \label{cor:gt_noisy}
      Under the preceding setup for the noisy group testing problem with $\rho \in (0,0.5)$, $\nu = \log 2$, and $k=\Theta(p^{\theta})$ ($\theta\in(0,1)$), we have $\pe \to 0$ as $p\to\infty$ provided that
      \begin{multline}
          n \ge \inf_{\delta_2\in(0,1)}\max\bigg\{ \zeta(\rho,\delta_2,\theta) , \frac{ 1 }{\log 2 - H_2(\rho) }\bigg\} \Big( k \log\frac{p}{k} \Big) \\ \times (1 + \eta) \label{eq:gtn_ach}
      \end{multline}
      for some $\eta > 0$, where
      \begin{multline}
          \zeta(\rho,\delta_2,\theta) := \frac{2}{\log 2} \max\bigg\{\frac{ 2(1+\frac{1}{3}\delta_2 (1-2\rho)) \frac{\theta}{1-\theta}}{\delta_2^2 (1-2\rho)^2 } , \\ \frac{\frac{1+4\theta}{1-\theta} }{(1-2\rho)\log\frac{1-\rho}{\rho} (1-\delta_2)} \bigg\}. \label{eq:gt_zeta}
      \end{multline}
        Conversely, we have $\pe \to 1$ as $p\to\infty$ whenever
      \begin{equation}
      n \le \frac{ k \log{\frac{p}{k}} }{ \log 2 - H_2(\rho) }  (1 - \eta). \label{eq:gtn_conv}
      \end{equation}
      for some $\eta > 0 $. 
 \end{cor}
 \begin{proof}
     See Appendix \ref{sec:PROOFS_GT_NOISY}.
 \end{proof}
 
As we will see in the numerical examples below, Corollary \ref{cor:gt_noisy} provides an exact asymptotic threshold for a narrower range of $\theta$ values compared to the noiseless case.  This is due to the difficulty in  precisely characterizing the concentration behavior of the information density tail probabilities.  Nevertheless, the second term in the maximum in \eqref{eq:gtn_ach} is always dominant for sufficiently small $\theta$, thus matching the converse.  To see this, we first note that the first term in the maximum in \eqref{eq:gt_zeta} tends to zero as $\theta \to 0$, and cannot be dominant in this limit.  This implies that $\delta_2$ may be arbitrarily close to zero provided that $\theta$ is sufficiently small.  Assuming then that $\delta_2 $ and $\theta$ are small and the maximum in \eqref{eq:gt_zeta} is achieved by the second term, we can write $\zeta(\rho,\delta_2,\theta) \approx \frac{2}{\log 2}\frac{1}{ (1-2\rho) \log\frac{1-\rho}{\rho} }$.  This is strictly smaller than $\frac{1}{\log 2 - H_2(\rho)}$; see Proposition \ref{prop:gt_cmp} in Appendix \ref{sec:PROOFS_GT_NOISY}.
 
 \subsubsection{Partial Recovery}

The consideration of partial recovery (\emph{cf.}~\eqref{eq:pe_partial}) in fact leads to \emph{simpler} expressions and proofs, as seen in the following.

\begin{cor} \label{cor:gt_partial}
    Under the preceding setup for the group testing problem with $\rho\in[0,0.5)$ (i.e., possibly noiseless), $\nu=\log 2$, $k\to\infty$, $k=o(p)$, and $\dmax = \lfloor \alpha^* k \rfloor$ for some $\alpha^*\in(0,1)$, we have $\pe(\dmax) \to 0$ as $p\to\infty$ provided
    \begin{equation}
        n \ge \frac{ k \log{\frac{p}{k}} }{ \log 2 - H_2(\rho) }  (1 + \eta) \label{eq:gt_ach3}
    \end{equation}
    for some $\eta > 0$. Conversely, $\pe(\dmax) \to 1$ as $p\to\infty$ whenever
    \begin{equation}
        n \le \frac{(1-\alpha^{*}) \big(k \log{\frac{p}{k}}\big)}{ \log 2 - H_2(\rho) }  (1 - \eta) \label{eq:gt_conv3}
    \end{equation}
    for some $\eta > 0 $. 
\end{cor}
\begin{proof}
    The achievability part follows the proofs of Corollaries \ref{cor:gt} and \ref{cor:gt_noisy}, except that the ``small'' values of $\ell$ need not be handled.  That is, we only make use of the general concentration inequality in \eqref{eq:conc_gen_disc} in Appendix \ref{sec:CONCENTRATION}, and we end up with the single condition in \eqref{eq:gt_ach3}.  For the converse part, we again choose $\Lc=\{k\}$ in Theorem \ref{thm:simplified_conv_p}, and the steps are again similar, with the multiplicative factor $1-\alpha^{*}$ arising via identical reasoning to Corollary \ref{cor:linear_partial}.
\end{proof}

Corollary \ref{cor:gt_partial} shows that at least for sufficiently small $\theta$ (e.g., $k = O\big(p^\frac{1}{3}\big)$ in the noiseless case), there is not much to be saved by moving from exact recovery to partial recovery: Allowing for a fraction $\alpha^{*}$ of errors leads to at most a reduction in the number of measurements of a multiplicative factor $1-\alpha^{*}$. 
 
\subsubsection{Numerical Evaluations}
 
In Figure \ref{fig:gt_noiseless}, we compare the bounds in Corollary \ref{cor:gt} with existing asymptotic bounds in the literature.  For convenience, we switch to base-2 logarithms and plot the asymptotic limit of the ratio $\frac{k\log_2\frac{p}{k}}{n}$, so that a higher value corresponds to fewer measurements.  We see that our achievability bound improves on all of the existing bounds; however, we note that the Combinatorial Optimal Matching Pursuit (COMP) \cite{Cha11} and Definite Defective (DD) \cite{Ald14a} algorithms are computationally tractable and do not require knowledge of $k$. 

\begin{figure}
     \begin{centering}
         \includegraphics[width=0.95\columnwidth]{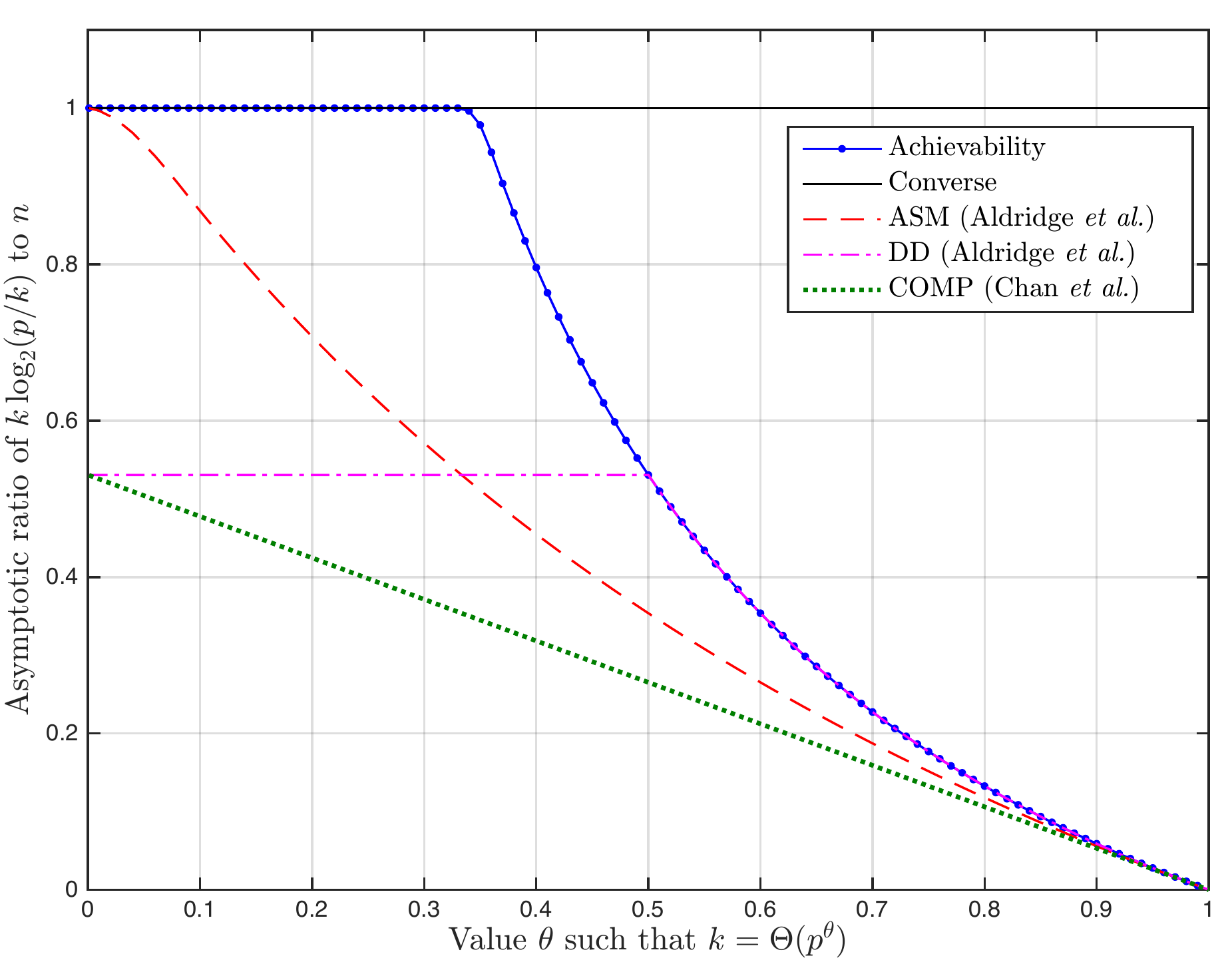}
         \par
     \end{centering}
  
     \caption{ Asymptotic thresholds on the number of measurements required for noiseless group testing. \label{fig:gt_noiseless}}
 \end{figure}
 
 \begin{figure}
      \begin{centering}
       \includegraphics[width=0.95\columnwidth]{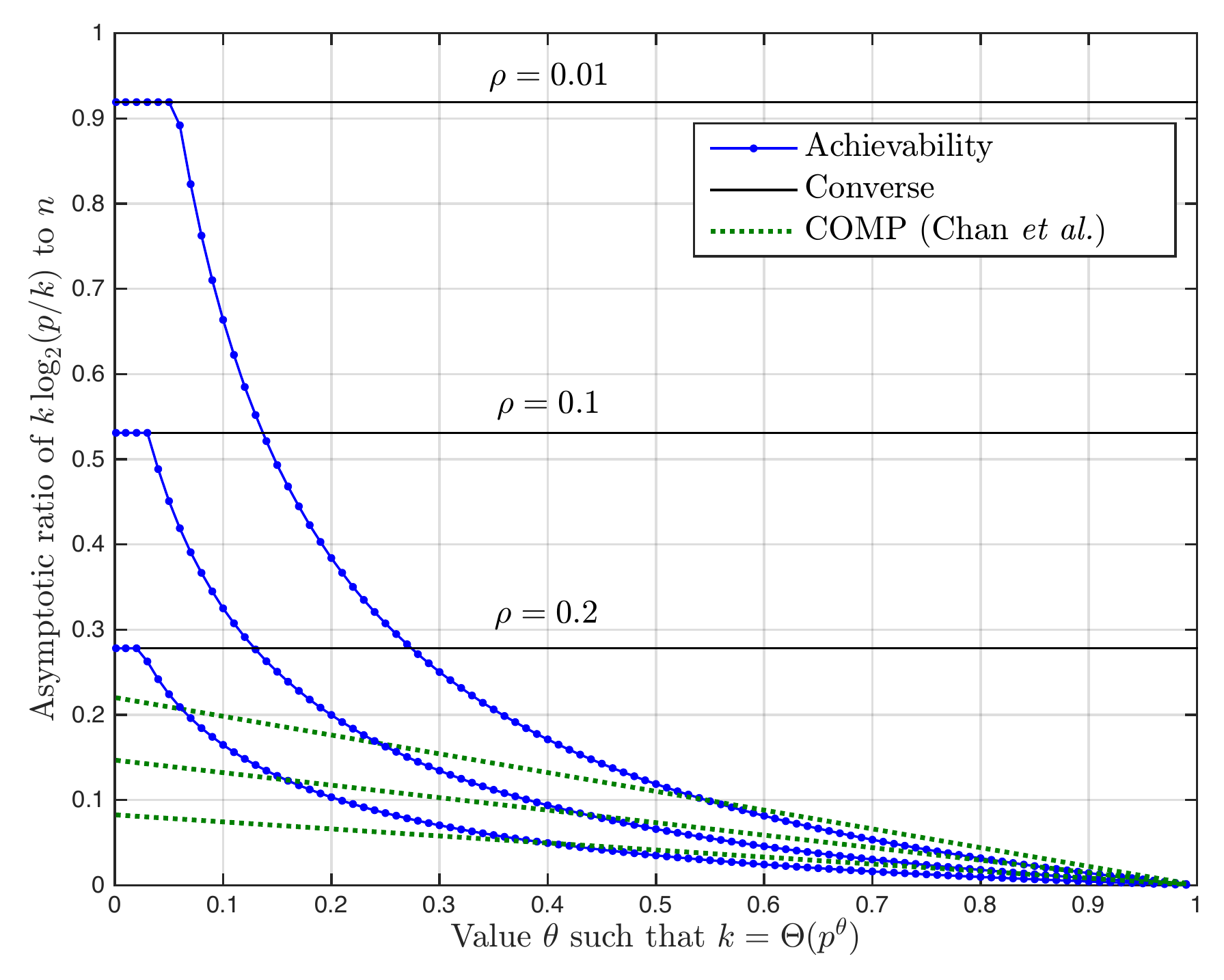}
          \par
      \end{centering}
   
      \caption{ Asymptotic thresholds on the number of measurements required for noisy group testing. \label{fig:gt_noisy}}
\end{figure}

The converse bound shown is known to hold even for adaptive measurement matrices \cite{Bal13}.  Thus, a key implication of our results is that adaptivity provides no asymptotic gain over non-adaptive Bernoulli measurements when $k = O(p^{\frac{1}{3}})$.  It remains an important open problem to derive \emph{practical} decoding schemes for achieving the bound in the non-adaptive setting. 

 Figure \ref{fig:gt_noisy} provides an analogous plot for the noisy case, with three different noise levels (i.e., values of $\rho$).  In each case, we obtain an exact threshold for sufficiently small $\theta$, albeit over a narrower range than the noiseless case.  Once again, the converse is known to hold even in the adaptive setting \cite{Cha11}, and we have thus provided cases where non-adaptive Bernoulli measurements yield the same asymptotics as optimal adaptive measurements.  To our knowledge, this has not been shown previously even in the limit as $\theta \to 0$. 

\subsection{General Strong Impossibility result for Discrete Observation Models} \label{sec:STRONG_IMP_DISC}

Equation \eqref{eq:uniformV} in Appendix \ref{sec:CONCENTRATION} bounds the variance of the information density uniformly in terms of the output alphabet size for models with discrete observations.  Notable examples include group testing, the 1-bit model (or more generally, quantizations with more than two levels), and logistic regression. We obtain the following general strong impossibility result (i.e., conditions under which $\pe \to 1$) by combining Proposition \ref{prop:chebyshev} in Appendix \ref{sec:CONCENTRATION} with a variant of Theorem \ref{thm:simplified_conv}. 

\begin{cor} \label{cor:gen_finite}
	If the observations lie in a finite set $\Yc \subset \RR$ with probability one, then $\pe \to 1$ whenever there exist vanishing sequences $\delta_{1,p} \to 0$ and $\epsilon_{p} \to 0$ such that
	\begin{equation}
		n \ge \max_{(\sdif,\seq)\,:\,\sdif\ne\emptyset} \frac{\log {{p-k+|\sdif|} \choose |\sdif|}  - \log\delta_{1,p} }{ I_{\sdif,\seq}(b_s) + \sqrt{\frac{|\Yc|}{n\epsilon_{p}}} } \label{eq:gen_finite}
	\end{equation}
	for all $b_s \in \RR^k$ within a set whose probability under $P_{\beta_s}$ approaches one.
\end{cor}
\begin{proof}
    In this application, we do not use Theorem \ref{thm:simplified_conv} directly, but instead follow the arguments leading up to it with \eqref{eq:B_cond_conv}--\eqref{eq:psi_conv} replaced by
    \begin{equation}
        \log{{p-k+|\sdif|} \choose |\sdif|} - \log \delta_1 \ge n(I_{\sdif,\seq}(b_s) + \delta); \label{eq:B_cond_conv_2}
    \end{equation}
    and
    \begin{multline}
        \PP\Big[ \imath^n(\Xv_{\sdif}; \Yv | \Xv_{\seq}, b_s) \le n(I_{\sdif,\seq}(b_s) + \delta) \,\big|\, \beta_s = b_s \Big] \\ \ge 1 - \frac{|\Yc|(\frac{4}{e})^2}{\delta^2 n}. \label{eq:psi_conv_2}
    \end{multline}
    By Proposition \ref{prop:chebyshev} in Appendix \ref{sec:CONCENTRATION}, we have for all $(\sdif,\seq,b_s)$ that \eqref{eq:psi_conv_2} holds, so the analogous probability to that on the right-hand side of \eqref{eq:nonasymp_conv} is dictated only by \eqref{eq:B_cond_conv_2}.  Moreover, the right-hand side of \eqref{eq:psi_conv_2} tends to one upon setting $\delta = \sqrt{\frac{|\Yc|}{n\epsilon_p}}$ for some $\epsilon_p \to 0$.  By also setting $\delta_1 = \delta_{1,p} \to 0$ (so that the analogous additive term to that of $\delta_1$ in \eqref{eq:nonasymp_conv} vanishes), we see that \eqref{eq:B_cond_conv_2} coincides with \eqref{eq:gen_finite}.
\end{proof}

When $\beta_s$ deterministic, this theorem recovers a recent result by Tan and Atia \cite{Tan14}, which was proved using combinatorial techniques.  Our result is in fact slightly stronger in the sense that the additive term in the denominator only behaves as $\omega\big(\frac{1}{\sqrt{n}}\big)$, whereas the corresponding term in \cite{Tan14} behaves as $\omega\big(\frac{1}{n^{1/4}}\big)$.  Thus, in our result, the mutual information term remains the dominant one in a wider range of settings.  

%
%
\section{Proofs of General Bounds} \label{sec:PROOFS}

Here we provide the proofs of Theorems \ref{thm:ach1} and \ref{thm:conv}, and then give the changes required to obtain the results for partial recovery in Section \ref{sec:PARTIAL}.  As mentioned previously, the proofs bear some resemblance to those of mixed channels in channel coding \cite[Sec.~3.3]{Han03}.  However, the analysis here is more involved, primarily due to the fact that the ``codewords'' $\Xv_{s}$ are not independent for different values of $s\in\Sc$, but instead share common columns corresponding to the overlapping parts of the support set.  See \cite{Wai09,Wan10,Jin11} for further discussions on the differences between support recovery and channel coding.

\subsection{Proof of Theorem \ref{thm:ach1}} \label{sec:PROOF_GEN_ACH}

\subsubsection{Initial Non-Asymptotic Bound}

Recall the definitions of the random variables in \eqref{eq:joint_dist}--\eqref{eq:joint_dist_n}, and the information densities in \eqref{eq:idens_multi}--\eqref{eq:idens_b}.  We fix the constants $\gamma_{1},\dotsc,\gamma_{k}$ arbitrarily, and consider a decoder that searches for the unique set $s\in\Sc$ such that
\begin{equation}
\imath(\xv_{\sdif}; \yv | \xv_{\seq} ) > \gamma_{|\sdif|} \label{eq:decoder_cond}
\end{equation}
for all $2^{k} - 1$ partitions $(\sdif,\seq)$ of $s$ with $\sdif\ne\emptyset$.  An error occurs if no such $s$ exists, if multiple exist, or if such a set differs from the true value. 

Since the joint distribution of $(\beta_s,\Xv_s,\Yv_s \,|\, S=s)$ is the same for all $s$ in our setup (\emph{cf.}~Section \ref{sec:SETUP}), and the decoder that we have chosen exhibits a similar symmetry, we can condition on a fixed and arbitrary value of $S$, say $s=\{1,\dotsc,k\}$.  By the union bound, the error probability is upper bounded by
\begin{multline}
\pe \le \PP\bigg[ \bigcup_{(\sdif,\seq)} \Big\{ \imath(\Xv_{\sdif}; \Yv | \Xv_{\seq} ) \le \gamma_{|\sdif|} \Big\} \bigg] \\ +  \sum_{\sbar \in \Sc \backslash \{s\}} \PP\Big[\imath(\Xvdif; \Yv | \Xveq) > \gamma_{|\sdif|} \Big], \label{eq:thresh2} 
\end{multline}
where here and subsequently we let the condition $\sdif\ne\emptyset$ remain implicit.  The first term in \eqref{eq:thresh2} corresponds to the true set failing the threshold test, and the second term corresponds to some incorrect set $\sbar$ passing the threshold test.  In the summand of the second term, we have upper bounded the probability of an intersection of $2^k - 1$ events by just one such event, namely, the one corresponding to $\sdif = \sbar \backslash s$ and $\seq = s \cap \sbar$.

Using the shorthand $\ell := |\sbar \backslash s|$, we can weaken the second probability in \eqref{eq:thresh2} as follows:
\begin{align}
&\PP\Big[\imath(\Xvdif; \Yv | \Xveq) > \gamma_{\ell} \Big] \nonumber \\
&= \sum_{\xveq,\xvdif,\yv} P_X^{n \times (k-\ell)}(\xveq) P_X^{n \times \ell}(\xvdif) P_{\Yv|\Xv_{\seq}}(\yv|\xveq) \nonumber \\
    &\qquad \times\openone\bigg\{ \log\frac{P_{\Yv|\Xv_{\sdif}\Xv_{\seq}}(\yv|\xvdif,\xveq)}{P_{\Yv|\Xv_{\seq}}(\yv|\xveq)} > \gamma_{\ell} \bigg\} \label{eq:ach_na_3} \\
&\le \sum_{\xveq,\xvdif,\yv} P_X^{n \times (k-\ell)}(\xveq) P_X^{n \times \ell}(\xvdif) \nonumber \\
    &\qquad \times P_{\Yv|\Xv_{\sdif}\Xv_{\seq}}(\yv|\xvdif,\xveq) e^{-\gamma_{\ell}} \label{eq:ach_na_4} \\
&= e^{-\gamma_{\ell}}, \label{eq:ach_na_5}
\end{align}
where in \eqref{eq:ach_na_3} we used the fact that the output vector depends only on the columns of $\xv_{\sbar}$ corresponding to entries of $\sbar$ that are also in $s$, and \eqref{eq:ach_na_4} follows by bounding $P_{\Yv|\Xv_{\seq}}$ using the event within the indicator function, and then upper bounding the indicator function by one.  Substituting \eqref{eq:ach_na_5} into \eqref{eq:thresh2} gives
\begin{multline}
\pe \le \PP\bigg[ \bigcup_{(\sdif,\seq)} \Big\{ \imath(\Xv_{\sdif}; \Yv | \Xv_{\seq} ) \le \gamma_{\ell} \Big\} \bigg] \\ + \sum_{\ell=1}^k {{p-k} \choose \ell}{k \choose \ell}e^{-\gamma_{\ell}}, \label{eq:thresh3}
\end{multline}
where the combinatorial terms arise from a standard counting argument \cite{Wai09}.

Note that while the bound in \eqref{eq:thresh3} appears to be simpler than that in the theorem statement, it is difficult to directly apply it to specific problems, since $\imath(\Xv_{\sdif}; \Yv | \Xv_{\seq} )$ is not an i.i.d.~summation in general.

\subsubsection{Completion of the Proof}

We fix the constants $\gamma'_1,\dotsc,\gamma'_{\ell}$ arbitrarily, and apply the following elementary steps with $\ell = |\sdif|$:
\begin{align}
    &\PP\bigg[ \bigcup_{(\sdif,\seq)} \Big\{ \imath(\Xv_{\sdif}; \Yv | \Xv_{\seq} ) \le \gamma_{\ell} \Big\} \bigg] \nonumber \\
    &= \PP\bigg[ \bigcup_{(\sdif,\seq)} \bigg\{ \log\frac{ P_{\Yv|\Xv_{\sdif}\Xv_{\seq}}(\Yv|\Xv_{\sdif},\Xv_{\seq}) }{ P_{\Yv|\Xv_{\seq}}(\Yv|\Xv_{\seq}) } \le \gamma_{\ell} \bigg\} \bigg]  \label{eq:ach_disc1} \\
    &\le \PP\bigg[ \bigcup_{(\sdif,\seq)} \bigg\{ \log\frac{ P_{\Yv|\Xv_{\sdif}\Xv_{\seq}}(\Yv|\Xv_{\sdif},\Xv_{\seq}) }{ P_{\Yv|\Xv_{\seq}}(\Yv|\Xv_{\seq}) } \le \gamma_{\ell} \nonumber \\
            & \qquad\qquad\qquad \,\cap\,\log\frac{ P_{\Yv|\Xv_{\seq}}(\Yv|\Xv_{\seq}) }{ P_{\Yv|\Xv_{\seq}\beta_s}(\Yv|\Xv_{\seq},\beta_s) } \le \gamma'_{\ell} \bigg\} \bigg] \nonumber \\ 
            & \quad + \PP\bigg[ \bigcup_{(\sdif,\seq)} \bigg\{ \log\frac{ P_{\Yv|\Xv_{\seq}}(\Yv|\Xv_{\seq}) }{ P_{\Yv|\Xv_{\seq}\beta_s}(\Yv|\Xv_{\seq},\beta_s) } > \gamma'_{\ell} \bigg\} \bigg] \label{eq:ach_disc2} 
\end{align}
\begin{align}
    &\le \PP\bigg[ \bigcup_{(\sdif,\seq)} \bigg\{ \log\frac{ P_{\Yv|\Xv_{\sdif}\Xv_{\seq}}(\Yv|\Xv_{\sdif},\Xv_{\seq}) }{ P_{\Yv|\Xv_{\seq}\beta_s}(\Yv|\Xv_{\seq},\beta_s) } \le \tilde{\gamma}_{\ell} \bigg\} \bigg] \nonumber \\ 
        & \quad + \PP\bigg[ \bigcup_{(\sdif,\seq)} \bigg\{ \log\frac{ P_{\Yv|\Xv_{\seq}}(\Yv|\Xv_{\seq}) }{ P_{\Yv|\Xv_{\seq}\beta_s}(\Yv|\Xv_{\seq},\beta_s) } > \gamma'_{\ell} \bigg\} \bigg],\label{eq:fein1_split}
\end{align}
where $\tilde{\gamma}_{\ell} = \gamma_{\ell} + \gamma'_{\ell}$. The second term in \eqref{eq:fein1_split} is upper bounded as
\begin{align}
    & \PP\bigg[ \bigcup_{(\sdif,\seq)} \bigg\{ \log\frac{ P_{\Yv|\Xv_{\seq}}(\Yv|\Xv_{\seq}) }{ P_{\Yv|\Xv_{\seq}\beta_s}(\Yv|\Xv_{\seq},\beta_s) } > \gamma'_{\ell} \bigg\} \bigg] \nonumber \\
    &\le \sum_{(\sdif,\seq)} \PP\bigg[ \log\frac{ P_{\Yv|\Xv_{\seq}}(\Yv|\Xv_{\seq}) }{ P_{\Yv|\Xv_{\seq}\beta_s}(\Yv|\Xv_{\seq},\beta_s) } > \gamma'_{\ell} \bigg] \label{eq:ach_denom_ub0} \\
    & = \sum_{(\sdif,\seq)}\sum_{b_s,\xv_{\seq},\yv}P_{\beta_s}(b_s)P_{X}^{n\times(k-\ell)}(\xv_{\seq}) \nonumber \\
        & \quad\qquad\times P_{\Yv|\Xv_{\seq}\beta_s}(\yv|\xv_{\seq},b_s) \nonumber \\
        & \quad\qquad \times\openone\bigg\{ \log\frac{ P_{\Yv|\Xv_{\seq}}(\yv|\xv_{\seq}) }{ P_{\Yv|\Xv_{\seq}\beta_s}(\yv|\xv_{\seq},b_s) } > \gamma'_{\ell} \bigg\} \\
    &\le \sum_{(\sdif,\seq)}\sum_{b_s,\xv_{\seq},\yv}P_{\beta_s}(b_s)P_{X}^{n\times(k-\ell)}(\xv_{\seq}) \nonumber \\
        & \qquad\qquad\qquad\qquad\qquad \times P_{\Yv|\Xv_{\seq}}(\yv|\xv_{\seq}) e^{-\gamma'_{\ell}} \\
    & = \sum_{\ell=1}^k {k \choose \ell}e^{-\gamma'_{\ell}}, \label{eq:ach_denom_ub}
\end{align}
where \eqref{eq:ach_denom_ub0} follows from the union bound, and the remaining steps follow the arguments used in \eqref{eq:ach_na_3}--\eqref{eq:ach_na_5}.

We now upper bound the first term in \eqref{eq:fein1_split}.  The numerator in \eqref{eq:fein1_split} equals $P_{\Yv|\Xv_{s}}(\Yv|\Xv_{s})$ for all $(\sdif,\seq)$ (\emph{cf.}, \eqref{eq:p_split1}), and we can thus write the overall term as
\begin{multline}
    \PP\bigg[ \log P_{\Yv|\Xv_{s}}(\Yv|\Xv_{s}) \\ \le \max_{(\sdif,\seq)} \big\{ \log P_{\Yv|\Xv_{\seq}\beta_s}(\Yv|\Xv_{\seq},\beta_s) + \gamma_{\ell} + \gamma'_{\ell} \big\} \bigg]. \label{eq:simp_gen_1}
\end{multline}
Using the same steps as those used in \eqref{eq:ach_disc1}--\eqref{eq:fein1_split}, we can upper bound this by
\begin{multline}
    \hspace*{-2ex}\PP\bigg[ \log P_{\Yv|\Xv_s\beta_s}(\Yv|\Xv_{s},\beta_s) \\ \le \max_{(\sdif,\seq)} \big\{ \log P_{\Yv|\Xv_{\seq}\beta_s}(\Yv|\Xv_{\seq},\beta_s) + \gamma_{\ell} + \gamma'_{\ell} + \gamma \big\} \bigg] \\ + \PP\bigg[ \log \frac{P_{\Yv|\Xv_s,\beta_s}(\Yv|\Xv_{s},\beta_s)}{P_{\Yv|\Xv_{s}}(\Yv|\Xv_{s})} > \gamma \bigg] \label{eq:simp_gen_2}
\end{multline}
for any constant $\gamma$.  Reversing the step in \eqref{eq:simp_gen_1}, this can equivalently be written as
\begin{multline}
    \hspace*{-2ex}\PP\bigg[ \bigcup_{(\sdif,\seq)} \bigg\{ \log\frac{ P_{\Yv|\Xv_{\sdif}\Xv_{\seq}\beta_s}(\Yv|\Xv_{\sdif},\Xv_{\seq},\beta_s) }{ P_{\Yv|\Xv_{\seq}\beta_s}(\Yv|\Xv_{\seq},\beta_s) } \\ \le \gamma_{\ell} + \gamma'_{\ell} + \gamma \bigg\} \bigg] + \PP\bigg[ \log \frac{P_{\Yv|\Xv_s,\beta_s}(\Yv|\Xv_{s},\beta_s)}{P_{\Yv|\Xv_{s}}(\Yv|\Xv_{s})} > \gamma \bigg]. \label{eq:simp_gen_3}
\end{multline}
Observe that the first logarithm appearing here is precisely the information density in \eqref{eq:idens_bn}.  Moreover, the choices
\begin{align}
    \gamma_{\ell} &= \log\bigg(\frac{k}{\delta_1}{{p-k} \choose \ell}{k \choose \ell}\bigg) \label{eq:gamma_l} \\
    \gamma'_{\ell} &= \log\bigg(\frac{k}{\delta_1}{k \choose \ell}\bigg) \label{eq:gamma'_l}
\end{align}
make \eqref{eq:ach_denom_ub} and the second term in \eqref{eq:thresh3} be upper bounded by $\delta_1$ each.  Hence, and combining \eqref{eq:fein1_split} with \eqref{eq:ach_denom_ub} and \eqref{eq:simp_gen_3}, and recalling that $\ell = |\sdif|$, we obtain \eqref{eq:thresh_sl}.

\subsection{Proof of Theorem \ref{thm:conv}} \label{sec:PROOF_GEN_CONV}

As has been done in several previous proofs of information-theoretic converse bounds for sparsity pattern recovery \cite{Wan10,Ree13,Ati12}, we consider an argument based on a genie.  As explained formally below, the genie reveals some of elements of the support set to the decoder, which is left to estimate the remaining entries.  An important novelty in our arguments is that we also let the revealed indices depend on the random non-zero entries of $\beta$; this leads to the improvement stated following Theorem \ref{thm:simplified_conv}.

It will prove convenient to present the proof under the following assumption of symmetry.  

\begin{assump}  \label{ass:genie}
    The pair $(\sdif(b_s),\seq(b_s))$ in Theorem \ref{thm:conv} satisfies the following property:  If $b'_s$ is a permutation of $b_s$, then the entries of $b'_s$ indexed by $\sdif(b'_s)$ (respectively, $\seq(b'_s)$) are a permutation of the entries of $b_s$ indexed by $\sdif(b_s)$ (respectively, $\seq(b_s)$).
\end{assump}

We claim that the theorem statement under this assumption also implies the more general case.  To see this, we use the symmetry of $P_{Y|X_S\beta_S}$ with respect to $S$ from in Section \ref{sec:SETUP}, and the fact that $\Xv$ has an i.i.d.~distribution.  Among all the possible choices of functions $(\sdif(\cdot),\seq(\cdot))$, there always exists a pair that maximizes the lower bound in \eqref{eq:conv_sl} and satisfies Assumption \ref{ass:genie}.  More precisely, for any realization of $\beta_s$, the probability in \eqref{eq:conv_sl} is determined by entries appearing in the partition $(\beta_{\sdif},\beta_{\seq})$ but not by their order, so one can always maximize \eqref{eq:conv_sl} by forming this partition in a manner which is symmetric with respect to permutations of $\beta_s$.

We now formally define the genie-aided setup as follows:
\begin{enumerate}
    \item Generate a random $k$-dimensional vector $\betatil' \sim P_{\beta_S}$.
    \item Given $\betatil'$, let $\betatildif'$ and $\betatileq'$ be the subvectors indexed by $\sdif(\betatil')$ and $\seq(\betatil')$ respectively.
    \item Let $\betatildif$ (respectively, $\betatileq$) be a uniformly random permutation of $\betatildif'$ (respectively, $\betatileq'$). 
    \item Generate $\Seq$ uniformly on $\Sceq(\ell)$, defined to contain the $p \choose k - \ell$ subsets of $\{1,\dotsc,p\}$ having cardinality $k - \ell$, where $\ell = |\betatildif|$.  Set $\beta_{\Seq} = \betatileq$.
    \item Generate $\Sdif$ uniformly on $\Scdif(\Seq)$, defined to contain the ${p - k + \ell} \choose \ell$ subsets of $\{1,\dotsc,p\} \backslash \Seq$ having cardinality $\ell$.  Set $\beta_{\Sdif} = \betatildif$.
    \item Set $S = \Sdif \cup \Seq$ and $\beta_{S^c} = 0$.  The measurement matrix $\Xv$ is i.i.d.~on $P_{X}$, and the observation vector $\Yv$ is generated from $S$, $\Xv$, and $\beta$ according to \eqref{eq:single_output}, as in the original problem setup.
    \item Reveal the indices  $\Seq$ and the vectors $\betatildif$ and $\betatileq$ to the decoder.  The decoder forms an estimate $\Sdifhat$ of $\Sdif$, and an error occurs if $\Sdifhat \ne \Sdif$.
\end{enumerate} 
The joint distribution of $S$ and $\beta$ is the same here as in the original setup: The support set is uniform on the $p \choose k$ elements of $\Sc$, and the distribution of the non-zero entries $\beta_S$ is that of a uniformly random permutation of $\betatil' \sim P_{\beta_S}$.  Since $P_{\beta_S}$ is permutation-invariant by assumption, this yields $\beta_S \sim P_{\beta_S}$ as required.  Thus, the only difference in this modified setup is that the decoder has further information, and it follows that any converse for this setup implies the same converse for the original setup.

Throughout the proof, we make use of the random variables defined in the preceding steps, departing from the notation implicitly conditioned on $S$ equaling a fixed value $s$ (see \eqref{eq:joint_dist_n}) until the final step in obtaining \eqref{eq:conv_sl}.

We first study the error probability for the genie-aided setting conditioned on $(\Seq,\betatildif,\betatileq) = (\seq,\btildif,\btileq)$, denoted by $\pe(\seq,\btildif,\btileq)$.  By the identity $\PP[\Ac] = \PP[\Ac \cap \Ec] + \PP[\Ac \cap \Ec^c]$, we have for any event $\Ac(\seq,\btildif,\btileq)$ that
\begin{multline}
    \pe(\seq,\btildif,\btileq) \ge \PP[\Ac(\seq,\btildif,\btileq)] \\ - \PP[\Ac(\seq,\btildif,\btileq) \cap \text{no error}]. \label{eq:conv_init}
\end{multline}
We fix the constant $\gamma_{\ell}$ and choose
\begin{equation}
    \Ac(\seq,\btildif,\btileq) = \big\{ \imath^n(\Xv_{\Sdif};\Yv|\Xv_{\seq},\btil) \le \gamma_{\ell} \big\},
\end{equation}
where $\ell = k - |\seq|$, and $\btil := \btil(\btildif,\btileq,\sdif,\seq)$ equals $\btildif$ (respectively, $\btileq$) on the entries indexed by $\sdif$ (respectively, $\seq$).  Using the definitions in \eqref{eq:idens_bn}--\eqref{eq:idens_b}, and defining $\Dc(\sdif | \seq,\btildif,\btileq)$ to be the set of pairs $(\xv,\yv)$ such that the decoder outputs $\sdif$ given $(\seq,\btildif,\btileq,\xv,\yv)$, we obtain
\begin{align} 
    &\PP[\Ac(\seq,\btildif,\btileq) \cap \text{no error}] \nonumber \\
    &=\sum_{\sdif\in\Scdif(\seq)} \frac{1}{{{p-k+\ell} \choose \ell}} \sum_{(\xv,\yv)\in\Dc(\sdif | \seq,\btildif,\btileq)} P_{X}^{n \times p}(\xv) \nonumber \\
       	&\qquad \times P^n_{Y|X_{\sdif}X_{\seq}\beta_s}(\yv|\xv_{\sdif},\xv_{\seq},\btil) \nonumber \\
       	&\qquad \times \openone\bigg\{ \log\frac{P^n_{Y|X_{\sdif}X_{\seq}\beta_s}(\yv|\xv_{\sdif},\xv_{\seq},\btil)}{P^n_{Y|X_{\seq}\beta_s}(\yv|\xv_{\seq},\btil)} \le \gamma_{\ell} \bigg\} \label{eq:vh_deriv1}
\end{align}
\begin{align}
    &\le \frac{1}{{{p-k+\ell} \choose \ell}}\sum_{\sdif\in\Scdif(\seq)} \sum_{(\xv,\yv)\in\Dc(\sdif | \seq,\btildif,\btileq)}  P_{X}^{n \times p}(\xv) \nonumber \\
        &\qquad\qquad\qquad\qquad\qquad \times P^n_{Y|X_{\seq}\beta_s}(\yv|\xv_{\seq},\btil)e^{\gamma_{\ell}} \label{eq:vh_deriv2} \\
    &= \frac{e^{\gamma_{\ell}}}{{ {p-k+\ell} \choose \ell }},  \label{eq:vh_deriv3}
\end{align}
where \eqref{eq:vh_deriv1} follows since an error occurs if and only if $(\xv,\yv)\notin\Dc(\sdif | \seq,\btildif,\btileq)$, \eqref{eq:vh_deriv2} follows by upper bounding $P^n_{Y|X_{\seq}\beta_s}$ using the event in the indicator function, and \eqref{eq:vh_deriv3} follows since the sets $\Dc(\sdif | \seq,\btildif,\btileq)$ are disjoint, and their union over $\sdif$ is the entire space of $(\xv,\yv)$ pairs.

Averaging \eqref{eq:conv_init} over $(\Seq,\betatil',\betatildif,\betatileq)$ and applying \eqref{eq:vh_deriv3}, we obtain
\begin{align}
\pe &\ge \sum_{\btil'} P_{\beta_S}(\btil') \sum_{\btildif,\btileq} \PP\big[ (\betatildif,\betatileq) = (\btildif,\btileq) \,|\, \btil' \big] \nonumber \\
    &\times \sum_{\seq\in\Seq(\ell)}\sum_{\sdif\in\Sdif(\seq)} \frac{1}{{p \choose {k -\ell}}}  \frac{1}{{{p-k+\ell} \choose \ell}} \nonumber \\ 
    &\times \bigg(\PP\Big[ \imath^n(\Xv_{\sdif};\Yv|\Xv_{\seq},\btil) \le \gamma_{\ell} \,\big|\, \sdif,\seq,\btildif,\btileq \Big] \nonumber \\
    &\qquad\qquad\qquad\qquad\qquad\qquad\qquad-  \frac{ e^{\gamma_{\ell}}}{{ {p-k+\ell} \choose \ell }}\bigg),\label{eq:vh_deriv4}
\end{align}
where $\ell = |\btildif|$, and the conditioning on $\btil'$ is a shorthand for $\betatil' = \btil'$, and similarly for the second probability.  Finally, we claim that this recovers \eqref{eq:conv_sl} upon setting
\begin{align}
	\gamma_{\ell} &= \log{{p-k+\ell} \choose \ell} + \log \delta_1.
\end{align}
To see this, we first note that all of the terms in the summations over $\sdif$ and $\seq$ in \eqref{eq:vh_deriv4} are equal, since in the probability appearing in the summand, the entries $\btil_{\sdif}$ and $\btil_{\seq}$ are the same for any such pair, namely, $\btil_{\sdif} = \btildif$ and $\btil_{\seq} = \btileq$  (recall also the symmetry of $P_{Y|X_S\beta_S}$ with respect to $S$ assumed in Section \ref{sec:SETUP}).  Due to Assumption \ref{ass:genie}, this probability also coincides with that in \eqref{eq:conv_sl} with $b_s := \btil'$, regardless of the realization of $(\betatildif,\betatileq)$ given $\betatil'$; the only randomness in the corresponding distribution is that of the two random permutations of the subvectors.

\subsection{Extensions to Partial Recovery} \label{sec:PROOF_PARTIAL}

The achievability analysis in Section \ref{sec:PROOF_GEN_ACH} extends immediately to handle the partial recovery criterion in \eqref{eq:pe_partial}, since we have already split the error events according to the amount of overlap between the true support and the incorrect support.  The only difference is that the decoder searches for a set $s$ such that \eqref{eq:decoder_cond} holds whenever $|\sdif| > \dmax$ (as opposed to $\sdif\ne\emptyset$), and chooses one arbitrarily if multiple such $s$ exist. It follows that Theorem \ref{thm:ach1} remains true when the union in \eqref{eq:thresh_sl} is restricted to $|\sdif|\in\{\dmax+1,\dotsc,k\}$.

The extension of the converse analysis in Section \ref{sec:PROOF_GEN_CONV} is less immediate, but still straightforward.  We first recall the observation from \cite{Ree13} that the performance metric in \eqref{eq:pe_partial} allows us to focus without loss of generality on decoders such that the estimated support $\hat{S}$ (or $\Sdifhat \cup \Seq$ in the genie-aided setting) has cardinality $k$ almost surely.  For any such decoder, the definition in \eqref{eq:pe_partial} is unchanged when the second term in the union is removed.

We restrict the partitions $(\sdif(b_s),\seq(b_s))$ of $s$ to satisfy $|\sdif(b_s)| > \dmax$.  In \eqref{eq:vh_deriv1}--\eqref{eq:vh_deriv2}, we change the definition of $\Dc(\sdif | \seq,\btildif,\btileq)$ to be the set of pairs $(\xv,\yv)$ such that the decoder outputs a sequence $\sdifhat$ such that $|\sdif \backslash \sdifhat| \le \dmax$.  This means that the sets $\Dc(\cdot | \seq,\btildif,\btileq)$ are no longer disjoint.  However, we can easily count the number of such sets that each $(\xv,\yv)$ pair falls into.  For fixed $(\seq,\sdif)$ and $d\in\{0,\dotsc,\dmax\}$, the number of sets $\sdifhat \subseteq \{1,\dotsc,p\} \backslash \seq$ such that $|\sdif \backslash \sdifhat| = d$ is ${{p-k} \choose d}{|\sdif| \choose {|\sdif| - d}} = {{p-k} \choose d}{|\sdif| \choose d}$.  Thus, each $(\xv,\yv)$ pair is included in $\sum_{d=0}^{\dmax}{{p-k} \choose d}{|\sdif| \choose d}$ of the sets $\Dc(\cdot | \seq,\btildif,\btileq)$, and \eqref{eq:vh_deriv3} is replaced by
\begin{align}
    \PP[\Ac(\seq,\btildif,\btileq) \cap \text{no error}] \le \frac{ \sum_{d=0}^{\dmax}{{p-k} \choose d}{\ell \choose d} }{{ {p-k+\ell} \choose \ell }} e^{-\gamma_{\ell}}.
\end{align}
Thus, Theorem \ref{thm:conv} remains true when the pair $(\sdif(\cdot),\seq(\cdot))$ is constrained to satisfy $|\sdif|\in\{\dmax+1,\dotsc,k\}$, and ${ {p-k+|\sdif|} \choose |\sdif| }$ is replaced by ${ {p-k+|\sdif|} \choose |\sdif| } - \sum_{d=0}^{\dmax}{{p-k} \choose d}{|\sdif| \choose d}$.

%
%
\section{Conclusion} \label{sec:CONCLUSION}

Taking an approach motivated by thresholding techniques in channel coding, we have presented a framework for developing necessary and sufficient conditions on the number of measurements for exact and partial support recovery with probabilistic models.  We have provided several new results for the linear, 1-bit, and group testing models, as well as general discrete observation models.  In several cases, we have provided exact asymptotic thresholds on the number of measurements with strong converse results.

There are several possible directions for future research.  While we have focused on i.i.d.~measurement matrices, it would be of significant interest to consider other types of random matrices, and to present converse results that hold for arbitrary measurement matrices, subject to suitable constraints such as power constraints.  We provided some work in these directions for specific models in \cite{Sca16b,Sca16c}.

One could also attempt to move from standard sparsity models to structured sparsity models \cite{Bar10}, and from probabilistic guarantees with random $\beta$ to minimax guarantees.  There are several additional non-linear models that our general results could be applied to, such as the Poisson and gamma models.  Finally, it may be interesting to apply similar analysis techniques to other statistical problems beyond support recovery.

\appendices

\section{Concentration Inequalities} \label{sec:CONCENTRATION}

In order to apply our general bounds to specific models, we use concentration inequalities to obtain expressions for $\psi_{\ell}$ and $\psi'_{\ell}$ in \eqref{eq:psi_ach} and \eqref{eq:psi_conv}, seeking to make the corresponding terms in \eqref{eq:nonasymp_ach} and \eqref{eq:nonasymp_conv} vanish.  Here we present two general inequalities that will be used throughout Section \ref{sec:EXAMPLES}.

\begin{prop} \label{prop:chebyshev}
    For general observation models, we have for all $(\sdif,\seq,b_s)$ and $\delta>0$ that
    \begin{multline}
    \PP\Big[ \big|\imath^n(\Xv_{\sdif}; \Yv | \Xv_{\seq}, b_s) - nI_{\sdif,\seq}(b_s) \big| \ge n\delta \,\Big|\, \beta_s = b_s \Big] \\ \le \frac{V_{\sdif,\seq}(b_s)}{\delta^2 n}, \label{eq:conc_linear1}
    \end{multline}
    where $V_{\sdif,\seq}(b_s)$ is defined in \eqref{eq:Vb}.  Moreover, if the observations lie in a finite set $\Yc \subset \RR$ with probability one, then the following holds for all $(\sdif,\seq,b_s)$ and $\delta>0$:
    \begin{equation}
    V_{\sdif,\seq}(b_s) \le |\Yc|\Big( \frac{4}{e} \Big) ^2. \label{eq:uniformV}
    \end{equation}
\end{prop}

Before providing the proof, we state the following generalization of \eqref{eq:uniformV} to higher-order moments.

\begin{prop} \label{prop:gen_discrete}
    If the observations lie in a finite set $\Yc \subset \RR$ with probability one, then the following holds  for all $(\sdif,\seq,b_s)$ and $\delta > 0$:
    \begin{multline}
    \PP\bigg[ \big|\imath^n(\Xv_{\sdif}; \Yv | \Xv_{\seq}, b_s) - nI_{\sdif,\seq}(b_s) \big| \ge n\delta \,\Big|\, \beta_s = b_s \bigg] \\ \le 2\exp\bigg(- \frac{\delta^2 n}{2(8|\Yc|+2\delta)} \bigg). \label{eq:conc_gen_disc}
    \end{multline}
\end{prop}

In the remainder of this appendix, we prove these propositions.  Equation \eqref{eq:conc_linear1} follows from Chebyshev's inequality, so we focus our attention on \eqref{eq:uniformV}--\eqref{eq:conc_gen_disc}.  We make use of the following form of Bernstein's inequality \cite[Sec.~2.8]{Bou13}.

\begin{lem} \label{lem:bernstein}
    Let $W_1,\dotsc,W_n$ be independent real-valued random variables such that
    \begin{align}
        \sum_{i=1}^n \EE[W_i^2] &\le \tau \\
        \sum_{i=1}^n \EE[ |W_i|^q ] &\le \frac{q!}{2}\tau c^{q-2} \qquad (q \ge 3)
    \end{align}
    for some $\tau,c>0$.  Then
    \begin{equation}
    \PP\bigg[ \sum_{i=1}^n \big( W_i - \EE[W_i] \big) \ge t \bigg] \le \exp\bigg( \frac{t^2}{2(\tau+ct)} \bigg) \label{eq:bernstein}
    \end{equation}
    for all $t>0$.
\end{lem}

To bound the moments of $\imath$, we follow the arguments of \cite[Rmk.~3.1.1]{Han03} and \cite[App.~D]{Tan14a}. Recall the definition of the information density in \eqref{eq:idens_b}.  For any $q \ge 2$, we have from Minkowski's inequality that
\begin{multline}
\EE\big[ |\imath(X_{\sdif}; Y | X_{\seq}, b_s)|^q\big]^{1/q} \\ \le \EE\bigg[ \Big(\log\frac{1}{ P_{Y|X_{\sdif}X_{\seq}\beta_s}(Y|X_{\sdif},X_{\seq},b_s) }\Big)^q\bigg]^{1/q} \\ + \EE\bigg[ \Big(\log\frac{ 1 }{ P_{Y|X_{\seq}\beta_s}(Y|X_{\seq},b_s) }\Big)^q\bigg]^{1/q}, \label{eq:disc_unif1}
\end{multline}
where here and subsequently we implicitly condition on $\beta_s = b_s$.  For any given $(x_{\sdif},x_{\seq})$, the remaining averaging over $Y$ in the first term has the form
\begin{multline}
    \sum_{y} P_{Y|X_{\sdif}X_{\seq}\beta_s}(y|x_{\sdif},x_{\seq},b_s) \\ \times\Big(\log\frac{1}{ P_{Y|X_{\sdif}X_{\seq}\beta_s}(y|x_{\sdif},x_{\seq},b_s) }\Big)^q,
\end{multline}
and is thus upper bounded by $|\Yc| \big( \frac{q}{e} \big)^{1/q}$, since the function $f(z) = z \log^{q}\frac{1}{z}$ has a maximum value of $\big( \frac{q}{e} \big)^{1/q}$ for $z\in[0,1]$.  Handling the second term in \eqref{eq:disc_unif1} similarly, we obtain 
\begin{equation}
    \EE\big[ \big|\imath(X_{\sdif}; Y | X_{\seq}, b_s)\big|^q\big]^{1/q} \le 2\Big( |\Yc| \Big( \frac{q}{e} \Big)^{q} \Big)^{1/q},
\end{equation}
or equivalently
\begin{align}
\EE\big[ \big|\imath(X_{\sdif}; Y | X_{\seq}, b_s)\big|^q\big] 
    & \le \Big( \frac{q}{e} \Big)^{q} 4|\Yc| 2^{q-2} \label{eq:disc_unif4} \\
    & \le \frac{q!}{2} 8|\Yc| 2^{q-2}, \label{eq:disc_unif5}
\end{align}
where \eqref{eq:disc_unif5} follows since $\big( \frac{q}{e} \big)^{q} \le q!$.

We obtain \eqref{eq:uniformV} by setting $q=2$ in \eqref{eq:disc_unif4}.  Furthermore, we obtain Proposition \ref{prop:gen_discrete} using Lemma \ref{lem:bernstein} with $c = 2$, $\tau = n\cdot8|\Yc|$, and $t=\delta n$.

\section{Proofs of Auxiliary Results for the Linear Model} \label{sec:LINEAR_PROOFS}

\subsection{Proof of Proposition \ref{prop:conc_linear2}}

We again use Lemma \ref{lem:bernstein}, and we thus seek suitable values for $\tau$ and $c$.  Throughout the proof, we consider the random variables $(X_{\sdif},X_{\seq},Y)$ distributed according to \eqref{eq:distr_sl}, implicitly conditioning on $\beta_s=b_s$.  From \eqref{eq:linear_model}, we have $Z = Y - \sum_{i \in s}X_i b_i$, and a direct calculation gives
\begin{gather}
    P_{Y|X_{\sdif}X_{\seq}\beta_s}(Y|X_{\sdif},X_{\seq},b_s) = \phi(Z; 0, \sigma^2) 
\end{gather}
\vspace*{-3ex}
\begin{multline}
    P_{Y|X_{\seq}\beta_s}(Y|X_{\seq},b_s) \\ = \phi\Big( \sum_{i\in\sdif} X_i b_i + Z; 0, \sigma^2 + \sum_{i\in\sdif}b_i^2 \Big),
\end{multline}
where $\phi(\cdot;\mu,\sigma^2)$ is the $N(\mu,\sigma^2)$ density function.  Substituting these into \eqref{eq:idens_b} gives 
\begin{multline}
\imath(X_{\sdif}; Y | X_{\seq}, b_s) = I_{\sdif,\seq}(b_s) - \frac{Z^2}{2\sigma^2} \\ + \frac{1}{2\big(\sigma^2 + \sum_{i\in\sdif}b_i^2 \big)} \bigg( \sum_{i\in\sdif} X_i b_i + Z \bigg)^2, \label{eq:n_i_dens}
\end{multline}
where $I_{\sdif,\seq}(b_s)$ is given in \eqref{eq:I_linear}.

The mean of \eqref{eq:n_i_dens} is $I_{\sdif,\seq}(b_s)$, and we will apply Lemma \ref{lem:bernstein} with $W_i$ corresponding to the sum of the second and third terms on the right-hand side.  We can write these in terms of independent $N(0,1)$ random variables (denoted by $\Zhat_1$ and $\Zhat_2$) as follows:
\begin{align}
W &= - \frac{\Zhat_1^2}{2} + \frac{1}{2(\sigma^2 + \siglv^2)} \big( \sigma \Zhat_1 + \siglv \Zhat_2  \big)^2 \label{eq:linear4} \\
  &= \frac{\siglv^2}{2(\sigma^2 + \siglv^2)}\big( \Zhat_2^2 - \Zhat_1^2 \big) + \frac{\sigma\siglv}{\sigma^2 + \siglv^2} \Zhat_1\Zhat_2, \label{eq:linear4a}
\end{align}
where we have used the definitions in the proposition statement, and \eqref{eq:linear4a} follows from simple manipulations.  Defining $\Zhatmax = \max\{ |\Zhat_1|,|\Zhat_2| \}$, we have the following with probability one:
\begin{align}
    |W| &\le \frac{\siglv^2}{2(\sigma^2 + \siglv^2)} 2\Zhatmax^2 + \frac{\sigma\siglv}{\sigma^2 + \siglv^2} \Zhatmax^2 \label{eq:linear5} \\
    &= \frac{\siglv(\sigma + \siglv)}{\sigma^2 + \siglv^2} \Zhatmax^2. \label{eq:linear6}
\end{align}
Since $\EE[\Zhatmax^4] \le \EE[\Zhat_1^4 + \Zhat_2^4] = 6$, we obtain
\begin{align}
    \EE[W^2] &\le 6 \bigg(\frac{\siglv(\sigma + \siglv)}{\sigma^2 + \siglv^2}\bigg)^2. \label{eq:linear7}
\end{align}
Similarly, we can bound the higher moments as follows:
\begin{align}
    \EE[|W|^q] &\le \bigg(\frac{\siglv(\sigma + \siglv)}{\sigma^2 + \siglv^2}\bigg)^{q} \EE[\Zhat_1^{2q} + \Zhat_2^{2q}] \label{eq:linear9} \\
    &\le \bigg(\frac{2\siglv(\sigma + \siglv)}{\sigma^2 + \siglv^2}\bigg)^{q} \frac{2}{\sqrt{\pi}} \Gamma\Big(q + \frac{1}{2}\Big) \label{eq:linear10} \\
    &\le 2 \cdot \bigg(\frac{2\siglv(\sigma + \siglv)}{\sigma^2 + \siglv^2}\bigg)^{q} \cdot q! \label{eq:linear11},
\end{align}
where \eqref{eq:linear10} follows by the same argument as \eqref{eq:linear6} and the fact that the $2q$-th moment of an $N(0,1)$ random variable is $\frac{2^q}{\sqrt{\pi}} \Gamma\big(q + \frac{1}{2}\big)$, and \eqref{eq:linear11} follows since $\Gamma\big(q + \frac{1}{2}\big) \le \sqrt{\pi} q!$.

Combining \eqref{eq:linear7} and \eqref{eq:linear11}, we see that the random variables $W_i = \imath(X_{\sdif}^{(i)}; Y^{(i)} | X^{(i)}_{\seq}, b_s) - I_{\sdif,\seq}(b_s)$ satisfy the conditions of Lemma \ref{lem:bernstein} with $\tau = n\cdot4\alpha_{\sdif}^2$ and $c=\alpha_{\sdif}$ (see \eqref{eq:alpha_linear}).  We thus obtain the desired result from \eqref{eq:bernstein} by identifying $t = \delta n$.

\subsection{Proof of Proposition \ref{prop:boundMI}}

Since $\Yv = \Xv_s\beta_s + \Zv$, we have
\begin{align}
    I_0 = I(\beta_s;\Yv|\Xv_s) &= H(\Yv|\Xv_s) - H(\Yv|\Xv_s,\beta_s) \\
                                   &= H(\Xv_s\beta_s + \Zv|\Xv_s) - H(\Zv). \label{eq:pf_bmi2}
\end{align}
From \cite[Ch.~9]{Cov01}, we have $H(\Zv) = \frac{n}{2}\log(2\pi e \sigma^2)$ and $H(\Xv_s\beta_s + \Zv|\Xv_s = \xv_s) = \frac{1}{2}\log\big( (2\pi e)^n \det( \sigma^2 \Iv_n + \sigbeta^2 \xv_s \xv_s^T ) \big) $, where $\Iv_n$ is the $n \times n$ identity matrix.  Averaging the latter over $\Xv_s$ and substituting these into \eqref{eq:pf_bmi2} gives
\begin{align}
    I_0 &= \frac{1}{2} \EE\Big[  \log \det \Big( \Iv_n + \frac{\sigbeta^2}{\sigma^2} \Xv_s \Xv_s^T \Big) \Big] \\
         &= \frac{1}{2} \EE\Big[ \log \det \Big( \Iv_k + \frac{\sigbeta^2}{\sigma^2} \Xv_s^T \Xv_s \Big) \Big] \label{eq:pf_bmi4} \\
         &= \frac{1}{2} \sum_{i=1}^k \EE\Big[ \log \Big( 1+ \frac{\sigbeta^2}{\sigma^2} \lambda_i (\Xv_s^T \Xv_s) \Big) \Big] \label{eq:pf_bmi5} \\
         &\le \frac{k}{2} \log \Big( 1 + \frac{n\sigbeta^2}{\sigma^2} \Big), \label{eq:pf_bmi6}
\end{align}
where \eqref{eq:pf_bmi4} follows from the identity $\det(\Iv + \Av\Bv) = \det(\Iv + \Bv\Av)$, \eqref{eq:pf_bmi5} follows by writing the determinant as a product of eigenvalues (denoted by $\lambda_i(\cdot)$), and \eqref{eq:pf_bmi6} follows from Jensen's inequality and the following calculation:
\begin{equation}
    \frac{1}{k}\EE\Big[ \sum_{i=1}^k \lambda_i (\Xv_s^T \Xv_s) \Big] = \frac{1}{k} \EE[ \Tr( \Xv_s^T \Xv_s )] = \EE[ \Xv_1^T \Xv_1 ] = n.
\end{equation}
This concludes the proof of \eqref{eq:bmi_I0}.

We now turn to the bounding of the variance.  Again using the fact that $\Yv = \Xv_s\beta_s + \Zv$, we have
\begin{align}
    &\log \frac{P_{\Yv|\Xv_s,\beta_s}(\Yv|\Xv_{s},\beta_s)}{P_{\Yv|\Xv_{s}}(\Yv|\Xv_{s})} \nonumber \\
        &= \log\frac{P_{\Zv}(\Zv)}{P_{\Yv|\Xv_{s}}(\Xv_s\beta_s + \Zv|\Xv_s)} \\
        &= I_0 - \frac{1}{2\sigma^2}\Zv^T\Zv \nonumber \\
            &~~ + \frac{1}{2}(\Xv_s\beta_s + \Zv)^T \big( \sigma^2 \Iv + \sigbeta^2 \Xv_s\Xv_s^T \big)^{-1} (\Xv_s\beta_s + \Zv), \label{eq:pf_bmi8}
\end{align}
where $P_{\Zv}$ is the density of $\Zv$, and \eqref{eq:pf_bmi8} follows by a direct substitution of the densities $P_{\Zv} \sim N(\bzero,\sigma^2\Iv)$ and $P_{\Yv|\Xv_S}(\cdot|\xv_s) \sim N(\bzero, \sigma^2\Iv + \sigbeta^2\xv_s\xv_s^T)$, where $\bzero$ is the zero vector.  Observe now that  $\frac{1}{\sigma^2} \Zv^T \Zv$ is a sum of $n$ independent $\chi^2$ random variables with one degree of freedom (each having a variance of $2$), and hence, the second term in \eqref{eq:pf_bmi8} has a variance of $\frac{n}{2}$.  Moreover, by writing $\Mv^{-1} = (\Mv^{-\frac{1}{2}})^T \Mv^{-\frac{1}{2}}$ for the symmetric positive definite matrix $\Mv = \sigma^2 \Iv + \sigbeta^2 \Xv_s\Xv_s^T$, where $(\cdot)^{-\frac{1}{2}}$ denotes the positive definite matrix square root of the inverse, we find that the final term in \eqref{eq:pf_bmi8} is distributed as a sum of $\chi^2$ variables when conditioned on any value of $\Xv_s$, and hence, the same is true unconditionally.  We therefore again obtain a variance of $\frac{n}{2}$, and \eqref{eq:bmi_V0} follows using the identity $\var[A+B] \le \var[A] + \var[B] + 2\max\{\var[A], \var[B]\}$.

\section{Proofs of Auxiliary Results for the 1-bit Model} \label{sec:1BIT_PROOFS}

We first write down the relevant probability distributions and information densities conditioned on a fixed value $b_s$ of $\beta_s$.  Under the model $Y = \sign\big(\sum_{i\in s}X_i b_i + Z\big)$ with $X_i \sim N(0,1)$ and $Z \sim N(0,\sigma^2)$, we have
\begin{align}
    P_{Y|X_s\beta_s}(1|x_s,b_s) 
        &= \PP\bigg[ Z \ge - \sum_{i\in s}x_i b_i\bigg] \\
        &= Q\bigg( - \frac{1}{\sigma} \sum_{i\in s}x_i b_i \bigg). \label{eq:1bit_p1}
\end{align}
Similarly, for any partition of $s$ into $(\sdif,\seq)$, we can write $Y = \sign\big(\sum_{i\in \seq}X_i b_i + \sum_{i\in \sdif}X_i b_i + Z\big)$ and use the same steps to conclude that
\begin{align}
    P_{Y|X_{\seq}\beta_s}(1|x_{\seq},b_s) = Q\bigg( \frac{ -\sum_{i\in \seq}x_i b_i }{\sqrt{\sigma^2 + \sum_{i\in \sdif}b_i^2}} \bigg). \label{eq:1bit_p2}
\end{align}
The corresponding probabilities for $y=0$ are one minus these expressions, which amounts to multiplying the argument to the Q-function by $-1$.  Substitution into \eqref{eq:idens_b} gives
\begin{equation}
    \imath(x_{\sdif}; y | x_{\seq}, b_s) = \log \frac{ Q\Big( -y \frac{1}{\sigma} \sum_{i\in s}x_i b_i\Big) }{ Q\Big( \frac{ -y \sum_{i\in \seq}x_i b_i }{\sqrt{\sigma^2 + \sum_{i\in \sdif}b_i^2}} \Big) } \label{eq:1bit_i_dens}
\end{equation}
for $y\in\{-1,1\}$.

Throughout this appendix, we will use the fact that the first two derivatives of the function
\begin{equation}
    f(x) := H_2(Q(x)) \label{eq:1bit_f}
\end{equation}
are given by
\begin{align}
    f'(x)  &= \log\frac{1-Q(x)}{Q(x)} \frac{-1}{\sqrt{2\pi}} e^{-\frac{x^2}{2}} \label{eq:1bit_f'} \\
    f''(x) &= -\frac{1}{2\pi} e^{-x^2}\frac{1}{Q(x)(1-Q(x))} \nonumber \\
        & \hspace*{13ex} + \log\frac{1-Q(x)}{Q(x)} \frac{x}{\sqrt{2\pi}} e^{\frac{-x^2}{2}}. \label{eq:1bit_f''}
\end{align} 

\subsection{Proof of Proposition \ref{prop:1bit_moments} Part (i)}

Recalling that the coefficients $X_i$ ($i\in s$) are i.i.d. on $N(0,1)$, we directly obtain from \eqref{eq:1bit_p1} that
\begin{align}
    H(Y|X_s,\beta_s=b_s)
        &= \EE\bigg[ H_2\bigg( Q\bigg( \frac{1}{\sigma} \sum_{i\in s}X_i b_i \bigg) \bigg)\bigg] \\
        &= \EE\bigg[ H_2\bigg( Q\bigg( W \sqrt{\frac{1}{\sigma^2}\sum_{i \in s} b_i^2} \bigg) \bigg) \bigg],
\end{align}
where $W \sim N(0,1)$.  By evaluating $H(Y|X_{\seq},\beta_s=b_s)$ similarly using \eqref{eq:1bit_p2} and taking the difference between the two, we obtain \eqref{eq:1bit_Ilv}.

\subsection{Proof of Proposition \ref{prop:1bit_moments} Part (ii)}

We obtain from \eqref{eq:1bit_f'}--\eqref{eq:1bit_f''} that $f'(0) = 0$ and $f''(0) = -\frac{2}{\pi}$.  By performing further differentiations, one can also verify that $f^{(3)}(0) = 0$, and that $|f^{(4)}(x)|$ is uniformly upper bounded by $f^{(4)}(0) = \frac{8(\pi - 1)}{\pi^2}$.  We thus obtain via a fourth-order Taylor expansion that
\begin{multline}
    \log 2 - \frac{1}{\pi}x^2 - \frac{4(\pi - 1)}{3\pi^2}x^4 \le H_2(Q(x)) \\ \le \log 2 - \frac{1}{\pi}x^2 + \frac{4(\pi - 1)}{3\pi^2}x^4 \label{eq:1bit_pfii_1}
\end{multline}
for all $x \in \RR$.  Substituting \eqref{eq:1bit_pfii_1} into \eqref{eq:1bit_Ilv} and noting that the fourth moments of the arguments to $H_2(Q(\cdot))$ therein decay to zero strictly faster than the second moments (by the assumptions on $k$, $\bmin$ and $\bmax$), we obtain
\begin{equation}
    I_{\sdif,\seq}(b_s) = \frac{1}{\pi}\bigg(  \frac{1}{\sigma^2}\sum_{i \in s} b_i^2 - \frac{\sum_{i\in\seq}b_i^2}{\sigma^2 + \sum_{i\in\sdif}b_i^2} \bigg)(1+o(1)).
\end{equation}
Again using the assumptions on $k$, $\bmin$ and $\bmax$, we observe that the denominator is dominated by the term $\sigma^2$, thus yielding \eqref{eq:1bit_Ilv_asymp}.

\subsection{Proof of Proposition \ref{prop:1bit_moments} Part (iii)}

In this part, we have assumed that the values $\{b_i\}$ take a common value $b_0$.  Since $\sigma^2=\Theta(1)$, we may set $\sigma^2=1$ without loss of generality; the implied constant can be factored into $b_0$.  In this case, \eqref{eq:1bit_Ilv} with $\ell=1$ simplifies to
\begin{multline}
   I_1 =  \EE\bigg[ H_2\bigg( Q\bigg( W \sqrt{ \frac{(k-1)b_0^2}{1 + b_0^2} } \bigg) \bigg) \\ - H_2\bigg( Q\Big( W \sqrt{kb_0^2} \Big) \bigg) \bigg]. \label{eq:1bit_Ilv_simp}
\end{multline}
By the assumptions $k=\Theta(p)$ and $b_0^2 = \Theta\big( \frac{\log p}{p} \big)$, it is easily verified by a Taylor expansion of the function $f(z) = \frac{1}{\sqrt{1+z}}$ as $z\to0$ that $\sqrt{\frac{(k-1)b_0^2}{1 + b_0^2}} = \sqrt{kb_0^2}\big(1-\frac{b_0^2}{2} + o(b_0^2)\big)$.  For convenience, we write this identity as
\begin{equation}
    \sqrt{\frac{(k-1)b_0^2}{1 + b_0^2}} = \sqrt{kb_0^2}\big(1 - \zeta b_0^2\big), \label{eq:sqrt_Taylor}
\end{equation}
where $\zeta$ is a constant depending on $p$ such that $\zeta \to \frac{1}{2}$.  Substituting \eqref{eq:sqrt_Taylor} into \eqref{eq:1bit_Ilv_simp}, we obtain
\begin{multline}
   I_1 =  \EE\bigg[ H_2\bigg( Q\Big( W \sqrt{kb_0^2}\big(1 - \zeta b_0^2\big) \Big) \bigg) \\ - H_2\bigg( Q\Big( W \sqrt{kb_0^2} \Big) \bigg) \bigg]. \label{eq:1bit_Ilv_simp2}
\end{multline}
The next step is to Taylor expand the function $f(x) = H_2(Q(x))$.  For any $x$ and $\delta>0$, we have
\begin{equation}
    f(x - \delta) = f(x) + \frac{\delta}{\sqrt{2\pi}} \log\frac{1-Q(x)}{Q(x)}e^{-\frac{x^2}{2}} + \frac{\delta^2}{2}f''(x - \delta_0) \label{eq:1bit_taylor}
\end{equation}
for some $\delta_0 \in [0,\delta]$, where the middle term follows from \eqref{eq:1bit_f'}.  Next, we claim that $f''$ in \eqref{eq:1bit_f''} is bounded as follows:
\begin{equation}
    |f''(x)| \le \frac{2}{\sqrt{2\pi}} (1+|x|) e^{\frac{-x^2}{2}} + \frac{|x|^3}{\sqrt{2\pi}} e^{\frac{-x^2}{2}}. \label{eq:1bit_f''_bound}
\end{equation}
In the case that $x \ge 0$, this is seen by applying $Q(x) \ge \frac{1}{\sqrt{2\pi}(1+x)}e^{-\frac{x^2}{2}}$ and $1-Q(x) \ge \frac{1}{2}$ to obtain the first term, and applying $Q(x) \le e^{-x^2}$ (and hence $\log\frac{1-Q(x)}{Q(x)} = \log\big(\frac{1}{Q(x)} - 1\big) \le x^2$) to obtain the second term (e.g., see \cite{Fan12} for bounds on the Q-function).  The case $x < 0$ follows since \eqref{eq:1bit_f''} is symmetric about zero.

Substituting \eqref{eq:1bit_taylor} into \eqref{eq:1bit_Ilv_simp2} with the identifications $x=W\sqrt{kb_0^2}$ and $\delta = W\sqrt{kb_0^2}\zeta b_0^2$, we can write 
\begin{equation}
    T_1 - T_2 - T_3 \le I_1 \le T_1 + T_2 + T_3,
\end{equation}
where
\begin{align}
    T_1 &:=  \zeta b_0^2 \EE\bigg[ \frac{W \sqrt{kb_0^2}}{\sqrt{2\pi}} \log\frac{1-Q(W \sqrt{kb_0})}{Q(W \sqrt{kb_0})}e^{-\frac{W^2 kb_0^2}{2}} \bigg] \\
    T_2 &:=  (\zeta b_0^2)^2 \EE\bigg[ \frac{W^2 kb_0^2}{\sqrt{2\pi}} \big(1 + |W|\sqrt{kb_0^2}\big) e^{-\frac{W^2 kb_0^2}{2}\big(1 - \zeta b_0^2\big)^2} \bigg] \\
    T_3 &:=  (\zeta b_0^2)^2 \EE\bigg[ \frac{W^2 kb_0^2}{2\sqrt{2\pi}} |W|^3 (kb_0^2)^{3/2} e^{-\frac{W^2 kb_0^2}{2}\big(1 - \zeta b_0^2\big)^2} \bigg],
\end{align}
and where for $T_2$ and $T_3$ we used the fact that $\delta_0 \in [0,\delta]$ in \eqref{eq:1bit_taylor} to upper bound the corresponding terms by the value at $\delta_0=0$ or $\delta_0 = \delta$.

We will complete the proof by showing that $T_1$ behaves as \eqref{eq:1bit_I1_exact} (with $\sigma^2 = 1$), and that $T_2$ and $T_3$ behave as $o\big( \frac{\sqrt{\log p}}{p} \big)$.  Letting $\phi(\cdot)$ denote the standard normal PDF, we have
\begin{align}
    T_1 &= \frac{\zeta b_0^2}{\sqrt{2\pi}} \int_{-\infty}^{\infty} \phi(w) w \sqrt{kb_0^2} \log\frac{1-Q(w\sqrt{kb_0^2})}{Q(w\sqrt{kb_0^2}} \nonumber \\
        & \hspace*{30ex}\times e^{-\frac{w^2 kb_0^2}{2}} dw \label{eq:1bit_T1_1} \\
        &= \frac{\zeta b_0^2}{\sqrt{2\pi}} \int_{-\infty}^{\infty} \phi\Big( \frac{t}{\sqrt{kb_0^2}} \Big) t \log\frac{1-Q(t)}{Q(t)} e^{-\frac{t^2}{2}} \frac{1}{\sqrt{kb_0^2}} dt \label{eq:1bit_T1_2} \\
        &= \frac{\zeta b_0^2}{\sqrt{2\pi kb_0^2}} \int_{-\infty}^{\infty} \frac{1}{\sqrt{2\pi}} e^{-\frac{t^2}{2}(1 + \frac{1}{kb_0^2})} t \log\frac{1-Q(t)}{Q(t)} dt \label{eq:1bit_T1_3} \\
        &= \frac{\zeta b_0^2}{\sqrt{2\pi kb_0^2}} \frac{1}{\sqrt{1 + \frac{1}{kb_0^2}}} \int_{-\infty}^{\infty} \frac{1}{\sqrt{2\pi(1 + \frac{1}{kb_0^2})^{-1}}} \nonumber \\
            & \hspace*{15ex} \times e^{-\frac{t^2}{2}(1 + \frac{1}{kb_0^2})} t \log\frac{1-Q(t)}{Q(t)} dt \label{eq:1bit_T1_4} \\
        &= \frac{1}{2} \frac{b_0^2}{\sqrt{2\pi kb_0^2}} \EE\Big[ W \log\frac{1-Q(W)}{Q(W)} \Big] (1+o(1)), \label{eq:1bit_T1_5}
\end{align}
where \eqref{eq:1bit_T1_2} follows by a change of variable of the form $t = w\sqrt{k b_0^2}$, \eqref{eq:1bit_T1_3} follows from the definition of $\phi$, and \eqref{eq:1bit_T1_5} follows since $\zeta \to \frac{1}{2}$, and since the integral in \eqref{eq:1bit_T1_4} is the average of $t\log\frac{1-Q(t)}{Q(t)} \openone\{t \ge 0\} $ over an $N(0,(1 + \frac{1}{kb_0^2})^{-1})$ random variable; since $kb_0^2 \to \infty$, this converges to the corresponding average over $W \sim N(0,1)$, which is easily verified to be finite.  

The terms $T_2$ and $T_3$ are handled similarly to $T_1$, so we only briefly comment on the analysis of $T_3$.  By the same arguments as those leading to \eqref{eq:1bit_T1_3}, we obtain
\begin{equation}
    T_3 = \frac{(\zeta b_0^2)^2}{2\sqrt{2\pi kb_0^2}} \int_{-\infty}^{\infty} \frac{1}{\sqrt{2\pi}} e^{-\frac{t^2}{2}\big( (1 - \zeta b_0^2)^2 + \frac{1}{kb_0^2} \big)}  |t|^5 dt.
\end{equation}
The integral is once again $\Theta(1)$, and thus $T_3 = \Theta\Big( \frac{b_0^4}{\sqrt{kb_0^2}} \Big)$, which decays to zero strictly faster than \eqref{eq:1bit_T1_5}.

\subsection{Proof of Proposition \ref{prop:1bit_moments} Part (iv)}

We again assume without loss of generality that $\sigma^2 = 1$.  Defining $\Weq := \sum_{i\in\seq} X_i b_i$ and $\Wdif := \sum_{i\in\sdif} X_i b_i$, it follows from \eqref{eq:1bit_i_dens} that
\begin{equation}
    \imath(X_{\sdif}; Y | X_{\seq}, b_s) = \log \frac{ Q\big( -Y(\Wdif+\Weq)\big) }{ Q\big( -Y \tau \Weq \big)}, \label{eq:1bit_iii_0} 
\end{equation}
where $\tau := \frac{1}{1+\sum_{i\in\sdif}b_i^2}$, and we implicitly condition on $\beta_s = b_s$.  Using \eqref{eq:1bit_p1} and the fact that the variance is upper bounded by the second moment, we have
\begin{align}
    & V_{\sdif,\seq}(b_s) \nonumber \\
    &\le \EE\bigg[ Q\big( -(\Wdif+\Weq) \big) \bigg(\log \frac{ Q\big( -(\Wdif+\Weq)\big) }{ Q\big( -\tau \Weq \big)}\bigg)^2 \nonumber \\
    & \qquad + Q\big( \Wdif+\Weq \big) \bigg(\log \frac{ Q\big( \Wdif+\Weq\big) }{ Q\big( \tau \Weq \big)}\bigg)^2 \bigg] \label{eq:1bit_iii_1} \\
        &= 2 \EE\bigg[ Q\big( \Wdif+\Weq \big) \bigg(\log \frac{ Q\big( \Wdif+\Weq\big) }{ Q\big( \tau \Weq \big)}\bigg)^2 \bigg] \label{eq:1bit_iii_2}, 
\end{align}
where \eqref{eq:1bit_iii_2} follows since the distributions of $\Wdif$ and $\Weq$ are symmetric about zero, and the two are independent.

The function $g(x) := - \log Q(x)$ is convex, and hence it lies above any given tangent vector.  This implies that
\begin{align}
    |g(x_1) - g(x_2)| &\le \max\big\{ |g'(x_1)|, |g'(x_2)| \big\} |x_1 - x_2| \\
        &\le \big( |g'(x_1)| + |g'(x_2)| \big)  |x_1 - x_2|, \label{eq:1bit_iii_4}
\end{align}
where $g'(x) = -\frac{\phi(x)}{Q(x)}$ is the derivative of $g$.  Writing the logarithm of the ratio in \eqref{eq:1bit_iii_2} as a difference of logarithms and applying \eqref{eq:1bit_iii_4}, we obtain
\begin{equation}
    V_{\sdif,\seq}(b_s) \le 2(T_1 + T_2),
\end{equation}
where, overloading the notation from part (iii), we define
\begin{align}
    & T_1 := \iint \fdif(\wdif)\feq(\weq) Q(\wdif+\weq) \nonumber \\
        & \times \bigg( \frac{\phi(\wdif + \weq)}{Q(\wdif + \weq)} \bigg)^2 \big( |\wdif| + (1-\tau)|\weq| \big)^2 d\wdif d\weq \label{eq:1bit_iii_T1} \\
    & T_2 := \iint \fdif(\wdif)\feq(\weq) Q(\wdif+\weq) \nonumber \\
        & \times\bigg( \frac{\phi(\tau\weq)}{Q(\tau\weq)} \bigg)^2 \big( |\wdif| + (1-\tau)|\weq| \big)^2 d\wdif d\weq \label{eq:1bit_iii_T2}
\end{align}
with $\fdif$ and $\feq$ denoting the densities of $\Wdif$ and $\Weq$.  The function $Q(x)\big(\frac{\phi(x)}{Q(x)}\big)^2$ lies between $0$ and $\frac{1}{2}$, and hence $T_1 \le \frac{1}{2}\EE\big[ \big( |\Wdif| + (1-\tau)|\Weq| \big)^2 \big]$, yielding
\begin{equation}
    T_1 = O\big( \EE[\Wdif]^2 + (1-\tau)^2 \EE[\Weq]^2 \big).
    \label{eq:1bit_T1bound}
\end{equation}
We will further simplify this expression below, but we first bound $T_2$, which requires more effort.  

We split the integral over $\RR^2$ in \eqref{eq:1bit_iii_T2} according to whether $|\wdif| \le \frac{1}{2}|\weq|$ or $|\wdif| > \frac{1}{2}|\weq|$; the resulting expressions are denoted by $T_{1,1}$ and $T_{1,2}$ respectively.  In each case, we use the following standard bounds on the Q-function (e.g., see \cite{Fan12}):
\begin{align}
    \frac{\phi(\tau\weq)}{Q(\tau\weq)} 
        &\le \begin{cases} 1+\tau\weq & \weq \ge 0 \\ 1 & \weq < 0 \end{cases} \label{eq:1bit_cases1} \\
    Q(\wdif+\weq) 
        &\le \begin{cases} \frac{1}{2}e^{-\frac{(\wdif + \weq)^2}{2}} & \wdif + \weq \ge 0 \\ 1 & \wdif + \weq < 0. \end{cases} \label{eq:1bit_cases2}
\end{align}

To bound $T_{1,1}$, we note that the condition $|\wdif| \le \frac{1}{2}|\weq|$ implies that $\sign(\wdif + \weq) = \sign(\weq)$, and hence only two of the four combinations of the cases in \eqref{eq:1bit_cases1}--\eqref{eq:1bit_cases2} can occur.  When $\weq < 0$, we can use the second of each of these cases to upper bound the integrand in \eqref{eq:1bit_iii_T2} by $\fdif(\wdif)\feq(\weq) \big( |\wdif| + (1-\tau)|\weq| \big)^2$.  On the other hand, when $\weq \ge 0$ we can use the first of each of the cases to upper bound the integrand by 
\begin{multline}
    \fdif(\wdif)\feq(\weq) \frac{1}{2}e^{-\frac{(\wdif + \weq)^2}{2}} \\ \times \big(1+\tau|\weq|\big)^2 \big( |\wdif| + (1-\tau)|\weq| \big)^2.
\end{multline}
Again using the condition $|\wdif| \le \frac{1}{2}|\weq|$, we find that $e^{-\frac{(\wdif + \weq)^2}{2}} \le e^{-\frac{1}{8}\weq^2}$.  Since $\tau \le 1$ by its definition following \eqref{eq:1bit_iii_0}, it follows that $e^{-\frac{(\wdif + \weq)^2}{2}}  \big(1+\tau\weq\big)^2$ is upper bounded by a universal constant, and we are again left only with $\fdif(\wdif)\feq(\weq) \big( |\wdif| + (1-\tau)|\weq| \big)^2$.  Combining the two cases, we conclude that
\begin{equation}
    T_{2,1} = O\big( \EE[\Wdif^2] + (1-\tau)^2 \EE[\Weq^2] \big).
    \label{eq:1bit_T21bound}
\end{equation}

To upper bound $T_{2,2}$, we upper bound the integrand in \eqref{eq:1bit_iii_T2} by
\begin{align}
    &\fdif(\wdif)\feq(\weq) \big(1+\tau|\weq|\big)^2 \big( |\wdif| + (1-\tau)|\weq| \big)^2 \label{eq:1bit_iii_end1} \\
    &\qquad \le \fdif(\wdif)\feq(\weq) \big(1+2|\wdif|\big)^2 \big( 3|\wdif| \big)^2, \label{eq:1bit_iii_end2}
\end{align}
where \eqref{eq:1bit_iii_end1} follows by taking the higher of the two cases in both \eqref{eq:1bit_cases1} and \eqref{eq:1bit_cases2}, and \eqref{eq:1bit_iii_end2} follows since $|\wdif| > \frac{1}{2}|\weq|$ and $\tau \in [0,1]$.  It follows that
\begin{equation}
    T_{2,2} = O\big( \EE[\Wdif^2] + \EE[\Wdif^4] \big).
    \label{eq:1bit_T22bound}
\end{equation}

We now observe that the first two terms in \eqref{eq:1bit_Vlv} account for all of the terms in \eqref{eq:1bit_T1bound}, \eqref{eq:1bit_T21bound} and \eqref{eq:1bit_T22bound} except for $(1-\tau)^2 \EE[\Weq^2]$.  Recalling that $\tau = \frac{1}{1+\sum_{i\in\sdif}b_i^2}$, we see that $(1-\tau)^2 = \Theta(1)$ whenever $\sum_{i\in\sdif}b_i^2 = \Omega(1)$, whereas a Taylor expansion yields $(1-\tau)^2 = \Theta\big( \big(\sum_{i\in\sdif}b_i^2\big)^2 \big)$ whenever $\sum_{i\in\sdif}b_i^2 = o(1)$.  Combining these cases, we obtain the third term in \eqref{eq:1bit_Vlv}; recall that $\sigma^2=1$ throughout this proof.

\section{Proofs of Auxiliary Results for Noiseless Group Testing} \label{sec:PROOFS_GT}

\subsection{Proof of Proposition \ref{prop:gt_boundI}}

As stated in \cite[Eq.~(36)]{Ati12}, we have $I_{\ell} = \big(1-\frac{\nu}{k}\big)^{k-\ell} H_2\big( \big(1-\frac{\nu}{k}\big)^\ell \big)$, where $H_2(p)$ is the binary entropy function.  For $k\to\infty$ and $\frac{\ell}{k}\to\alpha$, we immediately obtain \eqref{eq:gt_I_const} using the limits $\big(1-\frac{\nu}{k}\big)^{k-\ell} \to e^{-(1-\alpha)\nu}$ and $\big(1-\frac{\nu}{k}\big)^\ell \to e^{-\alpha\nu}$, along with the continuity of the binary entropy function.  In the case that $\frac{\ell}{k} \to 0$, the analogous limits are $\big(1-\frac{\nu}{k}\big)^{k-\ell} \to e^{-\nu}$ and $\big(1-\frac{\nu}{k}\big)^\ell = 1 - \frac{\nu\ell}{k}(1+o(1))$, and we obtain \eqref{eq:gt_I_ord} using the fact that $H_2(1-\epsilon) = (-\epsilon\log\epsilon)(1+o(1))$ as $\epsilon\to0$. Note also that $\log\frac{k}{\nu\ell} = \big(\log\frac{k}{\ell}\big)(1+o(1))$ since $\frac{k}{\ell} \to \infty$.

\subsection{Proof of Proposition \ref{prop:conc_gt}}

We begin by evaluating the information density in \eqref{eq:idens_b}; for brevity, we write $\imath_{\ell} := \imath(X_{\sdif};Y|X_{\seq},b_s)$ and $\imath_{\ell}^n := \imath(\Xv_{\sdif};\Yv|\Xv_{\seq},b_s)$.  Recalling that $P_X \sim \mathrm{Bernoulli}\big(\frac{\nu}{k}\big)$, $\ell = o(k)$, and we are considering the noiseless case, we obtain the following:
\begin{enumerate}
    \item We have $X_{\seq} \ne \bzero$ with probability $1-\big(1-\frac{\nu}{k}\big)^{k-\ell} = (1 - e^{-\nu})(1+o(1))$, and in this case we have $\imath_{\ell} = 0$.
    \item Given $X_{\seq} = \bzero$, we have $X_{\sdif} \ne \bzero$ with probability $1-\big(1-\frac{\nu}{k}\big)^{\ell} = \frac{\nu\ell}{k}(1+o(1))$, and in this case we have $\imath_{\ell} = \log\frac{1}{1-(1-\frac{\nu}{k})^{\ell}} = \big(\log\frac{k}{\ell}\big)(1+o(1))$. 
    \item Given $X_{\seq} = \bzero$, we have $X_{\sdif} = \bzero$ with probability $\big(1-\frac{\nu}{k}\big)^{\ell} = 1+o(1)$, and in this case we have $\imath_{\ell} = \log\frac{1}{(1-\frac{\nu}{k})^{\ell}} = \frac{\nu\ell}{k}(1+o(1))$.
\end{enumerate}
The asymptotic identities given here follow from the assumption $\ell = o(k)$, along with standard Taylor expansions.

Let $N_0$ (respectively, $N_1$) be the random number of measurements such that $X_{\seq} = \bzero$ and $X_{\sdif} = \bzero$ (respectively, $X_{\seq} = \bzero$ and $X_{\sdif} \ne \bzero$). For any $\epsilon_1 \in (0,1)$, the above observations imply the following with probability one when $p$ is sufficiently large:
\begin{align}
    \imath_{\ell}^n &\ge N_1 \Big(\log\frac{k}{\ell}\Big) (1-\epsilon_1) + N_0 \nu\frac{\ell}{k} (1-\epsilon_1) \\
                     &\ge N_1 \Big(\log\frac{k}{\ell}\Big) (1-\epsilon_1).
\end{align}
We also have from \eqref{eq:gt_I_ord} that $I_{\ell} \le \big(e^{-\nu}\nu \frac{\ell}{k} \log\frac{k}{\ell}\big)(1+\epsilon_1)$ for sufficiently large $p$.  Combining these, we conclude that
\begin{equation}
    N_1 > n\frac{1+\epsilon_1}{1-\epsilon_1} e^{-\nu}\nu\frac{\ell}{k} (1-\delta_2) \implies \imath_{\ell}^n > nI_{\ell}(1-\delta_2).
\end{equation} 
By considering the contrapositive statement, we have for any $\epsilon_2 > 0$ and sufficiently large $p$ that
\begin{multline}
    \PP\Big[ \imath^n(\Xv_{\sdif}; \Yv | \Xv_{\seq}, b_s) \le nI_{\ell}(1-\delta_{2}) \Big] \\ \le \PP\Big[ N_1 \le ne^{-\nu}\nu\frac{\ell}{k}(1-\delta_2)(1+\epsilon_2) \Big]. \label{eq:gt_conc4}
\end{multline}
By the observations at the start of this subsection, we have $N_1 \sim \mathrm{Binomial}(n,q)$ with $q = e^{-\nu}\nu\frac{\ell}{k} (1+o(1))$.  We can thus further upper bound the right-hand of \eqref{eq:gt_conc4} by
\begin{equation}
    \PP\big[ N_1 \le nq(1-\delta_2(1-\epsilon_3)) \big]
\end{equation}
for any $\epsilon_3\in(0,1)$ and sufficiently large $p$; here we have used the fact that $(1-\delta_2)(1+o(1)) = (1-\delta_2(1+o(1))$, since $\delta_2$ is fixed.  It follows from a standard Chernoff-based tail bound for Binomial random variables (e.g., see \cite[Sec.~4.1]{Mot10}) that
\begin{align}
    &\PP\Big[ \imath^n(\Xv_{\sdif}; \Yv | \Xv_{\seq}, b_s) \le nI_{\ell}(1-\delta_{2}) \Big] \nonumber \\ 
    &\le e^{-nq \big((1-\delta_2(1-\epsilon_3))\log(1-\delta_2(1-\epsilon_3)) + \delta_2(1-\epsilon_3) \big)}.
\end{align}
The proof is concluded by substituting $q = e^{-\nu}\nu\frac{\ell}{k} (1+o(1))$ and noting that $\epsilon_3$ may be arbitrarily small.

\subsection{Proof of Proposition \ref{prop:gt_psi}}

For the first part, we write $\sum_{\ell=1}^{\lfloor \frac{k}{\log k} \rfloor} {k \choose \ell} \psi_{\ell}(n,\delta_{2}^{(1)}) =: T_1 + T_2$, where $T_1$ sums the terms from $1$ to $\lfloor \log k \rfloor$, and $T_2$ sums the terms from $\lfloor \log k \rfloor +1$ to $\lfloor \frac{k}{\log k} \rfloor$.  For each of these, we upper bound the summation by the number of terms times the maximum term.  

For $T_1$, there are at most $\log k$ terms, and we apply \eqref{eq:gt_psi1}, with $\delta_{2,\ell} = \delta_2^{(1)}$.  The term $(1-\delta_2^{(1)})\log(1-\delta_2^{(1)}) + \delta_2^{(1)}$ can be made arbitrarily close to one by choosing $\delta_2^{(1)}$ to be sufficiently close to one.  Writing $\log{k \choose \ell} = \big(\ell\log\frac{k}{\ell}\big)(1+o(1))$ and performing some simple rearrangements, we obtain the following condition for $T_1 \to 0$:
\begin{equation}
    n \ge \max_{\ell} \frac{k\log \frac{k}{\ell} + \frac{k}{\ell}\log\log k }{ e^{-\nu}\nu } (1+\eta_1), \label{eq:gt_psi_pf1}
\end{equation}
where $\eta_1$ may be arbitrarily small.  Note that $\log\log k$ arises as the logarithm of the number of terms in the summation.  We obtain \eqref{eq:gt_n_cond_psi} by noting that this bound is minimized at $\ell=1$ and writing $k\log k = \big(\frac{\theta}{1-\theta} k \log \frac{p}{k}\big) (1+o(1))$, which follows from $k=\Theta(p^\theta)$.

For $T_2$, a similar argument yields \eqref{eq:gt_psi_pf1} with $\frac{1}{\ell}\log k$ in place of $\frac{1}{\ell}\log\log k$; this follows by upper bounding the number of terms in the summation by $k$.  Since $\ell \ge \log k$, we have $\frac{1}{\ell}\log k = O(1)$, and we conclude that $T_2 \to 0$ provided that \eqref{eq:gt_n_cond_psi} holds.

Finally, for the second part of the proposition, we substitute \eqref{eq:gt_psi2}.  By an analogous argument to that leading to \eqref{eq:gt_psi_pf1}, along with the scaling laws of $I_{\ell}$ in \eqref{eq:gt_I_ord}--\eqref{eq:gt_I_const}, it is readily verified that it suffices that $n = \Omega\big( \max_{\ell} \frac{ \ell\log\frac{k}{\ell} }{ 1 + ( \frac{\ell}{k} \log \frac{k}{\ell} )^2 } \big)$ with a sufficiently large implied constant.  Using the fact that $\ell > \frac{k}{\log k}$ for this part, this reduces to $\Omega\big( \frac{ k\log k }{\log\log k} \big)$.  Thus, any $\Omega(k\log k)$ scaling suffices, and the proof is concluded by noting that $\log k = \Theta\big(\log\frac{p}{k}\big)$.

\section{Noisy Group Testing} \label{sec:PROOFS_GT_NOISY}

Here we provide the relevant details for noisy group testing, leading to Corollary \ref{cor:gt_noisy}.  We focus our attention on the parts that differ from the noiseless case.  Throughout the appendix, we use the notation $q_1 \star q_2 := q_1q_2 + (1-q_1)(1-q_2)$.  We work with an arbitrary Bernoulli distribution $P_X \sim \mathrm{Bernoulli}\big(\frac{\nu}{k}\big)$ to begin, and later substitute $\nu = \log 2$. 

Before proceeding, we analyze the values taken by the information density 
$\imath_{\ell} := \imath(X_{\sdif};Y|X_{\seq},b_s)$ (with $\ell := |\sdif|$) given in \eqref{eq:idens_b}, under the model in \eqref{eq:gtn_model}:
\begin{enumerate}
    \item We have $X_{\seq} \ne \bzero$ with probability $1-\big(1-\frac{\nu}{k}\big)^{k-\ell}$, and in this case we have $\imath_{\ell} = 0$.
    \item Given $X_{\seq} = \bzero$, we have the following, where we define $\xi := \big(1-\frac{\nu}{k}\big)^{\ell}$:
    \begin{itemize}
      \item $X_{\sdif} = \bzero \cap Y = 0$ with probability $(1-\rho)\xi$, yielding  $\imath_{\ell} = \log\frac{1-\rho}{(1-\rho)\xi + \rho(1-\xi)}$;
      \item $X_{\sdif} = \bzero \cap Y = 1$ with probability $\rho\xi$, yielding  $\imath_{\ell} = \log\frac{\rho}{\rho\xi + (1-\rho)(1-\xi)}$;
      \item $X_{\sdif} \ne \bzero \cap Y = 0$ with probability $\rho(1-\xi)$, yielding  $\imath_{\ell} = \log\frac{\rho}{(1-\rho)\xi + \rho(1-\xi)}$;
      \item $X_{\sdif} \ne \bzero \cap Y = 1$ with probability $(1-\rho)(1-\xi)$, yielding  $\imath_{\ell} = \log\frac{1-\rho}{\rho\xi + (1-\rho)(1-\xi)}$.
    \end{itemize} 
\end{enumerate}
In the case that $\ell = o(k)$, we can write $\xi = 1 - \frac{\nu\ell}{k} (1+o(1))$, yielding the following simplifications:
\begin{enumerate}
	\item The preceding four probabilities behave as $(1-\rho)\big(1 - \frac{\nu\ell}{k}(1+o(1))\big)$, $\rho\big(1 - \frac{\nu\ell}{k}(1+o(1)) \big)$, $\rho\frac{\nu\ell}{k}(1+o(1))$, and $(1-\rho)\frac{\nu\ell}{k}(1+o(1))$.  
	\item The corresponding information densities behave as $\frac{1-2\rho}{1-\rho}\frac{\nu\ell}{k}(1+o(1))$, $-\frac{1-2\rho}{\rho}\frac{\nu\ell}{k}(1+o(1))$, $-\log\frac{1-\rho}{\rho}(1+o(1))$ and $\log\frac{1-\rho}{\rho}(1+o(1))$.  For example, the first of these follows by writing $\log\frac{1-\rho}{(1-\rho)(1-\frac{\nu\ell}{k}) + \rho\frac{\nu\ell}{k} } = \log\frac{1 - \rho}{1 - \rho - (1-2\rho)\frac{\nu\ell}{k}}$, dividing the numerator and denominator by $1-\rho$, and Taylor expanding the logarithm. 
\end{enumerate}

\subsection{Analogs of Propositions \ref{prop:gt_boundI}--\ref{prop:gt_psi}}

The analog of Proposition \ref{prop:gt_boundI} is as follows. 

\begin{prop} \label{prop:gtn_boundI}
    Under the noisy group testing setup in Section \ref{sec:GROUP_TESTING}, consider arbitrary sequences of sparsity levels $k\to\infty$ and $\ell \in \{1,\dotsc,k\}$ (both indexed by $p$).  If $\frac{\ell}{k} = o(1)$, then 
    \begin{equation}
         I_{\ell} = \bigg(e^{-\nu}\nu \frac{\ell}{k} (1-2\rho)\log\frac{1-\rho}{\rho}\bigg)(1+o(1)). \label{eq:gtn_I_ord}
    \end{equation}
    Moreover, if $\frac{\ell}{k}\to\alpha \in(0,1]$, then
    \begin{equation}
        I_{\ell} = e^{-(1-\alpha)\nu} \big(H_2\big(e^{-\alpha \nu} \star \rho\big) - H_2(\rho)\big) (1+o(1)). \label{eq:gtn_I_const}
   \end{equation}
\end{prop}
\begin{proof}
	We obtain \eqref{eq:gtn_I_ord} by recalling that the mutual information is the average of the information density, and applying the above-given asymptotic expansions, along with $1-\big(1-\frac{\nu}{k}\big)^{k-\ell} \to e^{-\nu}$.

    To prove \eqref{eq:gtn_I_const}, we write $I(X_{\sdif};Y|X_{\seq}) = H(Y|X_{\seq}) - H(Y|X_{\seq},X_{\sdif})$.  The system model \eqref{eq:gtn_model} immediately gives $H(Y|X_{\seq},X_{\sdif}) = H_2(\rho)$.  Moreover, a direct calculation reveals that $H(Y|X_{\seq}=x_{\seq})$ equals $H_2(\rho)$ if $x_{\seq}$ has an entry equal to one, and $H_2\big(\xi \star \rho\big)$ otherwise, where we again write $\xi := \big(1-\frac{\nu}{k}\big)^{\ell}$.  The proof is concluded by noting that $\xi \to e^{-\alpha\nu}$ when $\frac{\ell}{k} \to \alpha$, and by similarly noting that $\PP[X_{\seq} = \bzero] = \big(1 - \frac{\nu}{k}\big)^{k - \ell} \to e^{-(1-\alpha)\nu}$. 
\end{proof}

As in the noiseless case, we use Proposition \ref{prop:gen_discrete} to characterize $\psi_{\ell}$ for $\ell > \lfloor \frac{k}{\log k} \rfloor$, and $\psi'_{\ell}$ for $\ell = k$.  For $\ell \le \lfloor \frac{k}{\log k} \rfloor$, we instead use the following.

\begin{prop} \label{prop:conc_gtn}
    Under the noisy group testing setup in Section \ref{sec:GROUP_TESTING}, consider sequences $k\to\infty$ and $\ell$, indexed by $p$, such that $\frac{\ell}{k} \to 0$.  For any $\epsilon > 0$ and $\delta_{2} > 0$ not depending on $p$, the following holds for sufficiently large $p$:
    \begin{multline}
    \PP\Big[ \imath^n(\Xv_{\sdif}; \Yv | \Xv_{\seq}, b_s) \le nI_{\ell}(1-\delta_{2}) \Big] \\ \le \exp\bigg(-n\frac{\ell}{k} e^{-\nu}\nu\Big( \frac{\delta_2^2 (1-2\rho)^2}{2(1+\frac{1}{3}\delta_2(1-2\rho))} \Big)(1-\epsilon)\bigg). \label{eq:conc_gt_noisy}
    \end{multline}
    for all $(\sdif,\seq)$ with $|\sdif|=\ell$.
\end{prop}
\begin{proof}
	We make use of the asymptotic identities for $\imath_{\ell}$ at the start of this appendix.  We first note that by simple averaging analogous to that used to obtain \eqref{eq:gtn_I_ord}, we have $v := \EE[\imath_{\ell}^2] = e^{-\nu}\nu \frac{\ell}{k} \big(\log^2\frac{1-\rho}{\rho}\big) (1+o(1))$.  Moreover, we have $\imath_{\ell} \le \big(\log\frac{1-\rho}{\rho}\big)(1+o(1))$ with probability one.  Using the form of Bernstein's inequality based on Bennet's inequality \cite[Sec.~2.7]{Bou13}, we have $\PP[\imath^n \le n(I_{\ell} - \delta)] \exp\big(-n \frac{\delta^2}{2(v + \frac{1}{3}\delta M)} \big)$, where $M$ is any almost-sure upper bound on $\imath_{\ell}$.  Setting $\delta = \delta_2 I_{\ell}$, substituting \eqref{eq:gtn_I_ord} and the preceding expressions for $v$ and $M$, and canceling the common terms in the numerator and denominator, we obtain \eqref{eq:conc_gt_noisy}.
\end{proof}

Letting $\psi_{\ell}$ equal the right-hand side of \eqref{eq:conc_gt_noisy} for $\ell \le \lfloor \frac{k}{\log k} \rfloor$, while being the same as in \eqref{eq:gt_psi2} for $\ell > \lfloor \frac{k}{\log k} \rfloor$, we obtain the following.

\begin{prop} \label{prop:gtn_psi}
    Let $k = \Theta( p^{\theta} )$ for some $\theta \in (0,1)$. 
    
    (i) For any $\eta > 0$ and $\delta_2 \in (0,1)$, there exists a choice of $\epsilon > 0$ in \eqref{eq:gt_psi1} such that $\sum_{\ell=1}^{\lfloor \frac{k}{\log k} \rfloor} {k \choose \ell} \psi_{\ell}(n,\delta_{2}) \to 0$ provided
    \begin{equation}
        n \ge \frac{ 2(1+\frac{1}{3}\delta_2(1-2\rho)) \frac{\theta}{1-\theta}}{ e^{-\nu}\nu \delta_2^2 (1-2\rho)^2 } \, \Big(k\log\frac{p}{k}\Big) (1+\eta). \label{eq:gtn_n_cond_psi}
    \end{equation}
    (ii) For any $\delta_2 \in (0,1)$, we have $\sum_{\lfloor \frac{k}{\log k} \rfloor +1}^k {k \choose \ell} \psi_{\ell}(n,\delta_{2}) \to 0$ provided that $n = \Omega\big(k\log\frac{p}{k}\big)$. 
\end{prop}
\begin{proof}
	The proof is nearly identical to that of Proposition \ref{prop:gt_psi}, except that \eqref{eq:conc_gt_noisy} is used in place of \eqref{eq:conc_gt}, and $\delta_2$ is kept arbitrary in the first part.
\end{proof}

Note that the choices of $\delta_2$ in the two cases above need not coincide; see Remark \ref{rem:delta2}.

\subsection{Remaining Details in the Proof of Corollary \ref{cor:gt_noisy}}

Recall that we have set $\nu = \log 2$.  This yields $e^{-\nu}\nu$ = $\frac{\log 2}{2}$, and thus the first term in \eqref{eq:gt_zeta} follows  from \eqref{eq:gtn_n_cond_psi}.

Next, we consider the condition in \eqref{eq:final_ach} with $\ell = |\sdif| \le \lfloor \frac{k}{\log k} \rfloor$.  Setting $\gamma = 0$, letting $\delta_1 \to 0$ sufficiently slowly, applying Stirling's approximation, and substituting \eqref{eq:gtn_I_ord}, we obtain the condition
\begin{equation}
	n \ge \max_{\ell} \frac{ k \log \frac{p}{\ell} + 2k\log k + 2\frac{k}{\ell}\log k}{ e^{-\nu}\nu (1-2\rho)\log\frac{1-\rho}{\rho} (1-\delta_2) } (1+o(1)). \label{eq:gtn_main_cond0}
\end{equation}
This is maximized for $\ell = 1$, thus yielding the second term in \eqref{eq:gt_zeta} upon writing $k\log k = \frac{\theta}{1-\theta} \big(k\log\frac{p}{k}\big) (1+o(1))$ and $k \log p = \frac{1}{1-\theta} \big(k\log\frac{p}{k}\big) (1+o(1))$ (since $k=\Theta(p^{\theta})$).

Finally, we consider \eqref{eq:final_ach} with $\ell > \lfloor \frac{k}{\log k} \rfloor$.  In this case, the numerator is dominated by the first term, and for the case that $\frac{\ell}{k} \to \alpha\in(0,1]$, we obtain the  condition 
\begin{equation}
	n \ge \frac{ \alpha k \log \frac{p}{k}}{ e^{-(1-\alpha)\nu} \big(H_2(e^{-\alpha \nu} \star \rho) - H_2(\rho)\big) (1-\delta_2) } (1+o(1)), \label{eq:gtn_main_cond1}
\end{equation}
where we have used \eqref{eq:gtn_I_const}. For the case that $\frac{\ell}{k} \to 0$ with $\ell > \lfloor \frac{k}{\log k} \rfloor$, we obtain a condition of the form \eqref{eq:gtn_main_cond0} where only the first term of the numerator is kept.  Such a condition is clearly dominated by \eqref{eq:gtn_main_cond0}. 

Using the result in \cite[Thm.~3a]{Mal78} in the limiting case that the number of defective items grows large, we have for the worst-case choice of $\alpha \in [0,1]$ and an optimized choice of $\nu > 0$ that the minimax threshold resulting from \eqref{eq:gtn_main_cond1} is obtained with $\alpha = 1$ and $\nu = \log 2$.  Substituting these values yields the second term in \eqref{eq:gtn_ach}.


\subsection{An Auxiliary Result for Comparing the Terms}

The following result allows us to compare the terms appearing in the achievability part of Corollary \ref{cor:gt_noisy}.

\begin{prop} \label{prop:gt_cmp}
	For all $\rho \in (0,0.5)$, we have
	\begin{equation}
		(1-2\rho) \log\frac{1-\rho}{\rho} \ge 4\big(\log 2 - H_2(\rho)\big).
	\end{equation}
\end{prop}
\begin{proof}
	By some simple manipulations, the left-hand side can be written as $\log\frac{1}{\rho(1-\rho)} - 2H_2(\rho)$, and we may thus equivalently prove that $\log\frac{1}{\rho(1-\rho)} + 2H_2(\rho) \ge 4\log 2$.  This, in turn, can be verified by showing that the minimum of the function $\log\frac{1}{\rho(1-\rho)} + 2H_2(\rho)$ occurs at $\rho=0.5$, i.e., the point about which it is symmetric.
\end{proof} 

\section*{Acknowledgment}

We gratefully acknowledge Ya-Ping Hsieh for helpful comments and suggestions.

\bibliographystyle{IEEEtran}
\bibliography{../JS_References}

\begin{thebibliography}{10}
\providecommand{\url}[1]{#1}
\csname url@samestyle\endcsname
\providecommand{\newblock}{\relax}
\providecommand{\bibinfo}[2]{#2}
\providecommand{\BIBentrySTDinterwordspacing}{\spaceskip=0pt\relax}
\providecommand{\BIBentryALTinterwordstretchfactor}{4}
\providecommand{\BIBentryALTinterwordspacing}{\spaceskip=\fontdimen2\font plus
\BIBentryALTinterwordstretchfactor\fontdimen3\font minus
  \fontdimen4\font\relax}
\providecommand{\BIBforeignlanguage}[2]{{%
\expandafter\ifx\csname l@#1\endcsname\relax
\typeout{** WARNING: IEEEtran.bst: No hyphenation pattern has been}%
\typeout{** loaded for the language `#1'. Using the pattern for}%
\typeout{** the default language instead.}%
\else
\language=\csname l@#1\endcsname
\fi
#2}}
\providecommand{\BIBdecl}{\relax}
\BIBdecl

\bibitem{Mal13}
M.~B. Malyutov, ``Search for sparse active inputs: A review,'' in \emph{Inf.
  Theory, Comb. and Search Theory}, 2013, pp. 609--647.

\bibitem{Ati12}
G.~Atia and V.~Saligrama, ``Boolean compressed sensing and noisy group
  testing,'' \emph{IEEE Trans. Inf. Theory}, vol.~58, no.~3, pp. 1880--1901,
  March 2012.

\bibitem{Fou13}
S.~Foucart and H.~Rauhut, \emph{A Mathematical Introduction to Compressive
  Sensing}.\hskip 1em plus 0.5em minus 0.4em\relax Springer New York, 2013.

\bibitem{Mil02}
A.~Miller, \emph{Subset Selection in Regression}.\hskip 1em plus 0.5em minus
  0.4em\relax Chapman \& Hall, 2002.

\bibitem{Wai09}
M.~Wainwright, ``Information-theoretic limits on sparsity recovery in the
  high-dimensional and noisy setting,'' \emph{IEEE Trans. Inf. Theory},
  vol.~55, no.~12, pp. 5728--5741, Dec. 2009.

\bibitem{Wai09a}
------, ``Sharp thresholds for high-dimensional and noisy sparsity recovery
  using $\ell_{1}$-constrained quadratic programming ({L}asso),'' \emph{IEEE
  Trans. Inf. Theory}, vol.~55, no.~5, pp. 2183--2202, May 2009.

\bibitem{Wan10}
W.~Wang, M.~Wainwright, and K.~Ramchandran, ``Information-theoretic bounds on
  model selection for {G}aussian {M}arkov random fields,'' in \emph{IEEE Int.
  Symp. Inf. Theory}, 2010.

\bibitem{Rad11}
K.~Rahnama~Rad, ``Nearly sharp sufficient conditions on exact sparsity pattern
  recovery,'' \emph{IEEE Trans. Inf. Theory}, vol.~57, no.~7, pp. 4672--4679,
  July 2011.

\bibitem{Bou08}
P.~Boufounos and R.~Baraniuk, ``1-bit compressive sensing,'' in \emph{Conf. on
  Inf. Sci. and Sys.}, March 2008.

\bibitem{Dor43}
R.~Dorfman, ``The detection of defective members of large populations,''
  \emph{Ann. Math. Stats.}, vol.~14, no.~4, pp. 436--440, 1943.

\bibitem{Mal78}
M.~Malyutov, ``\BIBforeignlanguage{English}{The separating property of random
  matrices},'' \emph{\BIBforeignlanguage{English}{Math. notes Acad. Sci.
  {USSR}}}, vol.~23, no.~1, pp. 84--91, 1978.

\bibitem{Cha11}
C.~L. Chan, P.~H. Che, S.~Jaggi, and V.~Saligrama, ``Non-adaptive probabilistic
  group testing with noisy measurements: Near-optimal bounds with efficient
  algorithms,'' in \emph{Allerton Conf. Comm., Ctrl., Comp.}, Sep. 2011, pp.
  1832--1839.

\bibitem{Fle09}
A.~Fletcher, S.~Rangan, and V.~Goyal, ``Necessary and sufficient conditions for
  sparsity pattern recovery,'' \emph{IEEE Trans. Inf. Theory}, vol.~55, no.~12,
  pp. 5758--5772, Dec. 2009.

\bibitem{Aks13}
C.~Aksoylar, G.~Atia, and V.~Saligrama, ``Sparse signal processing with linear
  and non-linear observations: A unified {S}hannon theoretic approach,'' April
  2013, http://arxiv.org/abs/1304.0682.

\bibitem{Ree12}
G.~Reeves and M.~Gastpar, ``The sampling rate-distortion tradeoff for sparsity
  pattern recovery in compressed sensing,'' \emph{IEEE Trans. Inf. Theory},
  vol.~58, no.~5, pp. 3065--3092, May 2012.

\bibitem{Ree13}
------, ``Approximate sparsity pattern recovery: Information-theoretic lower
  bounds,'' \emph{IEEE Trans. Inf. Theory}, vol.~59, no.~6, pp. 3451--3465,
  June 2013.

\bibitem{Jin11}
Y.~Jin, Y.-H. Kim, and B.~Rao, ``Limits on support recovery of sparse signals
  via multiple-access communication techniques,'' \emph{IEEE Trans. Inf.
  Theory}, vol.~57, no.~12, pp. 7877--7892, Dec 2011.

\bibitem{Aer10}
S.~Aeron, V.~Saligrama, and M.~Zhao, ``Information theoretic bounds for
  compressed sensing,'' \emph{IEEE Trans. Inf. Theory}, vol.~56, no.~10, pp.
  5111--5130, Oct. 2010.

\bibitem{Tan10}
G.~Tang and A.~Nehorai, ``Performance analysis for sparse support recovery,''
  \emph{IEEE Trans. Inf. Theory}, vol.~56, no.~3, pp. 1383--1399, March 2010.

\bibitem{Akc10}
M.~Akcakaya and V.~Tarokh, ``Shannon-theoretic limits on noisy compressive
  sampling,'' \emph{IEEE Trans. Inf. Theory}, vol.~56, no.~1, pp. 492--504,
  Jan. 2010.

\bibitem{Tul13}
A.~Tulino, G.~Caire, S.~Verd\'u, and S.~Shamai, ``Support recovery with
  sparsely sampled free random matrices,'' \emph{IEEE Trans. Inf. Theory},
  vol.~59, no.~7, pp. 4243--4271, July 2013.

\bibitem{Sca13e}
J.~Scarlett, J.~Evans, and S.~Dey, ``Compressed sensing with prior information:
  Information-theoretic limits and practical decoders,'' \emph{IEEE Trans. Sig.
  Proc.}, vol.~61, no.~2, pp. 427--439, Jan. 2013.

\bibitem{Wu10}
Y.~Wu and S.~Verd\'u, ``R\'enyi information dimension: Fundamental limits of
  almost lossless analog compression,'' \emph{IEEE Trans. Inf. Theory}, no.~8,
  pp. 3721--3748, Aug. 2010.

\bibitem{Wu12}
------, ``Optimal phase transitions in compressed sensing,'' \emph{IEEE Trans.
  Inf. Theory}, vol.~58, no.~10, pp. 6241--6263, Oct. 2012.

\bibitem{Tan14}
V.~Tan and G.~Atia, ``Strong impossibility results for sparse signal
  processing,'' \emph{IEEE Sig. Proc. Letters}, vol.~21, no.~3, pp. 260--264,
  March 2014.

\bibitem{Lee15}
J.~D. Lee, Y.~Sun, J.~E. Taylor \emph{et~al.}, ``On model selection consistency
  of regularized $m$-estimators,'' \emph{Elec. J. Stats.}, vol.~9, no.~1, pp.
  608--642, 2015.

\bibitem{Li14}
Y.-H. Li, J.~Scarlett, P.~Ravikumar, and V.~Cevher, ``Sparsistency of
  $\ell_1$-regularized $m$-estimators,'' in \emph{Int. Conf. Art. Intel. Stats.
  (AISTATS)}, 2015, pp. 644--652.

\bibitem{Tan14d}
J.~Tan, D.~Carmon, and D.~Baron, ``Signal estimation with additive error
  metrics in compressed sensing,'' \emph{IEEE Trans. Inf. Theory}, vol.~60,
  no.~1, pp. 150--158, 2014.

\bibitem{Han03}
T.~S. Han, \emph{Information-Spectrum Methods in Information Theory}.\hskip 1em
  plus 0.5em minus 0.4em\relax Springer, 2003.

\bibitem{Ame14}
D.~Amelunxen, M.~Lotz, M.~B. McCoy, and J.~A. Tropp, ``Living on the edge:
  Phase transitions in convex programs with random data,'' \emph{Information
  and Inference}, vol.~3, no.~3, pp. 224--294, 2014.

\bibitem{Bal13}
L.~Baldassini, O.~Johnson, and M.~Aldridge, ``The capacity of adaptive group
  testing,'' in \emph{IEEE Int. Symp. Inf. Theory}, July 2013, pp. 2676--2680.

\bibitem{Ald14}
M.~Aldridge, L.~Baldassini, and K.~Gunderson, ``Almost separable matrices,''
  \emph{J. Comb. Opt.}, pp. 1--22, 2015.

\bibitem{Sca15b}
J.~Scarlett and V.~Cevher, ``Phase transitions in group testing,'' in
  \emph{Proc. ACM-SIAM Symp. Disc. Alg. (SODA)}, 2016.

\bibitem{Fei54}
A.~Feinstein, ``A new basic theorem of information theory,'' \emph{IRE Prof.
  Group. on Inf. Theory}, vol.~4, no.~4, pp. 2--22, Sept. 1954.

\bibitem{Sha57}
C.~E. Shannon, ``Certain results in coding theory for noisy channels,''
  \emph{Information and Control}, vol.~1, no.~1, pp. 6--25, 1957.

\bibitem{Ver94}
S.~Verd\'{u} and T.~S. Han, ``A general formula for channel capacity,''
  \emph{IEEE Trans. Inf. Theory}, vol.~40, no.~4, pp. 1147--1157, July 1994.

\bibitem{Pol10}
Y.~Polyanskiy, V.~Poor, and S.~Verd\'{u}, ``Channel coding rate in the finite
  blocklength regime,'' \emph{IEEE Trans. Inf. Theory}, vol.~56, no.~5, pp.
  2307--2359, May 2010.

\bibitem{Bou13}
S.~Boucheron, G.~Lugosi, and P.~Massart, \emph{Concentration Inequalities: A
  Nonasymptotic Theory of Independence}.\hskip 1em plus 0.5em minus 0.4em\relax
  OUP Oxford, 2013.

\bibitem{Cov01}
T.~M. Cover and J.~A. Thomas, \emph{Elements of Information Theory}.\hskip 1em
  plus 0.5em minus 0.4em\relax John Wiley \& Sons, Inc., 2006.

\bibitem{Van00}
A.~van~der Vaart, \emph{Asymptotic Statistics}.\hskip 1em plus 0.5em minus
  0.4em\relax Cambridge Univ. Press, 2000.

\bibitem{Bar00}
A.~R. Barron, ``Limits of information, {M}arkov chains, and projection,'' in
  \emph{IEEE Int. Symp. Inf. Theory}, 2000.

\bibitem{Laa14}
T.~Laarhoven, ``Asymptotics of fingerprinting and group testing: Tight bounds
  from channel capacities,'' \emph{IEEE Trans. Inf. Forens. Sec.}, vol.~10,
  no.~9, pp. 1967--1980, 2015.

\bibitem{Ald14a}
M.~Aldridge, L.~Baldassini, and O.~Johnson, ``Group testing algorithms: Bounds
  and simulations,'' \emph{IEEE Trans. Inf. Theory}, vol.~60, no.~6, pp.
  3671--3687, June 2014.

\bibitem{Sca16b}
J.~Scarlett and V.~Cevher, ``Converse bounds for noisy group testing with
  arbitrary measurement matrices,'' in \emph{IEEE Int. Symp. Inf. Theory},
  Barcelona, 2016.

\bibitem{Sca16c}
------, ``Limits on sparse support recovery via linear sketching with random
  expander matrices,'' in \emph{Int. Conf. Art. Intel. Stats. (AISTATS)},
  Cadiz, Spain, 2016, pp. 149--158.

\bibitem{Bar10}
R.~Baraniuk, V.~Cevher, M.~Duarte, and C.~Hegde, ``Model-based compressive
  sensing,'' \emph{IEEE Trans. Inf. Theory}, vol.~56, no.~4, pp. 1982--2001,
  April 2010.

\bibitem{Tan14a}
V.~Tan and O.~Kosut, ``On the dispersions of three network information theory
  problems,'' \emph{IEEE Trans. Inf. Theory}, vol.~60, no.~2, pp. 881--903,
  Feb. 2014.

\bibitem{Fan12}
P.~Fan, ``New inequalities of {M}ill's ratio and its application to the inverse
  {Q}-function approximation,'' http://arxiv.org/abs/1212.4899.

\bibitem{Mot10}
R.~Motwani and P.~Raghavan, \emph{Randomized Algorithms}.\hskip 1em plus 0.5em
  minus 0.4em\relax Chapman \& Hall/CRC, 2010.

\end{thebibliography}

\begin{IEEEbiographynophoto}{Jonathan Scarlett}
(S'14 -- M'15) received 
the B.Eng. degree in electrical engineering and the B.Sci. degree in 
computer science from the University of Melbourne, Australia. In 2011, 
he was a research assistant at the Department of Electrical \& Electronic 
Engineering, University of Melbourne.  From October 2011 to August 2014,  
he was a Ph.D. student in the Signal Processing and Communications Group
at the University of Cambridge, United Kingdom. He
is now a post-doctoral researcher with the Laboratory for Information
and Inference Systems at the \'Ecole Polytechnique F\'ed\'erale de Lausanne,
Switzerland.  His research interests are in the areas of information theory, 
signal processing, machine learning, and high-dimensional statistics. 
He received the Cambridge Australia Poynton International Scholarship, and the EPFL Fellows postdoctoral fellowship co-funded by Marie Sk{\l}odowska-Curie.
\end{IEEEbiographynophoto}

\begin{IEEEbiographynophoto}{Volkan Cevher}
(SM'10) received the B.Sc. (valedictorian)
in electrical engineering from Bilkent
University in Ankara, Turkey, in 1999 and the Ph.D.
in electrical and computer engineering from the
Georgia Institute of Technology in Atlanta, GA in
2005. He was a Research Scientist with the University
of Maryland, College Park from 2006-2007 and also
with Rice University in Houston, TX, from 2008-2009.
Currently, he is an Associate Professor at the
Swiss Federal Institute of Technology Lausanne and
a Faculty Fellow in the Electrical and Computer
Engineering Department at Rice University. His research interests include signal
processing theory, machine learning, convex optimization, and information
theory. Dr. Cevher was the recipient of a Best Paper Award at SPARS in 2009,
a Best Paper Award at CAMSAP in 2015, and an ERC StG in 2011.
\end{IEEEbiographynophoto}
 
\end{document}